\documentclass[twoside]{amsbook}
\usepackage{etoolbox} \usepackage[utf8]{inputenc}
\usepackage[T1]{fontenc}
\usepackage{amsmath}
\usepackage{amssymb,empheq,mathtools,dsfont}
\usepackage{physics}
\usepackage{stmaryrd}
\usepackage[mathscr]{eucal}
\usepackage{color}
\usepackage{nicefrac}
\usepackage{thm-restate}
\usepackage{calrsfs}
\usepackage{caption}
\usepackage{subcaption}
\usepackage{tikz}
\usetikzlibrary{patterns}
\usepackage{thmtools} 
\usepackage{thm-restate}
\usepackage[breaklinks,colorlinks=true,allcolors=black]{hyperref}

\usepackage{chngcntr}
\counterwithout*{equation}{section}
\counterwithout*{figure}{section}

\counterwithin{equation}{chapter}
\counterwithin{figure}{chapter}

\usepackage{enumitem}

\newtheorem{exercise}{Exercise}[section]
\newtheorem{problem}{Problem}[section]
\newtheorem{example}{Example}[section]
\newtheorem{theorem}{Theorem}[section]
\newtheorem{corollary}{Corollary}[section]
\newtheorem{remark}{Remark}[section]
\newtheorem{definition}{Definition}[section]
\newtheorem{proposition}{Proposition}[section]
\newtheorem{lemma}{Lemma}[section]

\newtheorem{fact}{Fact}[section]
\newtheorem{assumption}{Assumption}[section]

\newtoggle{amsbook}

\makeatletter
\@ifclassloaded{amsbook}{\toggletrue{amsbook}}{\togglefalse{amsbook}}
\makeatother

\usepackage{scalerel,stackengine}
\stackMath
\newcommand\reallywidehat[1]{\savestack{\tmpbox}{\stretchto{\scaleto{\scalerel*[\widthof{\ensuremath{#1}}]{\kern.1pt\mathchar"0362\kern.1pt}{\rule{0ex}{\textheight}}}{\textheight}}{2.4ex}}\stackon[-6.9pt]{#1}{\tmpbox}}

\newcommand{\F}{\mathbb{F}}

\newcommand{\Fq}{\F_q}

\newcommand{\CC}{\ensuremath{\mathscr{C}}}

\renewcommand{\vec}[1]{\mathbf{#1}}

\newcommand{\cv}{\vec{c}}

\renewcommand{\ev}{\vec{e}}

\newcommand{\xv}{\vec{x}}
\newcommand{\yv}{\vec{y}}

\newcommand{\und}{\mathds{1}}

\newcommand{\dmin}{d_{\textup{min}}}

\makeatletter
\newcommand*{\transp}{{\mathpalette\@transpose{}}}
\newcommand*{\@transpose}[2]{\raisebox{\depth}{$\m@th#1\intercal$}}
\makeatother
\newcommand{\transpose}[1]{{#1}^{\transp}}

\newcommand*{\eqdef}{\stackrel{\text{def}}{=}}

\newcommand{\decP}{\textup{$\mathsf{DP}$}}
\newcommand{\DP}{\decP}

   \renewcommand*\thesection{\arabic{section}}

\newcommand{\pr}{p_{\textup{pr}}}

 \makeatletter
\newcommand\bibalias[2]{\@namedef{bibali@#1}{#2}}

\newtoks\biba@toks
\newcommand\acite[2][]{\biba@toks{\cite#1}\def\biba@comma{}\def\biba@all{}\@for\biba@one:=#2\do{\@ifundefined{bibali@\biba@one}{\edef\biba@all{\biba@all\biba@comma\biba@one}}{\PackageInfo{bibalias}{Replacing citation `\biba@one' with `\@nameuse{bibali@\biba@one}'
      }\edef\biba@all{\biba@all\biba@comma\@nameuse{bibali@\biba@one}}}\def\biba@comma{,}}\edef\biba@tmp{\the\biba@toks{\biba@all}}\biba@tmp
}
\makeatother

\newcommand\dual[1]{#1^{*}} 
\begin{document}

	\title{\bf\Huge Code-based Cryptography:\\[0.5cm] Lecture Notes\\[2cm]}
	\author{\huge Thomas Debris-Alazard\\[0.25cm]
		Inria}
	\date{}
	
	\maketitle
	
	\noindent These lecture notes have been written for courses given at \'Ecole normale sup\'erieure de Lyon and summer school 2022 in post-quantum cryptography that took place in the university of Budapest. Our objective is to give a general introduction to the foundations of code-based cryptography which is currently known to be secure even against quantum adversaries. In particular we focus our attention to the decoding problem whose hardness is at the ground of the security of many cryptographic primitives, the most prominent being McEliece and Alekhnovich' encryption schemes.

	Comments and criticism are very welcome and can be sent to \href{mailto:thomas.debris@inria.fr}{thomas.debris@inria.fr} 
	\newline
	
	{\bf \noindent Acknowledgments.} We would like to thank Maxime Bombar and Alain Couvreur for their helpful comments and corrections.

\let\cleardoublepage\clearpage
	\pagenumbering{gobble}
	\tableofcontents

\pagenumbering{arabic}

	\renewcommand{\thesection}{\thechapter.\arabic{section}}

	\chapter{An Intractable Problem Related to Codes: Decoding}	\label{chapt:1}
	\setcounter{page}{1}

	\section*{Introduction}

	In this course we will consider {\em linear codes}, but what are these mathematical objects known as a linear code? It is a subspace of any $n$-dimensional space over some finite field. Linear codes were initially introduced to preserve the quality of information stored on a physical device or transmitted across a noisy channel. The key principle for achieving such task is extremely simple and natural: {\em adding redundancy}. A trivial illustration is when we try to spell our name over the phone: S like Sophie, T like Terence, E like Emily, ...

	In a digital environment the basic idea to mimic our example is as follows, let $\vec{m}$ be a message of $k$ bits that we would like to transmit over a noisy channel. Let us begin by fixing a linear code $\CC$ (a subspace) of dimension $k$ over $\F_{2}^{n}$ (the words of $n$ bits). By linearity it is easy to map (and to invert) any $k$-bits word to some $n$-bits codeword. The task consisting in adding $n-k$ bits of redundancy is commonly called encoding.  Once our message $\vec{m}$ to transmit is encoded into some codeword $\vec{c}$, we send it across the noisy channel. The receiver will therefore get a corrupted codeword $\vec{c}\oplus\vec{e}$ where some bits of $\vec{c}$ have been flipped. Receiver's challenge lies now in recovering $\vec{c}$, and thus $\vec{m}$, from $\CC$ and $\vec{c}\oplus\vec{e}$, a task called {\em decoding}. The situation is described in the following picture.

		{{
			\begin{center}
\begin{tikzpicture}
					\draw (0,0) circle (0.4);
					\node at (0,0.8) {Sender};
					\draw[thick,->] (0.4,0) -- (2.1,0);
					\node at (1.25,0.3) {$\vec{m}$};
					\draw(2.5,0) circle (0.4);
					\node at (2.5,0.8) {Encoding};
					\draw[thick,->] (2.9,0) -- (7.1,0);
					\node at (3.75,0.3) {$\cv$};
					\draw(5,0) circle (0.4);
					\node at (5,0.8) {Noisy Channel};
					\node at (5,-1) {Error $\vec{e}$};
					\draw[thick,->] (5,-0.8) -- (5,-0.4);
					\draw (5,-0.4) -- (5,0);
					\node at (6.25,0.3) {$\cv\oplus\ev$};
					\draw(7.5,0) circle (0.4);
					\node at (7.5,0.8) {Decoding};
					\draw[thick,->] (7.9,0) -- (9.6,0);
					\node at (8.75,0.3) {$\cv$?};
					\draw(10,0) circle (0.4);
					\node at (10,0.8) { };
					\draw[thick,->] (10.4,0) -- (12.1,0);
					\node at (11.25,0.3) {$\vec{m}$};
					\draw(12.5,0) circle (0.4);
					\node at (12.5,0.8) { };
				\end{tikzpicture}
\end{center}
	}}

	A first, but quite simple, realistic and natural modelization for the noisy channel is the so-called binary symmetric channel: each bit of $\vec{c}$ is independently flipped with some probability $p\in[0,1/2)$. In such a case, given a received word $\vec{y} = (y_{1},\dots,y_{n})$, the probability that $\vec{c} = (c_{1},\dots,c_{n})$ was sent is given by:
	$$
	\mathbb{P}\left( \vec{c} \mbox{ was sent} \mid \vec{y} \mbox{ is received} \right) = p^{d_{\textup{H}}(\vec{c},\vec{y})} (1-p)^{n-d_{\textup{H}}(\vec{c},\vec{y})}
	$$
	where $d_{\textup{H}}(\vec{c},\vec{y}) \eqdef \sharp \left\{ i \in \llbracket 1,n \rrbracket \mbox{ : } c_{i}\neq y_{i} \right\}$ is known as the {\em Hamming distance} between $\vec{c}$ and $\vec{y}$. Using this probability, it is easily verified that any decoding candidate $\vec{c}\in\CC$ is even more likely as it is close to the received message $\vec{y}$ {\em for the Hamming distance.} It explains why ``decoding'' has historically consisted, given an input, to find the closest codeword {\em for the Hamming distance} (usually called maximum likelihood decoding).

	Obviously, the naive procedure enumerating the $\sharp\CC = 2^{k}$ codewords is to avoid.  
	Coding theory aims at proposing family of codes with an explicit and efficient decoding procedure. Until today
two families were roughly proposed: $(i)$ codes derived from strong algebraic structures like Reed-Solomon codes \cite[Chapter 10]{MS86} and their their natural generalization known as algebraic geometry codes \cite{G81}, Goppa codes \cite[Chapter 12]{MS86} or $(ii)$ those equipped with a probabilistic decoding algorithm like convolutional codes
	\cite{E55}, LDPC codes \cite{G63} or more recently polar codes \cite{A09} (which are used in the $5G$). It has been necessary to introduce all these structures because (even after $70$ years of research) decoding a linear code without any ``peculiar'' structure is an intractable problem, topic of these lecture notes.  
	\newline

	{\bf Basic notation.} The notation $x \eqdef y$ means that $x$ is defined to be equal to $y$. Given a finite set $\mathcal{E}$, we will denote by $\sharp\mathcal{E}$ its cardinality. We denote by $\Fq$ the finite field with $q$ elements, $\Fq^{n}$ will denote its $n$-dimensional version for some $n\in\mathbb{N}$. Vectors will be written with bold letters (such as $\vec{e}$) and upper-case bold letters are used to denote matrices (such as $\vec{H}$).  If $\vec{e}$ is a vector in $\Fq^{n}$, then its components are $(e_{1},\dots,e_{n})$. {\em Let us stress that vectors are in row notation}.

	\section{Codes: Basic Definitions and Properties}

	The main objective of this section is to introduce the basic concept of a code and define some of its parameters.

	\begin{definition}[Linear code, length, dimension, rate, codewords] A linear code $\CC$ of length $n$ and dimension $k$ over $\Fq$ $-$ for short, an $[n, k]_{q}$-code $-$ is a subspace of $\Fq^{n}$ of dimension $k$. The rate of $\CC$ is defined as $k/n$ and elements of $\CC$ are called codewords.
	\end{definition}

	Notice that the rate measures the amount of redundancy introduced by the code.

	\begin{remark} A code is more generally defined as a subset of $\Fq^{n}$. However in these lecture notes we will only consider \textup{linear} codes. It may happen that we confuse codes and linear codes but for us a code will always be linear. 
	\end{remark}

	\begin{exercise} \normalfont Give the dimension of the following linear codes:
		\begin{enumerate}
			\item[1.] $\left\{ (f(x_{1}),\dots,f(x_{n})) \mbox{ : } f \in \Fq[X] \mbox{ and } \deg(f) < k  \right\}$ where the $x_{i}$'s are distinct elements of $\Fq$,

			\item[2.] $\left\{ (\vec{u},\vec{u} + \vec{v}) \mbox{ : } \vec{u} \in U \mbox{ and } \vec{v} \in V\right\}$ where $U$ ({\em resp.} $V$) is an $\lbrack n,k_{U} \rbrack_{q}$-code ({\em resp.} $\lbrack n,k_{V} \rbrack_{q}$-code). 
		\end{enumerate}
		
	\end{exercise}

	\begin{definition}[Inner product, dual code] The inner product between $\vec{x},\vec{y} \in \Fq^{n}$ is defined as $\vec{x}\cdot\vec{y} \eqdef \sum_{i=1}^{n} x_{i}y_{i} \in \Fq$. 
	The dual of a code $\CC \subseteq \Fq^{n}$ is defined as $\dual{\CC} \eqdef \left\{\dual{\vec{c}}\in\Fq^{n} \mbox{ : } \forall \vec{c} \in \CC, \mbox{ } \vec{c} \cdot \dual{\vec{c}} = 0  \right\}$.
	\end{definition}

	The dual of an $\lbrack n,k \rbrack_{q}$-code $\CC$ is an $\lbrack n,n-k \rbrack_{q}$-code. Although $\dim(\CC) + \dim(\dual{\CC}) = n$, it may happen that $\CC$ and $\dual{\CC}$ are not in direct sum. This is even more remarkable that coding theorists have given a name to $\CC \cap \dual{\CC}$: the hull of $\CC$. 
	\newline

	{\bf \noindent Representation of a code.} To represent an $\lbrack n, k \rbrack_{q}$-code $\CC$ we may take any basis of it, namely a set of $k$ linearly independent vectors $\vec{g}_{1},\dots,\vec{g}_{k} \in \CC$, and form  the matrix $\vec{G}\in\Fq^{k \times n}$ whose rows are the $\vec{g}_{i}$'s. Then $\CC$ can be written as
	$$
	\CC = \left\{ \vec{m}\vec{G} \mbox{ : } \vec{m} \in \Fq^{k} \right\}.
	$$ 
	Conversely, any matrix $\vec{G}\in \Fq^{k \times n}$ of rank $k$ defines a code with the previous representation. Such matrix $\vec{G}$ is usually called {\em a generator matrix of $\CC$.}

	Another  representation of $\vec{\CC}$ is by a so-called {\em parity-check matrix}. Let $\vec{H}\in\Fq^{(n-k)\times n}$ be such that its rows form a basis of $\dual{\CC}$, the linear code $\CC$ can be written as
	$$
	\CC = \left\{ \vec{c} \in \Fq^{n} \mbox{ : } \vec{H}\transpose{\vec{c}} = \mathbf{0} \right\}.
	$$
	Conversely, any matrix $\vec{H}\in \Fq^{(n-k) \times n}$ of rank $n-k$ defines an $\lbrack n,k \rbrack_{q}$-code with the previous representation. We call $\vec{H}$ a {\em parity-check matrix} of $\CC$.

	Notice now by basic linear algebra that for any non-singular matrix $\vec{S}$ of size $k\times k$ ({\em resp.} $(n-k)\times (n-k)$), $\vec{S}\vec{G}$ ({\em resp.} $\vec{S}\vec{H}$) is still a generator ({\em resp.} parity-check) matrix of $\CC$. Therefore, left multiplication by an invertible matrix ``does not change the code'', it just gives another basis. This is not the case if we perform some right multiplication. For instance if $\vec{P}$ denotes an $n\times n$ permutation matrix, then $\vec{GP}$ will be the generator matrix of the code $\CC$ permuted, namely $\left\{ \vec{c}\vec{P} \mbox{ : } \vec{c}\in \CC \right\}$.

	\begin{exercise}
		Let $\vec{G}\in\Fq^{k\times n}$ be a generator matrix of some code $\CC$. Let $\vec{H}\in\Fq^{(n-k)\times n}$ of rank $n-k$ such that $\vec{G}\transpose{\vec{H}} = \mathbf{0}$. 	
		Show that $\vec{H}$ is a parity-check matrix of $\CC$. 
	\end{exercise}

	Among our two representations of a code, a natural question arises: are we able from one representation to compute the other one? The answer is obviously yes. To see this let $\vec{G}\in\Fq^{k \times n}$ be a generator matrix of $\CC$. As $\vec{G}$ has rank $k$, by a Gaussian elimination we can put it into {\em systematic form}, namely to compute a non-singular matrix $\vec{S} \in \Fq^{k \times k}$ such that $\vec{S}\vec{G} = (\vec{1}_{k}\mid \vec{A})$ (up to a permutation of the columns). The matrix $\vec{S}\vec{G}$ is still a generator matrix of $\CC$. Now it is readily seen that $\vec{H} \eqdef (-\transpose{\vec{A}}\mid \vec{1}_{n-k})$ verifies $\vec{G}\transpose{\vec{H}} = \mathbf{0}$, has rank $n-k$ and therefore is a parity-check matrix of $\CC$.

	Parity-check matrices may seem unnatural to represent a code when comparing to generator matrices. However, even if both representations are equivalent, the parity-check representation is in many applications more relevant (particularly in code-based cryptography). Let us give some illustration. 
	\newline

	{\bf The Hamming code.} Let $\CC_{\textup{Ham}}$ be the binary code of generator matrix:
	$$
	\vec{G} \eqdef \begin{pmatrix}
		1 & 0 & 0 & 0 & 0 & 1 & 1 \\
		0 & 1 & 0 & 0  & 1 & 0 & 1\\
		0 & 0 & 1 & 0 & 1 & 1 & 0 \\
		0 & 0 &  0 & 1 & 1 & 1 & 1 \\
	\end{pmatrix}
	$$
	This code has length $7$ and dimension $4$. The following matrix:
	$$
	\vec{H} \eqdef \begin{pmatrix}
		0 & 0 & 0 & 1 & 1 & 1 & 1 \\
		0 & 1 & 1 & 0 & 0 & 1 & 1 \\
		1 & 0 & 1 & 0 & 1 & 0 & 1 
	\end{pmatrix}
	$$
	has rank $3$ and verifies $\vec{G}\transpose{\vec{H}} = \mathbf{0}$. It is therefore a parity-check matrix of our code $\CC_{\textup{Ham}}$. Notice that $\vec{H}$ has a quite nice structure, its columns are the integers from $1$ to $7$ written in binary.

	Suppose now that we would like to recover $\vec{c} \in \CC_{\textup{Ham}}$ from $\vec{y} \eqdef \vec{c} + \vec{e}_{i}$ where $\vec{e}_{i}$ is the vector of all zeros except a single $1$ at the $i$'th position. It is not clear how to use $\vec{G}$ to recover $\vec{c}$. However note that $\vec{H}\transpose{\vec{y}} = \vec{H}\transpose{\vec{c}} + \vec{H}\transpose{\vec{e}}_{i} = \vec{H}\transpose{\vec{e}}_{i}$ which is the $i$'thm column of $\vec{H}$. Therefore the position $i$ of the error is given by the integer in the binary representation $\vec{H}\transpose{\vec{y}}$.

	The code $\CC_{\textup{Ham}}$ is in fact known as the Hamming code of length $7$. It belongs to the family of Hamming codes which are the $\lbrack 2^{r}-1, 2^{r}-1 - r \rbrack_{2}$-codes built from the $r\times (2^{r}-1)$ parity-check matrix whose $i$-th column is the binary representation of $i$. Hamming codes are the caricatural example of codes that are better understood with their parity-check matrices. Furthermore we can easily correct one error by using their parity-check matrix instead of their generator matrix. 
	\newline

	Our example has shown the relevance of parity-check matrices to understand a code. In particular, we saw how $\vec{H}\transpose{\vec{y}}$ could be a nice source of information when trying to decode $\vec{y}$. This is even more remarkable that we give a name to $\vec{H}\transpose{\vec{y}}$.

	\begin{definition}[Syndrome]
		Let $\vec{H}\in\Fq^{(n-k)\times n}$. The syndrome of $\vec{y}$ with respect to $\vec{H}$ is defined as $\vec{H}\transpose{\vec{y}}$. Any element of $\Fq^{n-k}$ is called a syndrome. 
	\end{definition}

	Syndromes are a natural set of representatives of the code cosets:

	\begin{definition}[Coset] Let $\CC$ be a linear code over $\Fq$ and $\vec{a}\in \Fq^{n}$. The coset of $\vec{a}$ (relatively to $\CC$) is defined as
		$$
		\CC(\vec{a}) \eqdef \vec{a} + \CC. 
		$$
	\end{definition}

	Notice that any parity-check matrix $\vec{H}\in\Fq^{(n-k)\times n}$ of an $\lbrack n,k \rbrack_{q}$-code $\CC$ has rank $n-k$. Therefore for any syndrome $\vec{s}\in\Fq^{n-k}$, there exists some $\vec{a}_{\vec{s}}\in\Fq^{n}$ such that $\vec{H}\transpose{\vec{a}}_{\vec{s}} = \transpose{\vec{s}}$.

	\begin{lemma}\label{lemma:synd}Let $\CC$ be an $\lbrack n,k \rbrack_{q}$-code of parity-check matrix $\vec{H}$. Then for any $\vec{a}, \vec{b} \in \Fq^{n}$,
		$$
		\CC(\vec{a}) = \CC(\vec{b}) \iff \vec{H}\transpose{\vec{a}} = \vec{H}\transpose{\vec{b}}.
		$$ 
	\end{lemma}

	\begin{proof}Notice that,
		\begin{align*}
			\vec{H}\transpose{\vec{a}} = \vec{H}\transpose{\vec{b}} &\iff \vec{H}\transpose{(\vec{a}-\vec{b})} = \mathbf{0} \\
			&\iff \vec{a} - \vec{b} \in \CC 
		\end{align*}
	which concludes the proof. 
	\end{proof} 

 	Cosets play an important role in the geometry of a code. They partition the space $\Fq^{n}$ according to $\CC$: they are the representatives of the ``torus'' $\Fq^{n}/\CC$. Notice now that syndromes are a nice set of representatives of $\F_{q}^{n}/\CC$ via the isomorphism for some parity-check matrix of $\CC$: $\vec{x} \in \Fq^{n}/\CC \mapsto \vec{H}\transpose{\vec{x}} \in \Fq^{n-k}$ (which is well defined and one to one by Lemma \ref{lemma:synd}). In particular we can partition $\F_{q}^{n}$ as follows
 	$$
 	\mathbb{F}_{q}^{n} = \bigsqcup_{\vec{s} \in \mathbb{F}_{q}^{n}}  \left( \vec{a}_{\vec{s}} + \mathcal{C} \right) 
 	$$
 	where $\vec{a}_{\vec{s}}\in \F_{q}^{n}$ is such that $\vec{H}\transpose{\vec{a}}_{\vec{s}} = \transpose{\vec{s}}$.
 	\newline

	{\bf Minimum distance.} Let us define now an important parameter for a code: its minimum distance. It measures the quality of a code in terms of ``decoding capacity'', namely how many errors has to be added before a noisy codeword could be confused with another noisy codeword.

	The minimum distance of a code relies on the definition of Hamming weight.

	\begin{definition}[Hamming weight, distance] The Hamming weight of $\vec{x}\in\Fq^{n}$ is defined as the number of its non-zero coordinates,
		$$
		|\vec{x}| \eqdef \sharp \left\{ i \in \llbracket 1,n \rrbracket \mbox{ : } x_{i} \neq 0 \right\}. 
		$$	
	The Hamming distance between $\vec{x}$ and $\vec{y}$ is defined as $|\vec{x} - \vec{y}|$.
\end{definition}

	\begin{remark}
		Notice that the Hamming metric is a coarse metric which can only take $n+1$ values. Furthermore, it does not distinguish ``small'' and ``large'' coefficients contrary to the Euclidean metric. For instance, in $\mathbb{F}_{11}^{3}$, vectors $(5,3,0)$ and $(1,0,1)$ have the same Hamming weight. 
	\end{remark}

	In what follows $\Fq^{n}$ will always be embedded with the Hamming distance. However one may wonder if other metrics could be interesting for telecommunication or cryptographic purposes. The answer is yes, we can cite the Lee or rank metrics but this is out of the scope of these lecture notes.

	\begin{definition}[Minimum distance] The minimum distance of a linear code $\CC$ is defined as the shortest Hamming weight of non-zero codewords,
		$$
		\dmin(\CC) \eqdef \min \left\{ |\vec{c}| \mbox{ : } \vec{c}\in\CC \backslash\{\mathbf{0}\} \right\}.
		$$
	\end{definition}

		\begin{exercise} \normalfont Give the minimum distance of the following codes:
		\begin{enumerate}
			\item[1.] $\left\{ (f(x_{1}),\dots,f(x_{n})) \mbox{ : } f \in \Fq[X] \mbox{ and } \deg(f) < k  \right\}$ where the $x_{i}$'s are distinct elements of $\Fq$.

			\item[2.] $\left\{ (\vec{u},\vec{u} + \vec{v}) \mbox{ : } \vec{u} \in U \mbox{ and } \vec{v} \in V\right\}$ where $U$ ({\em resp.} $V$) is a code of length $n$ over $\Fq$ and minimum distance $d_{U}$ ({\em resp.} $d_{V}$).

			\item[3.] The Hamming code of length $2^{r}-1$.

			 \small{{\bf Hint:} {\em A code has minimum distance $d$ if and only if for some parity-check matrix $\vec{H}$ every $(d-1)$-tuple of columns are linearly independent and there is at least one linearly linked $d$–tuple of columns.}}

		\end{enumerate}
		
	\end{exercise}

	The following elementary lemma asserts that for a code of minimum distance $d$, if a
	received word has less than $(d-1)/2$ errors (the error has an Hamming weight smaller than $(d-1)/2$), then it can be successfully decoded: the exhaustive search of the closest codeword will output the ``right'' codeword. 
	We stress here that this does not show the existence of an efficient decoding algorithm, which is far from being guaranteed. Furthermore we will see later that for random codes of minimum distance $d$, balls centered at codewords and with radius $\approx d$ typically do not intersect, showing that decoding can theoretically be done for these codes up to distance $\approx d$ and not $ (d-1)/2$.

	\begin{lemma}\label{lemma:d/2} Let $\CC$ be a code of minimum distance $d$, then balls of radius $\frac{d-1}{2}$ centered at codewords are disjoint, 
		$$
		\forall \vec{c},\vec{c}'\in\CC, \mbox{ } \vec{c} \neq \vec{c}', \mbox{ } \mathcal{B}\left(\vec{c},\frac{d-1}{2}\right) \bigcap \mathcal{B}\left( \vec{c}', \frac{d-1}{2} \right) = \emptyset
		$$
		where $\mathcal{B}(\vec{x},r)$ denotes the ball of radius $r$ and center $\vec{x}$ for the Hamming distance.  
	\end{lemma}

	\begin{proof} Let $\vec{c},\vec{c}'\in\CC$ be two distinct codewords. Let us assume that there exists $\vec{x} \in \mathcal{B}\left(\vec{c},\frac{d-1}{2}\right) \cap \mathcal{B}\left( \vec{c}', \frac{d-1}{2} \right)$. By using the triangle inequality we obtain
		\begin{align*}
			|\vec{c} - \vec{c}'| & \leq |\vec{c} - \vec{x}| + |\vec{x} - \vec{c}'| \\
			&\leq \frac{d-1}{2} + \frac{d-1}{2} \\
			&= d-1
		\end{align*}
	which contradicts the fact that $\CC$ has minimum distance $d$. 
	\end{proof}

	\begin{exercise}
		Let $\vec{H}$ be a parity-check matrix of a code $\CC$ of minimum distance $d$. Show that the $\vec{H}\transpose{\vec{e}}$'s are distinct when $|\vec{e}| \leq \frac{d-1}{2}$. 
	\end{exercise}

		\begin{exercise}
		Let $\CC \subseteq \F_{2}^{n}$ be a code of minimum distance $d$ and $t > n - \frac{d}{2}$. Show that there exists at most one codeword $\vec{c}\in\CC$ of weight $t$.  
	\end{exercise}
	
	From this lemma we see that a code with a large minimum distance is a ``good'' code in terms of decoding ability. However there is another parameter to take into account: the rate. A code of small rate asks for adding a lot of redundancy when encoding a message to send, thing that we would like to avoid in telecommunications (where the perfect situation corresponds to not adding any redundancy). Therefore we would like to find a code with large minimum distance and large rate. As it might be expected these two considerations are diametrically opposed to each other. There exists many bounds to quantify the relations between the rate and the minimum distance but this is out of the scope of these lecture notes.

	As we will see in \iftoggle{amsbook}{Chapter \ref{chapt:2}}{lectures notes $2$}, a ``random code'' with a constant rate $k/n\in (0,1)$ has a very good minimum distance, namely $d \sim Cn$ for some constant $C >0$ (known as the relative {\em Gilbert-Varshamov} bound) when $n \to +\infty$. However, while we expect a typical code to have a minimum distance linear in its length given some rate, it is a hard problem to explicitly build linear codes with such minimum distance.

	\section{The Decoding Problem}

	Now that linear codes are defined we are ready to present more formally the decoding problem. Below are presented two equivalent versions of this problem. The first presentation is natural when dealing with noisy codewords (as we did until now) but we will mostly consider in these lecture notes the second one (with syndromes) which is more suitable for cryptographic purposes. For each problem a code is given as input but with a generator or parity-check representation.

	\begin{problem}[Noisy Codeword Decoding]\label{prob:noisyDec}
		Given $\vec{G}\in\Fq^{k \times n}$ of rank $k$, $t\in \llbracket 0,n \rrbracket $, $\vec{y}\in\Fq^{n}$ where $\vec{y} = \vec{c} + \vec{e}$ with $\vec{c} = \vec{m}\vec{G}$ for some $\vec{m}\in\Fq^{k}$ and $|\vec{e}| = t$, find $\vec{e}$. 
	\end{problem} 	
	
	\begin{problem}[Syndrome Decoding]\label{prob:syndDec}
		Given $\vec{H}\in\Fq^{(n-k) \times n}$ of rank $n-k$, $t\in \llbracket 0,n \rrbracket $, $\vec{s}\in\Fq^{n-k}$ where $\vec{H}\transpose{\vec{e}}=\transpose{\vec{s}}$ with $|\vec{e}| = t$, find $\vec{e}$. 
	\end{problem}

	\begin{remark}
		Solving the decoding problem comes down to solve a linear system \textup{but} with some non-linear constraint on the solution (here its Hamming weight). Notice that without such constraint it would be easy to solve the problem with a Gaussian elimination. 
\end{remark}

	It turns out that these two (worst-case) problems are strictly equivalent, if we are able to solve one of them, then we can turn our algorithm into another algorithm that solves the other one in the same running time (up to some small polynomial time overhead).

	Suppose that $(i)$ we have an algorithm solving Problem \ref{prob:noisyDec} and $(ii)$ we would like to solve Problem \ref{prob:syndDec}. To this aim, let $(\vec{H},\vec{s})$ be an input of Problem \ref{prob:syndDec}. First, as $\vec{H}$ has rank $n-k$ we can compute a matrix $\vec{G}$ of rank $k$ such that $\vec{G}\transpose{\vec{H}} = \mathbf{0}$. It is equivalent to computing a generator matrix of the code $\CC$ with parity-check matrix $\vec{H}$. This can be done in polynomial time (over $n$) by performing a Gaussian elimination. Then, by solving a linear system (which also can be done in polynomial time) we can find $\vec{y}$ such that $\vec{H}\transpose{\vec{y}} = \transpose{\vec{s}}$. Notice now that $\vec{H}\transpose{(\vec{y} - \vec{e})} = \mathbf{0}$ where $\vec{H}\transpose{\vec{e}} = \transpose{\vec{s}}$ and $|\vec{e}|=t$. Therefore $\vec{y} = \vec{c} + \vec{e}$ for some $\vec{c} \in \CC$, namely $\vec{c} = \vec{m}\vec{G}$ for some $\vec{m}\in\Fq^{k}$. From this we can use our algorithm solving the noisy codeword version of the decoding problem to recover the error $\vec{e}$. 
	
	\begin{exercise}
Show that any solver of Problem \ref{prob:syndDec} can be turned in polynomial time into an algorithm solving Problem \ref{prob:noisyDec}. 
	\end{exercise}
	
	In what follows we will mainly consider the syndrome version of the decoding problem. Furthermore, we will call {\em decoding algorithm}, any algorithm solving this problem (or its equivalent version with noisy codewords). 
\newline

	{\bf A little bit about the decoding problem hardness.} Our aim in these lecture notes is to show that decoding is {\em hard} 
	\begin{itemize}
		\item in the worst case ($\mathsf{NP}$--complete),

		\item in average (it will be defined in a precise manner later).
	\end{itemize}

	However, even though the decoding problem is hard in the ``worst case'' and in ``average'', let us stress that there are codes that we know how to decode efficiently (hopefully for telecommunications...). It may seem counter-intuitive at first glance: is the decoding problem hard or not? All the subtlety lies in the inputs that are given. Is the code given as input particular? How is the decoding distance $t$ (for instance with $t = 1$ we have an easy problem)? In fact the hardness of the decoding problem relies on how we answer to these questions. There exists some codes and decoding distances for which the problem is easy to solve. The $\mathsf{NP}$--completeness shows that we cannot hope to solve the decoding problem in polynomial time for all inputs while the average hardness ensures (for well chosen $t$) that for almost all code the problem is intractable. 	
	 Our aim in what follows is to show this. But we will first exhibit a family of codes with associate decoding distances $t$ for which decoding is easy. The existence of such codes is at the foundation of code-based cryptography. 
	\newline

	{\bf \noindent Codes that we know to decode: Reed-Solomon codes.} The family of Reed–Solomon codes is of central interest in coding theory: many
	algebraic constructions of codes that we know how to decode efficiently such as BCH codes  \cite[Chapter 3]{MS86}, Goppa codes \cite[Chapter 12]{MS86} derive
	from this family. Reed-Solomon codes are practically used for instance in compact discs, DVD’s, BluRay’s, QR codes etc... \iftoggle{amsbook}{}{They also play an important role in code-based cryptography as we will see in lecture notes $4$.}

	\begin{definition}[Generalized Reed-Solomon Codes] Let $\vec{z}\in(\Fq^{\star})^{n}$ and $\vec{x}$ be an $n$-tuple of pairwise distinct elements of $\Fq$ (in particular $n \leq q$) and let $k \leq n$. The code $\mathsf{GRS}_{k}(\vec{x},\vec{z})$ is defined as
		$$
		\mathsf{GRS}_{k}(\vec{x},\vec{z}) \eqdef \left\{ (z_{1} f(x_{1}),\dots,z_{n} f(x_{n})) \mbox{ } : \mbox{ } f \in\Fq[X] \mbox{ and } \deg(f) < k \right\}.
		$$
	\end{definition}

	Generalized Reed–Solomon codes $\mathsf{GRS}_{k}(\vec{x},\vec{z})$ are $\lbrack n,k \rbrack_{q}$-codes with many remarkable properties. Among others, they are said ``MDS'', {\em i.e.} 
	their minimum distance equals $d \eqdef n-k+1$ and there exists an efficient decoding algorithm to correct any pattern of $\lfloor \frac{d-1}{2} \rfloor$ errors as we will see below. However, a major drawback of Generalized Reed-Solomon codes is that their length is upper-bounded
	by the size of the alphabet $\Fq$.

	\begin{exercise}
		Show that $\dual{\mathsf{GRS}_{k}(\vec{x},\vec{z})} = \mathsf{GRS}_{n-k}(\vec{x},\vec{z}')$ where $z_{i}' = \frac{1}{z_{i}\prod_{j\neq i}(x_{i}-x_{j})}$. Deduce that $\mathsf{GRS}_{k}(\vec{x},\vec{z})$ has a parity-check matrix of the following form: 
	\begin{equation}\label{eq:parRSG}
	\vec{H} \eqdef \begin{pmatrix}
		1 & 1 & \cdots & 1 \\
		x_1 & x_2 & \cdots & x_n \\
		x_1^{2} & x_2^{2} & \cdots & x_n^{2} \\ 
		\hdots & \hdots & \hdots & \hdots \\
		x_1^{n-k-1} & x_2^{n-k-1} & \cdots & x_{n}^{n-k-1}  
	\end{pmatrix}
	\begin{pmatrix}
		z_1' &  &  & 0 \\
		& z_2' & & \\
		&  &  \ddots & \\
		0  &  &    & z_n'   
	\end{pmatrix}
\end{equation} 
Furthermore, show that $\mathsf{GRS}_{k}(\vec{x},\vec{z})$ has minimum distance $n-k+1$. 
	\end{exercise}

	{\bf Decoding Generalized Reed-Solomon codes at distance $\leq \lfloor \frac{n-k}{2} \rfloor$.} Suppose that $\mathsf{GRS}_{k}(\vec{x},\vec{z})$ is given as input, namely that the $x_{i}$'s and $z_{i}$'s are known (or equivalently $\vec{H}$ given in Equation \eqref{eq:parRSG}). Let $\vec{y}$ be a noisy codeword that we would like to decode: 
	$$
	\vec{y} \eqdef \vec{c}+\vec{e}
	$$ 
	where $\vec{c} = (z_{1}f(x_1),\cdots, z_{n}f(x_n)) \in \mathsf{GRS}_{k}(\vec{x},\vec{z})$ with $f \in \Fq\lbrack X \rbrack$ such that $\deg(f) < k$ and $\ev\in\Fq^{n}$ be an error of Hamming weight $t \leq \lfloor \frac{n-k}{2} \rfloor$. Let us suppose without loss of generality that $z_1 = \cdots = z_n = 1$ (for the general case multiply each coordinate of $\vec{y}$ by $z_i^{-1}$ which does not change the weight of the error term).

	Our aim is to recover $f$ (or equivalently $\vec{e}$). Let us introduce the following (unknown) polynomial:
	$$
	E(X) \eqdef \prod_{i : e_i \neq 0}(X-x_i).
	$$
	Notice that $\deg(E) = t$. The key ingredient of the decoding algorithm is the following fact
	\begin{fact}
	\begin{equation} \label{eq:RSsyst0}
		\forall i \in \llbracket 1,n \rrbracket, \quad  y_iE(x_i) = f(x_i)E(x_i).
	\end{equation} 
		\end{fact}
	Coordinates $y_i$ and $x_i$ are known while $f$ and $E$ are unknown. System \eqref{eq:RSsyst0} is not linear and the basic idea to decode is to {\em linearize} it (to bring us to a pleasant case). Let,
	$$
	N \eqdef Ef
	$$
	Equation \eqref{eq:RSsyst0} can be rewritten as: 
	\begin{equation}\label{eq:RSsyst}
	\forall i \in \llbracket 1,n \rrbracket, \quad y_iE(x_i) = N(x_i)
	\end{equation} 
	where coefficients of the polynomial $N \in \Fq\lbrack X \rbrack$ of degree $< k + t$ and $E\in\Fq\lbrack X \rbrack$ 
	of degree $t$ are unknowns. 
This system has a non-trivial solution: $(E,Ef)$ but it may have many other solutions. The following lemma asserts that any other non-trivial solution enables to recover $f$.

	\begin{lemma} Let $E_1,E_{2} \in \Fq[X]$ of degree $\leq \lfloor \frac{n-k}{2} \rfloor$ and $N_{1},N_{2}\in\Fq[X]$ of degree $< k + \lceil \frac{n-k}{2} \rceil$ such that $(E_{1},N_{1})$ and $(E_{2},N_{2})$ are non-zero and solutions of Equation \eqref{eq:RSsyst}. Then,
		$$
		\frac{N_{1}}{E_{1}} = \frac{N_{2}}{E_{2}} = f. 
		$$
	\end{lemma} 
		
	\begin{proof} First, if $E_{i} = 0$ then $N_{i}$ has $n$ roots by Equation \eqref{eq:RSsyst} while its degree is smaller than $n$. Therefore we get that $E_{i}\neq 0$ as $(E_{i},N_{i})$ is non-zero. 
		Let $R \eqdef N_{1}E_{2} - N_{2}E_{1}$, we have
		$$
		\deg(R) < k + \left\lfloor \frac{n-k}{2} \right\rfloor + \left\lceil \frac{n-k}{2} \right\rceil \leq  n. 
		$$
		By using now that $(E_{1},N_{1})$ and $(E_{2},N_{2})$ are solutions of Equation \eqref{eq:RSsyst} we obtain for all $i \in \llbracket 1,n \rrbracket$, 
		\begin{align*}
			R(x_{i}) &= N_{1}(x_{i})E_{2}(x_{i}) - N_{2}(x_{i})E_{1}(x_{i}) \\
			&= y_{i}E_{1}(x_{i})E_{2}(x_{i}) - y_{i}E_{2}(x_{i})E_{1}(x_{i}) \\
			&= 0
		\end{align*}
	Therefore $R$ has $n$ roots while its degree is smaller than $n$. It shows that $R = 0$ and $\frac{N_{1}}{E_{1}} = \frac{N_{2}}{E_{2}}$. It concludes the proof as $(E,Ef)$ is also a non-zero solution of Equation \eqref{eq:RSsyst}. 
	\end{proof}

	The algorithm we just described to decode a generalized Reed-Solomon code up to the distance $\lfloor \frac{n-k}{2} \rfloor$ is known as the {\em Berlekamp-Welch} algorithm.

		\subsection{Worst Case Hardness}

	The aim of this subsection is to show that the decoding problem is hard in the {\em worst case}, namely $\mathsf{NP}$--complete. We have to be careful with this kind of statement. First, $\mathsf{NP}$--completeness is designed for decisional problems: ``is there a solution given some input?''. Furthermore, we need to be very cautious with inputs that are being fed to our problem.
The $\mathsf{NP}$--completeness ``only'' shows (under the assumption $\mathsf{P}\neq \mathsf{NP}$) that we cannot hope to have an algorithm solving our problem in polynomial time {\em for all inputs}. The set of possible inputs is therefore important, it may happen that a problem is easy to solve when its inputs are drawn from a set $A$ while it becomes hard ($\mathsf{NP}$--complete) when its inputs are taken from some set $B \supsetneq A$. This remark has to be carefully taken into consideration when using the $\mathsf{NP}$--completeness in cryptography as a safety guaranty. It is quite possible that the security of a cryptosystem relies on the difficulty to solve an $\mathsf{NP}$--complete problem but at the same time breaking the scheme amounts to solve the problem on a subset of inputs for which it is easy. To summarize, the $\mathsf{NP}$--completeness of a problem for a cryptographic use is a nice property but it is not the panacea to ensure its hardness.

	The foregoing discussion has shown that we have to rephrase the decoding problem as a decisional problem. Furthermore, it will be important to have a careful look on the set of inputs.  
	
	\begin{problem}[Decisional Decoding Problem]\label{prob:SyndDecDec}\mbox{ }
		
		\begin{itemize}
			\item \textup{Input:} $\vec{H}\in \F_{q}^{(n-k)\times n}$, $\vec{s}\in\F_{q}^{n-k}$ where $n,k \in \mathbb{N}$ with $k \leq n$ and an integer $t \leq n$.

			\item \textup{Decision:} it exists $\vec{e}\in\F_{q}^{n}$ of Hamming weight $t$ such $\vec{H}\transpose{\vec{e}} = \transpose{\vec{s}}$. 
		\end{itemize}
	\end{problem}

	\begin{proposition}[\cite{BMT78}]\label{propo:NPDec} 
		Problem \ref{prob:SyndDecDec} for $q = 2$ is $\mathsf{NP}$--complete. 
	\end{proposition}

	The proof of this proposition relies on a reduction of the following combinatorial decision problem, which is known to be $\mathsf{NP}$--complete. 
	
	\begin{problem}[Three Dimensional Matching (3DM)]\label{prob:TDM}\mbox{ }		
		\begin{itemize}
			
			\item \textup{Input:} a subset $U \subseteq T \times T \times T$ where $T$ is a finite set.
			
			\item \textup{Decision:} it exists $V \subseteq U$ such that $\sharp V = \sharp T$ and for all $(x_1,y_1,z_1),(x_2,y_2,z_2) \in V$ we have $x_1 \neq x_2,y_1\neq y_2$ et $z_1 \neq z_2$. 
		\end{itemize}
	\end{problem}

	The formalism of this problem may seem at first sight to be far away from the decoding problem. However we can restate it with incidence matrices. We proceed as follows: first we take each first coordinate of elements that belong to $U$, then we build an incidence matrix relatively to $T$ of size $\sharp T \times \sharp U$ and similarly for the two remaining coordinates. After that we vertically concatenate our three matrices. Therefore we get in polynomial time a matrix of size $3\sharp T \times \sharp U$, that we will call a 3DM{\em-incidence matrix}. But now, as shown in the following lemma, we have a solution to the 3DM-problem associated to $U$ and $T$ if and only if there are $\sharp T$ columns that sum up to the all one vector (which corresponds to our decoding problem). But let us first give an example to illustrate this discussion.

		\begin{example} Let $T = \{1,2,3\}$ and $U = \{ u_{1},u_{2},u_{3},u_{4},u_{5} \}$ such that:
		$$u_{1} = (1,1,2), \quad u_{2} = (2,3,1), \quad u_{3} = (1,2,3)  $$
		$$ u_{4} = (3,1,2) \quad \mbox{and} \quad u_{5} = (2,2,2). $$
		The \textup{3DM}-incidence matrix associated to these sets is given by: 
		
		\begin{center} 	
			\small 
			\begin{tabular}{ c | c c c c c | }	
				
				& 112 & 231 & 123 & 312 & 222 \\
				\hline 
				1 & 1 & 0 & 1 & 0 & 0  \\
				2 & 0 & 1 & 0 & 0 & 1 \\
				3 & 0 & 0 & 0 & 1 & 0 \\ 
				&   &   &   &   &   \\ 
				1 & 1 & 0 & 0 & 1 & 0  \\
				2 & 0 & 0 & 1 & 0 & 1 \\	
				3 & 0 & 1 & 0 & 0 & 0 \\ 
				&   &   &   &   &   \\ 
				1 & 0 & 1 & 0 & 0 & 0  \\
				2 & 1 & 0 & 0 & 1 & 1 \\	
				3 & 0 & 0 & 1 & 0 & 0 \\ 
				\hline  
			\end{tabular}
		\end{center} 
		We obtain the all one vector by summing columns $2,3$ and $4$. Therefore, $V = \{ u_{2},u_{3},u_{4} \}$ is a solution.
	\end{example}

	\begin{lemma}\label{lemm:MTD} Let $T$ and $U \subseteq T\times T \times T$ be an instance of \textup{3DM} and let $\vec{H}_{\textup{3DM}} \in \F_2^{3 \; \sharp T\times \sharp U}$ be the associated incidence matrix. We have
		$$
		\mbox{There is a solution for the instance } T,U \iff \exists \ev \in \F_2^{\sharp U} \mbox{ } : \mbox{ } |\ev| = \sharp T \quad \mbox{and} \quad \vec{H}_{\textup{3DM}}\transpose{\vec{e}} = \transpose{\mathbf{1}} \mbox{(all one vector)}. 
		$$
		
	\end{lemma}

	\begin{proof} By definition, columns of $\vec{H}_{\textup{3DM}}$ have length $3\sharp T$ and Hamming weight $3$. Therefore, $\sharp T$ columns sum up to the all one vector if and only if their supports are pairwise distinct. 
	\end{proof}

	We are now ready to prove Proposition \ref{propo:NPDec}.

	\begin{proof}[Proof of Proposition \ref{propo:NPDec}]
	Let $T,U$ be an instance of the three dimensional matching problem. We can build in polynomial time the matrix $\vec{H}_{\textup{3DM}}$. Now, by Lemma \ref{lemm:MTD}, there is a solution for $T$ and $U$ if and only if there is a solution of the decoding problem for the input $(\vec{H}_{\textup{3DM}}, \vec{1})$ and $t \eqdef \sharp T$. 
	\end{proof}

	We have just proven that decoding is an $\mathsf{NP}$--complete problem but when is given as input a binary matrix and a decoding distance. In other words, we cannot reasonably hope to find a polynomial time algorithm to solve the decoding problem for all codes over $\F_{2}$ and for all decoding distances. But can we find a proof that fits with a restricted set of inputs? The answer is yes. Below is presented an incomplete list of some improvements. The decoding problem is still $\mathsf{NP}$--complete if we restrict:

	\begin{itemize}
		\item the decoding distance at $t = n/\log_{2}n$ \cite{F09} or $t = C n$ for any constant $C\in(0,1)$ \cite{D19}.

		\item the input codes are restricted to Reed-Solomon codes \cite{GV05}

		\item etc... 
	\end{itemize}

	There are many other $\mathsf{NP}$--complete problems related to codes. For instance, computing the minimum distance of a code \cite{V97} or some codewords of weight $w$ \cite{BMT78} are $\mathsf{NP}$--complete.

		\subsection{Average Case Hardness}

		The decoding worst-case hardness makes it a suitable problem for cryptographic applications. However we have to be careful when dealing with the decoding problem in this context. Recall that the aim of any cryptosystem is to base its security on the ``hardness'' of solving some problem. However to study and to ensure the hardness (thus the security) it would be preferable first to define the problem {\em exactly} as it is stated when wanting to break the crypto-system. It leads us to the following question: ``how the decoding problem is used in cryptography?''. To answer this question let us briefly present the McEliece public key encryption scheme \cite{M78} 
that was introduced just few months after RSA. This scheme will motivate our definition of the ``cryptographic'' decoding in Problem \ref{prob:decoGenR}.
		\newline

		{\bf \noindent McEliece encryption scheme.} McEliece's idea to build a public key encryption scheme based on codes is as follows: Alice, the secret key owner, has a code $\CC$ that she can efficiently decode up to some distance $t$ (some ``quantity''  that enables to decode is the secret). Alice publicly reveals a parity-check matrix of her code, let us say $\vec{H}$, as well as its associated decoding distance $t$. For obvious security reasons Alice does not want $\vec{H}$ to reveal any information on how she decodes $\CC$. In that case, the perfect situation corresponds to a matrix $\vec{H}$ {\em which is uniformly distributed}. Now Bob wants to send a message $\vec{m}$ to Alice. First he associates with a public one to one mapping (in a sense to define) his message $\vec{m}$ to some vector $\vec{e}$ of Hamming weight $t$. Then he computes $\vec{H}\transpose{\vec{e}}$ and sends it to Alice. Once again, for obvious security reasons, Bob does not want $\vec{e}$ to share any information with $\vec{m}$ that could be used when observing $\vec{H}\transpose{\vec{e}}$. The perfect situation corresponds to a mapping such that $\vec{e}$ {\em is uniformly distributed over words of Hamming weight $t$}. Now Alice who got $\vec{H}\transpose{\vec{e}}$ recovers $\vec{e}$ and $\vec{m}$ thanks to her decoding algorithm.

		One may wonder why Bob has associated its message to some word of weight $t$ and not $\leq t$ as Alice can decode up to the distance $t$. The reason is that any malicious person looking at the discussion between Alice and Bob observes $\vec{H}\transpose{\vec{e}}$ and to recover the message she/he has to find $\vec{e}$. However, decoding is harder if $|\vec{e}|$ is larger. Therefore it is preferable if $\vec{e}$ has a Hamming weight as large as possible, thus $t$.

		\begin{remark}
			McEliece encryption scheme relies on the use of generator matrices. We have actually presented Niederreiter encryption scheme \cite{N86}. The security of both schemes is the same. The only differences are in term of efficiency, depending of the context. 
		\end{remark}

		\begin{exercise}
			Describe how the encryption scheme works with generator matrices. 
		\end{exercise}

	We are now ready to define the (average) decoding problem for cryptographic applications. In what follows $q$ will denote a fixed field size while $R(n)$ and  $\tau(n)$ will be functions taking their values in $(0,1)$. To simplify notation, since $n$ is clear here from the context, we will drop the dependency in $n$ and simply write $R$ and $\tau$.

				\begin{problem}[\textup{Decoding Problem - $\mathsf{DP}(n,q,R,\tau)$}]\label{prob:decoGenR} 
				Let $k \eqdef \lfloor Rn \rfloor$ and $t \eqdef \lfloor \tau n \rfloor$.

				\begin{itemize}		
					\item \textup{Input:}  $(\vec{H},\vec{s}\eqdef \vec{x}\transpose{\vec{H}})$ where $\vec{H}$ (resp. $\vec{x}$) is uniformly distributed over $\Fq^{(n-k)\times n}$ (resp. words of Hamming weight $t$ in $\Fq^{n}$).

					\item \textup{Output:} an error $\ev\in \Fq^{n}$ of Hamming weight $t$ such that $\ev\transpose{\vec{H}} = \vec{s}$. 
				\end{itemize}
			\end{problem}

		\begin{remark}
			This problem really corresponds to decode a code of rate $R$ and parity-check matrix $\vec{H}$. We call such a code a random code as its parity-check matrix is uniformly distributed (for more details see \iftoggle{amsbook}{Chapter \ref{chapt:2}}{lecture notes $2$}). 
		\end{remark}

		\begin{remark}
			In our definition of $\mathsf{DP}$, we ask given a code and a syndrome obtained via a vector $\vec{x}$ of weight $t$, to find a vector $\vec{e}$ with the same weight that reaches the syndrome. In particular, we do not ask to recover $\vec{x}$. It may seem confusing when looking at the original definition of decoding problem in telecommunications where it is requested to recover exactly $\vec{x}$ and thus the message that was sent. But such definition imposes some constraints over $t$, for instance $t$ smaller than the minimum distance of the code out of $2$, which ensures the uniqueness of the solution (see Lemma \ref{lemma:d/2}). However, in cryptography our constraints are not the same. Sometimes we ask $\mathsf{DP}$ to have a unique solution given some instance (like in encryption schemes), sometimes not (like in signatures). When thinking about the decoding problem in cryptography we have to forget the ``telecommunication context''. For now, our 
			concern is the hardness of $\mathsf{DP}$, whatever is the choice of $t$, whatever is the number of solutions. 
			We will further discuss this (important) remark in \iftoggle{amsbook}{Chapter \ref{chapt:2}}{lecture notes $2$} 
. As we will see, all the subtlety lies in the choice of $t$. 
		\end{remark}

		We could have defined $\mathsf{DP}$ without any distribution on its inputs. However we are interested in the algorithmic hardness of this problem in the following way. Let us assume that we have a probabilistic algorithm $\mathcal{A}$ that solves (sometimes) the decoding problem at distance $t$. Furthermore, let us suppose that a single run of this algorithm costs a time $T$. Inputs of $\mathcal{A}$ are a parity-check matrix $\vec{H}$ and a syndrome $\vec{s}$. We denote by $\vec{w} \in \{0,1\}^{\ell}$ the internal coins of $\mathcal{A}$ which tries to output some $\vec{e}$ of weight $t$ that reaches the syndrome  $\vec{s}$ with respect to $\vec{H}$. We are interested in its probability of success: 		
		$$
		\varepsilon = \mathbb{P}_{\vec{H},\vec{x},\vec{w}}\left(\mathcal{A}(\vec{H},\vec{s} = \vec{x}\transpose{\vec{H}},\vec{w}) = \vec{e}  \mbox{  s.t  }  |\vec{e}| = t \mbox{  and  }  \vec{e}\transpose{\vec{H}} = \vec{s}  \right)
		$$
		{\em where} the probability is computed over the internal coins of $\mathcal{A}$ and $\vec{H}$ ({\em resp.} $\vec{x}$) being uniformly distributed over $\Fq^{(n-k)\times n}$ ({\em resp.} words of Hamming weight $t$ in $\Fq^{n}$). 		
		This leads us to say that  $\mathcal{A}$ solves the decoding problem in {\em average} time 
		$$
		T/\varepsilon.
		$$ 
		In \iftoggle{amsbook}{Chapter \ref{chapt:3}}{lecture notes $3$} we will study algorithms solving this problem and in each case their complexity will be written as some $T/\varepsilon$. 
		
		\begin{remark} We have spoken of ``average time complexity'', it comes from the fact that $\varepsilon$ is the average success probability of $\mathcal{A}$ over all its possible inputs. By using the law of total probability it can be verified that:
\begin{align*} 
		\varepsilon &= \frac{1}{q^{(n-k)n} \; (q-1)^{t}\binom{n}{t}} \; \sum_{\substack{\vec{H}\in \Fq^{(n-k)\times n} \\ |\vec{x}| = t}} \mathbb{P}_{\vec{w}}\left(\mathcal{A}(\vec{H},\vec{s} = \vec{x}\transpose{\vec{H}},\vec{w}) = \vec{e}  \mbox{  s.t  }  |\vec{e}| = t \mbox{  and  }  \vec{e}\transpose{\vec{H}} = \vec{s}  \right)\\
		&= \mathbb{E}_{\vec{H},\vec{x}} \left( \mathbb{P}_{\vec{w}}\left(\mathcal{A}(\vec{H},\vec{s} = \vec{x}\transpose{\vec{H}},\vec{w}) = \vec{e}  \mbox{  s.t  }  |\vec{e}| = t \mbox{  and  }  \vec{e}\transpose{\vec{H}} = \vec{s}  \right) \right)
		\end{align*} 
		In particular we are interested in the probability to solve the decoding problem in average over all $\lbrack n,k \rbrack_{q}$-codes.
	\end{remark}

	$\mathsf{DP}$ is a problem parametrized by $n$ and two functions of $n$: $R$ and $\tau$. In the overwhelming majority of cryptographic applications the rate $R\in(0,1)$ is chosen as a constant. But it may be also interesting to consider the case where $R \mathop{\longrightarrow}\limits_{n\to+\infty} 0$. Actually this regime of parameters is basically the $\mathsf{LPN}$ problem that will be discussed at the end of this subsection. 
	Considering now the other parameter $\tau$, that we will call the {\em relative decoding distance}, many choices can be made but this greatly varies the difficulty $\mathsf{DP}$. For instance, when $\tau = O(\log n/n)$, there is at most a polynomial number of errors of weight $\tau n$ and 
	 a simple enumeration is enough to solve $\mathsf{DP}$ in polynomial time (over $n$). But surprisingly there are also many other non-trivial regimes of parameters for which $\mathsf{DP}$ can be solved in polynomial time. We will see in \iftoggle{amsbook}{Chapter \ref{chapt:3}}{lecture notes 3} that $\mathsf{DP}(n,q,R,\tau)$ can be solved in polynomial time as soon as $\tau \in \left[ (1-R) \; \frac{q-1}{q}, R + (1-R)\; \frac{q-1}{q} \right]$. However, despite many efforts, the best algorithms to solve $\mathsf{DP}$ (even after $70$ years of research) are all exponential in $\tau n$ for other relative distances $\tau$, namely $T/\varepsilon \mathop{=}\limits_{n \to+\infty} 2^{n\tau\; (\alpha(q,R,\tau) + o(1))}$ for some function $\alpha(q,R,\tau)$ which depends of the used algorithm $\mathcal{A}$, $q$, $R$ and $\tau$. Situation is depicted in Figure \ref{fig:hardEasyDP}.

	 	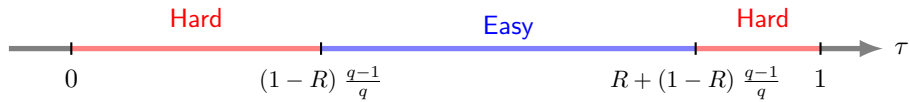
\begin{figure}[htb]
	 	\centering
	 	\begin{tikzpicture}[scale=0.83]
	 		\tikzstyle{valign}=[text height=1.5ex,text depth=.25ex]
	 		\draw[line width=2pt,gray] (0,2) -- (1,2);
	 		\draw (3,2.5) node[red]{{\sf Hard}};
	 		\draw (11.5,2.5) node[right,red]{{\sf Hard}};
\draw[line width=2pt,red!50] (1,2) -- (5,2);
	 		\draw[line width=2pt,blue!50] (5,2) --
	 		node[above,midway,blue,valign]{{\sf Easy}} (11,2);
	 		\draw[line width=2pt,red!50] (11,2) -- (13,2);
	 		\draw[->,>=latex,line width=2pt,gray] (13,2) -- (14,2)
	 		node[right,black] {$\displaystyle \tau$};
	 		\tikzstyle{valign}=[text height=2ex]
	 		\draw[thick] (1,1.9) node[below,valign]{$0$} -- (1,2.1);
	 		\draw[thick] (5,1.9) node[below,valign]{\scalebox{0.9}{$(1-R)\;\frac{q-1}{q}$}} -- (5,2.1);
	 		\draw[thick] (11,1.9) node[below,valign]{\scalebox{0.9}{$R+(1-R)\;\frac{q-1}{q}$}~~} -- (11,2.1);
	 		\draw[thick] (13,1.9) node[below,valign]{$1$} -- (13,2.1);
\end{tikzpicture}
	 	\caption{Hardness of $\mathsf{DP}(n,q,R,\tau)$ as function of $\tau$.\label{fig:hardEasyDP}}
	 \end{figure}

		$\mathsf{DP}(n,q,R,\tau)$ is hard in average but for well chosen relative distances $\tau$. Therefore anyone who wants to design a crypto-system whose security relies on the hardness of solving $\mathsf{DP}$ has to carefully choose $\tau$ (in most cases the choice is constrained by the design itself). We list below some choices that have been made according to the designed (asymmetric) primitive:
		
		\begin{itemize}
			\item McEliece encryption \cite{M78}: $\tau = \Theta\left( \frac{1}{\log n} \right)$,

			\item Encryption schemes \cite{A03,MTSB13,HQC17}: $\tau = \Theta\left( \frac{1}{\sqrt{n}} \right)$,

			\item Authentication protocol \cite{S93}: $\tau = C$ for some constant $C$ quite small,

			\item Signature \cite{DST19a}: $\tau = C$ for some constant $C$ large, $C \approx 0.95$.
			\mbox{ }
			\vspace*{0.25cm}
		\end{itemize}

		{\bf \noindent The Learning Parity with Noise Problem.} In the cryptographic literature, a problem closely related to $\DP$ and referred to as Learning Parity with Noise  ($\mathsf{LPN}$) is sometimes considered. 
It is a problem where is given as input an oracle that is function of some secret quantity. The aim is then to recover this secret but with as many samples as wanted (outputs of the oracle).

			\begin{definition}[$\mathsf{LPN}$-oracle]\label{def:LPNO} Let $k\in\mathbb{N}$, $\tau\in[0,1/2)$ and $\vec{s}\in\F_{2}^{k}$. We define the $\mathsf{LPN}(k,\tau)$-oracle $\mathcal{O}^{\mathsf{LPN}}_{\vec{s},\tau}$ as follows: on a call it outputs $(\vec{a},\langle \vec{s},\vec{a}\rangle + e)$ where $\vec{a}\leftarrow\F_{2}^{k}$ is uniformly distributed and $e$ being distributed according to a Bernoulli of parameter $\tau$. 
		\end{definition}

		\begin{problem}[\textup{Learning with Parity Noise Problem  - $\mathsf{LPN}(k,\tau)$}] \label{prob:LPN} 
			\mbox{ }

			\begin{itemize}		
				\item \textup{Input:} $\mathcal{O}^{\mathsf{LPN}}_{\vec{s},\tau}$ be an $\mathsf{LPN}(k,\tau)$-oracle parametrized by $\vec{s}\in\F_{2}^{k}$ which has been chosen uniformly at random.

				\item \textup{Output:} $\vec{s}$.
			\end{itemize}
		\end{problem}

		Let us stress that anyone who wants to solve this problem can ask as many samples  (outputs of $\mathcal{O}^{\mathsf{LPN}}_{\vec{s},\tau}$) as he wants. However, each call to the oracle costs one. All the game consists in finding efficient algorithms that solves $\mathsf{LPN}(k,\tau)$ with as few queries as possible. 	
		Notice that the difficulty greatly varies with $\tau$. The noise parameter $\tau$ deeply affects the gain of information on $\vec{s}$ that we obtain with each sample.

		When $\tau = 0$, it is necessary to make at least $k$ queries and then to solve a square linear system which has a complexity roughly given by $k^{3}$. On the other hand, when $\tau\in (0,1)$ is some constant, best algorithms \cite{BKW03} have a sub-exponential time complexity $2^{O(k/\log_{2} k)}$ and for them the number of queries is roughly the running time. 
		\newline

		{\noindent \bf $\mathsf{LPN}$: a special case of $\mathsf{DP}$.} It turns out that solving $\mathsf{LPN}(k,\tau)$ with $n$ samples basically corresponds to solving $\mathsf{DP}(n,2,R,\tau)$ where $R = k/n$. Therefore, as the number of samples $n$ is a priori unlimited, $\mathsf{LPN}$ really amounts to solve $\mathsf{DP}$ where the rate can be chosen arbitrarily close to $0$.

		Suppose that an algorithm asks for $n$ samples to solve $\mathsf{LPN}(k,\tau)$, here $\mathcal{O}^{\mathsf{LPN}}_{\vec{s},\tau}$ outputs the sequence:
		$$
		\vec{s}\cdot\vec{a}_{1} + e_1, \mbox{ }\dots, \mbox{ } \vec{s} \cdot \vec{a}_{n} + e_{n}.
		$$
		These $n$ samples can be rewritten as
		$
		\vec{s}\vec{G}+\vec{e} 
		$
		where columns of $\vec{G} \in \F_{2}^{k \times n}$ are the $\vec{a}_{i}$'s and $\vec{e} \eqdef (e_{1},\dots,e_{n})$. Now notice that $\mathbb{E}(|\vec{e}|) = \tau n$ as each $e_{i}$ is a Bernoulli distribution of parameter $\tau$. The algorithm that recovers $\vec{s}$ and thus $\vec{e}$ decodes at distance $|\vec{e}|$ the code of generator matrix $\vec{G}$. It corresponds to solve $\mathsf{DP}\left(n,q,\frac{k}{n},\frac{|\vec{e}|}{n}\right)$ where is given as input a matrix $\vec{H}\in\Fq^{(n-k)\times n}$ such that $\vec{G}\transpose{\vec{H}} = \mathbf{0}$ and the syndrome $\vec{e}\transpose{\vec{H}}$.

		\begin{remark}
			\DP could have been presented directly with generator matrix representation. For more details see \iftoggle{amsbook}{Chapter \ref{chapt:2}}{lecture notes $2$}, in particular Exercise \iftoggle{amsbook}{\ref{exo:GenversusPar}}{$1$}. 
		\end{remark}
		
		\subsection{Search to Decision Reduction}

		It is common in cryptography to consider for a same problem two variants: search or decision/distinguish. Roughly speaking, for some one-way function $f$ (easy to compute but hard to invert) we ask in the search version given $f(x)$ to recover $x$ while in the decision version we ask to distinguish between $f(x)$ and a uniform string. Obviously, the decision version is easier and therefore to rely a cryptosystem security on the hardness of some decision problem instead of its search counterpart is a strongest assumption to make. It turns out that Diffie-Hellman \cite{DH76} or El Gamal \cite{E84} primitives rely on this kind of assumption. But as we will see in this subsection, constructions based on codes do not suffer from this flaw, it has been shown in \cite{FS96}, through a reduction that the decision and search versions of the decoding problem are equivalent. The interesting direction has been to show that if there is an algorithm solving the decision version, then there is an algorithm that solves (in essentially the same time) the search part. We call such a result {\em a search-to-decision reduction.}

		However it may be tempting to say that obtaining a search-to-decision reduction for the decoding problem is ``only interesting'' but not crucial for any security guarantee. This is not true and to see this let us present Alekhnovich scheme \cite{A03}, which is after McEliece scheme the second way of building encryption schemes based on codes and the decoding problem\iftoggle{amsbook}{.}{ (fore more details see lecture notes $4$).} 
		\newline

		{\bf \noindent Alekhnovich encryption scheme.} 	By contrast with McEliece's idea, Alekhnovich did not seek to build a public key encryption scheme based on the use of a decoding algorithm as a secret key. He proposed to start from a code $\CC$ of length $n$ for which we do not necessarily have an efficient decoding algorithm. The public key in Alekhnovich scheme is defined as
		$
		(\CC,\vec{c}+\vec{e})$  where $\vec{c} \in \CC$ and  $|\vec{e}| \ll n
		$
		while the secret key is $\vec{e}$. Now if someone wants to encrypt {\em some bit} $b \in \{0,1\}$ into $\mathsf{Enc}(b)$ he proceeds as follows: 
		\begin{itemize}
			\item $\mathsf{Enc}(1) \eqdef \vec{u}$ where $\vec{u}$ is a uniform vector,

			\item $\mathsf{Enc}(0) \eqdef \dual{\vec{c}}+\vec{e}'$ where $|\vec{e}'| \ll n$ and $\dual{\vec{c}}$ belongs to the dual of the code spanned by $\CC$ and $\vec{c}+\vec{e}$. 
		\end{itemize}

		Now to decrypt we just compute the inner product $\mathsf{Enc}(b) \cdot \vec{e}$. The correction of this procedure relies on the fact that
		$$
		 \vec{e}\cdot \mathsf{Enc}(0) =  \vec{e} \cdot \left( \dual{\vec{c}}+\vec{e}'\right) =   \vec{e}\cdot \vec{e}',
		$$
		where in the last equality we used that $\vec{e}$ belongs to the code spanned by $\CC$ and $\vec{c}+\vec{e}$ while $\dual{\vec{c}}$ is in its dual.  But now, with a high probability, $\vec{e}\cdot \vec{e}'= 0$ as both vectors have a very small hamming weight ($\ll n$). On the other hand, $\vec{e} \cdot \mathsf{Enc}(1) = \vec{e}\cdot \vec{u}$ will be a uniform bit. Therefore, to securely send a bit $b$, it is enough to repeat this procedure a small amount of times and to choose the most likely input according to the most probable outcome.

		Notice now that a natural strategy for an adversary to decrypt is to distinguish between $\mathsf{Enc}(1)$ and $\mathsf{Enc}(0)$, namely a uniform string and a noisy codeword. Therefore the security of Alekhnovich scheme critically relies on the decision/distinguish version of the decoding problem. 
		\newline

		Our aim now is to show how to obtain a search-to-decision reduction for the decoding problem. However to explain how to get this result, let us come back to the viewpoint with one-way functions. Let $\mathcal{A}$ be an algorithm that can distinguish between a random string $u$ and $f(x)$ be some one-way function $f$. Given $f(x_{0})$ we would like to use $\mathcal{A}$ to glean some information about $x_{0}$. A natural idea is to disturb a little bit $f(x_{0})$ and to feed $\mathcal{A}$ with it with the hope, when looking at its answer, to gain some information on $x_{0}$ after repeating a small amount of times the operation. Here the key is the meaning of ``information'', we have to be careful about this. For instance, does it make sense to have a direct proposition like: if given $\mathcal{A}$ and $f(x_{0})$ we can deduce a bit of $x_{0}$ then we are able to invert $f$? In fact not. Given $f$, the following function $g(b,x) = (b,f(x))$ is also a one-way function but its first input bit is always revealed. In other words, to hope to be able to invert $f$, we need to obtain another information than obtaining directly an input bit from $\mathcal{A}$, but which information? A first answer to this question has been given in  \cite[Proposition 2.5.4]{G01_a}. Roughly speaking, it has been proven that if someone can extract from $f(x)$ and a uniform string $r$ the value $x \cdot r$, then one can invert $f$.

		\begin{proposition}[\cite{GL89,G01_a}]\label{propo:Goldreich} Let $f : \{0,1\}^{*} \rightarrow \{0,1\}^{*}$, $\mathcal{A}$ be a probabilistic algorithm running in time $T(n)$ and $\varepsilon(n) \in ( 0,1 )$ be such that
			$$	
			\mathbb{P}\left( \mathcal{A}( f(\xv_n),\vec{r}_{n})  = \vec{x}_{n} \cdot \vec{r}_{n} \right)  = \frac{1}{2} + \varepsilon(n)	
			$$	
			where the probability is computed over the internal coins of $\mathcal{A}$, $\vec{x}_{n}$ and $\vec{r}_{n}$ that are uniformly distributed over $\{ 0,1 \}^{n}$. 		
			Let $\ell(n) \eqdef \log(1/\varepsilon(n))$. Then, it exists an algorithm $\mathcal{A}'$ running in time $O\left(n^{2}\ell(n)^{3}T(n)\right)$ that satisfies
			$$
			\mathbb{P}\left( \mathcal{A}'(f(\xv_n)=\xv_n) \right) = \Omega\left(\varepsilon(n)^{2}\right)
			$$
			where the probability is computed over the internal coins of $\mathcal{A}'$ and $\vec{x}_n$.			
		\end{proposition}

		\begin{proof} A nice proof of this proposition can be found here \url{https://www.math.u-bordeaux.fr/~gzemor/alekhnovich.pdf}
		\end{proof}

		\begin{remark} Interestingly, the proof of this proposition relies on the use of linear codes (and their associated decoding algorithm) that are some Reed-Muller like codes of order one \cite[Chapitre 13]{MS86}. 
		\end{remark}

		This proposition will be at the core of the search-to-decision reduction of the decoding problem. Let us start by the formal definition of the decision decoding problem.

			\begin{problem}[Decision Decoding Problem - $\mathsf{DDP}(n,q,R,\tau)$]\label{prob:DDP} Let $k \eqdef \lfloor R n \rfloor$ and $t \eqdef \lfloor \tau n \rfloor$. 
			\begin{itemize}
				\item \textup{Distributions:}
				\begin{itemize}
					\item $\mathcal{D}_{0}$ : $(\vec{H},\vec{s})$ be uniformly distributed over $\Fq^{(n-k)\times n} \times \Fq^{n-k}$,

					\item $\mathcal{D}_{1}$ : $(\vec{H},\vec{x}\transpose{\vec{H}})$ where $\vec{H}$ (resp. $\vec{x}$) being uniformly distributed over $\Fq^{(n-k)\times n}$ (resp. words of Hamming weight $t$). 
				\end{itemize}

				\item \textup{Input:} $(\vec{H},\vec{s})$ distributed according to $\mathcal{D}_b$ where $b \in \{0,1\}$ is uniform,

				\item \textup{Decision:} $b' \in \{0,1\}$.
			\end{itemize}

		\end{problem}

		A first, but trivial, way to solve $\mathsf{DDP}$ would be to output a random bit $b'$. It would give the right solution with probability $1/2$ which is not very interesting. The efficiency of an algorithm solving this problem is measured by the difference between its probability of success and $1/2$. This quantity is the right one to consider and is defined as the {\em advantage}. 
	
			\begin{definition} The $\mathsf{DDP}(n,q,R,\tau)$-advantage of an algorithm $\mathcal{A}$ is defined as:
			\begin{equation}\label{eq:advA} 
				Adv^{\mathsf{DDP}(n,q,R,\tau)}(\mathcal{A}) \eqdef \frac{1}{2}\left( \mathbb{P}\left( \mathcal{A}(\vec{H},\vec{s}) = 1 \mid b=1  \right) - \mathbb{P}\left( \mathcal{A}(\vec{H},\vec{s}) = 1 \mid b = 0 \right)  \right) 
			\end{equation} 
			where the probabilities are computed over the internal randomness of $\mathcal{A}$, a uniform $b\in \{0,1\}$ and inputs be distributed according to $\mathcal{D}_{b}$ which is defined in $\mathsf{DDP}(n,q,R,\tau)$ (Problem \ref{prob:DDP}).			
			
\end{definition} 	
		
		For the sake of simplicity we will omit the dependence in the parameters $(n,q,R,\tau)$.

		\begin{exercise}
			Prove that when $(\vec{H},\vec{s})$ is distributed according to $\mathcal{D}_{b}$ (for a fixed $b \in \{0,1\}$) we have:
			$$
			\mathbb{P}\left(\mathcal{A}(\vec{H},\vec{s}) = b\right) = \frac{1}{2} + Adv^{\mathsf{DDP}}(\mathcal{A}).
			$$ 
		\end{exercise}

		Our aim now is to prove the following theorem which shows how an algorithm solving $\mathsf{DDP}$ can be turned into an algorithm solving $\mathsf{DP}$. More precisely, we will show how to turn $\mathcal{A}$ with advantage $Adv^{\mathsf{DDP}}(\mathcal{A})$ into an algorithm that computes $\vec{x}\cdot \vec{r}$ with probability $1/2 + Adv^{\mathsf{DDP}}(\mathcal{A})$ given as input $\vec{x}\transpose{\vec{H}}$ and $\vec{r}$. To conclude it will simply remain to apply Proposition \ref{propo:Goldreich}.

		\begin{theorem}\label{theo:StDdeco} 
Let $\mathcal{A}$ be a probabilistic algorithm running in time  $T(n)$ whose $\mathsf{DDP}(n,2,R,\tau)$-advantage is given by $\varepsilon(n)$ and let $\ell(n) \eqdef \log(1/\varepsilon(n))$. 
		Then it exists an algorithm $\mathcal{A}'$ that solves $\mathsf{DDP}(n,2,R,\tau)$ in time $O(n^{2}\ell(n)^{3})T(n)$ and with probability $\Omega(\varepsilon(n)^{2})$.  
		\end{theorem}

		\begin{remark}
			Theorem \ref{theo:StDdeco} is stated for binary codes. However it can be extended to $q$-ary codes by using a generalization of Proposition \ref{propo:Goldreich} proved in \cite{GRS00}. 
		\end{remark}

		\begin{proof}[Proof of Theorem \ref{theo:StDdeco}] Let $(\vec{H},\vec{s} \eqdef \vec{x}\transpose{\vec{H}})$ be an instance of $\mathsf{DP}(n,q,R,\tau)$. In what follows, $\mathcal{A}'$ is an algorithm such that on input $(\vec{H},\vec{x}\transpose{\vec{H}},\vec{r})$ it outputs $\vec{x}\cdot\vec{r}$ with probability $1/2+\varepsilon$. To end the proof it will be enough to apply Proposition \ref{propo:Goldreich}.

		\vspace*{0.1cm}	
			{\bf Algorithm} $\mathcal{A}'$ :

			\quad Input : $(\vec{H},\vec{s}) \in \F_2^{(n-k)\times n} \times \F_2^{n-k}$ and $\vec{r} \in \F_2^{n}$,

			\qquad 1. $\vec{u} \in \F_2^{n-k}$ be uniformly distributed

			\qquad 2. $\vec{H}' \eqdef \vec{H} - \transpose{\vec{u}}\vec{r}$

			\qquad 3. $b \eqdef \mathcal{A}(\vec{H}',\vec{s})$

			\quad Output : $b$
			\vspace*{0.1cm}

			The matrix $\vec{H}$ is uniformly distributed by definition, therefore $\vec{H}'$ is also uniformly distributed. Notice now,
			\begin{equation*} 
				\vec{s} = \vec{x}\transpose{\vec{H}} = \vec{x}\transpose{\vec{H}'} + \left( \vec{x}\cdot\vec{r} \right)\vec{u}.
			\end{equation*}
			Let,
			\begin{equation*}
				\vec{s}' \eqdef \vec{x}\transpose{\vec{H}'} + \vec{u}.
			\end{equation*}
			It is readily verified that $\vec{s}'$ is uniformly distributed. Therefore, according to $ b = \vec{x}\cdot\vec{r}\in\{0,1\}$, we obtain distributions of $\mathsf{DDP}$. The probability that $\mathcal{A}'$ outputs $\vec{x}\cdot \vec{r}$ is given by: 
			\begin{align}
				\mathbb{P}( \mathcal{A}'( \vec{H},\vec{x}\transpose{\vec{H}},\vec{r}) =  \vec{x}\cdot\vec{r}   ) 
				&=  \frac{1}{2}\mathbb{P}\left( \mathcal{A}'(\vec{H},\vec{x}\transpose{\vec{H}},\vec{r}) = 0 \mid  \vec{r}\cdot\vec{x}  = 0  \right) + \frac{1}{2}\mathbb{P}\left( \mathcal{A}'(\vec{H},\vec{x}\transpose{\vec{H}},\vec{r}) = 1 \mid \vec{r}\cdot\vec{x}  = 1  \right) \label{eq:red1} \\
				&= \frac{1}{2}\left(  \mathbb{P}\left( \mathcal{A}(\vec{H}',\vec{x}\transpose{\vec{H}'})=0\right) + \mathbb{P}\left( \mathcal{A}(\vec{H}',\vec{s}') = 1 \right) \right) \nonumber \\
				&= \frac{1}{2} + \frac{1}{2}\left( \mathbb{P}\left( \mathcal{A}(\vec{H}',\vec{s}') = 1 \right) - \mathbb{P}\left(\mathcal{A}(\vec{H}',\vec{x}\transpose{\vec{H}'})=1 \right) \right)\nonumber \\
				&= \frac{1}{2} + \varepsilon \label{eq:red2}  
			\end{align} 
		where we used in \eqref{eq:red1} the fact that $\vec{r}$ is uniformly distributed and in \eqref{eq:red2} the $\mathsf{DDP}$-advantage definition. 	
		\end{proof}

 	\chapter{Random Codes}\label{chapt:2}

	\section*{Introduction}

	We study in this course {\em random codes}, {\em i.e.} codes whose parity-check or generator matrix is drawn uniformly at random. However, in light of the history of error correcting codes which has consisted in finding codes becoming more and more complex and structured, one may wonder but why then are we wasting our time to study random codes? That may come as a surprise but random codes enlighten about what we could expect or not in a simple fashion, and even better what is optimal or not. The most prominent example of the interest of random codes is the famous Shannon theorem about the capacity of some ``noisy channels''. Roughly speaking, Shannon gave (for some error models) the maximum amount of errors that ``can'' be theoretically decoded with codes of fixed rate. Shannon made his proof by using random codes and he has shown that they are precisely those which reach the optimality.

	Our aim in these lecture notes is to study carefully these kind of codes and to show that they enable to answer many questions like for instance: 
	\begin{itemize}
		\item How many vectors of Hamming weight $t$ do we expect in a code?

		\item What is the typical minimum distance of a code?

		\item etc...
	\end{itemize}

	Our study will have an important consequence for cryptographic purposes: a better understanding of the Decoding Problem $\mathsf{DP}(n,q,R,\tau)$ that was defined \iftoggle{amsbook}{in Chapter \ref{chapt:1}}{in lecture notes $1$}. We will be able to predict with a very good accuracy the number of solutions of this problem as a function of its parameters. This will be particularly helpful to understand the behaviour of algorithms solving it.

	\section{Prerequisites}

	{\bf \noindent Basic notation.} In all these lecture notes, $q$ will denote a fixed field size while $R$ will be a {\em constant} in $(0,1)$. On the other hand, $\tau(n)$ will denote any function of $n$ taking its values in $(0,1)$.  To simplify notation, since $n$ is clear from the context, we will drop the dependency in $n$ and simply write $\tau$. Furthermore, parameters $k$ and $t$ will always (even implicitly) be defined as $k \eqdef \lfloor Rn \rfloor$ and $t \eqdef \lfloor \tau n \rfloor$. A function $f(n)$ is said to be negligible, and we denote this by $f \in \textup{negl}(n)$, if for all polynomial $p(n)$, $|f(n)| < |p(n)|^{-1}$ for all sufficiently large $n$.

	Many asymptotic results will be given. As all our parameters are functions of $n$, our asymptotic results will always hold for:
$$
	n \longrightarrow +\infty.
	$$
	The parameter $n$ is in most cryptographic applications roughly given by several thousands.

	The following function $h_{q}$, known as the {\em $q-$ary entropy}, will play an important role:
	$$
	h_{q} : x \in [0,1] \longmapsto -x \log_{q}\left(\frac{x}{q-1}\right) - (1-x)\log_{q}(1-x) \quad \mbox{(extended by continuity in $0$ and $1$)}. 
	$$
	It is equal to the entropy of a random variable $e$ over $\Fq$ distributed like the error
	for a $q-$ary symmetric channel of crossover probability $x$, {\em i.e.} $\mathbb{P}(e = 0) = 1 - x$ and $\mathbb{P}(e=\alpha) = \frac{x}{q-1}$ for any $\alpha\in\Fq^{\star}$.

	It can be verified that $h_{q}$ is an increasing function over $\left[0,\frac{q-1}{q}\right]$ and a decreasing function over $\left[ \frac{q-1}{q},1\right]$. The $q-$ary entropy is involved in the estimation of $\sharp\mathcal{S}_{t}$ where $\mathcal{S}_{t}$ is defined as the sphere of radius $t$ for the Hamming distance $|\cdot|$, namely
	$$
	\mathcal{S}_{t} \eqdef \left\{ \vec{x} \in \Fq^{n} \mbox{ } : \mbox{ } |\vec{x}| =t \right\}. 
	$$
	The following elementary lemma will be at the core of most of our asymptotic results. 
	\begin{lemma}\label{lemma:asymptSphere}
		Let $t\eqdef \lfloor \tau n \rfloor$. We have $\sharp \mathcal{S}_{t} = \binom{n}{t}(q-1)^{t}$ and
		$$
		q^{n\left( h_{q}(\tau) + O\left(\frac{\log_{q}(n)}{n}\right)\right)} \leq \binom{n}{t}(q-1)^{t} \leq q^{nh_{q}(\tau)}.
		$$
		Asymptotically,
		$$
		\frac{1}{n} \log_{q} \left( \binom{n}{t}(q-1)^{t} \right) = h_{q}(\tau) + O\left(\frac{\log_{q}n}{n}\right).
		$$
	\end{lemma}

	{\bf \noindent Probabilistic notation.} During these lecture notes we wish to emphasize on which probability space the probabilities or the expectations are taken. Therefore we will denote by a subscript the random variable specifying the associated probability space over which the probabilities or expectations are taken.
	For instance the probability $\mathbb{P}_{X}(A)$ of the event $A$ is taken over $\Omega$ the probability space over which the random variable
	$X$ is defined, {\em i.e.} if $X$ is for instance a real random variable, $X$ is a function from a probability space $\Omega$ to $\mathbb{R}$, and the aforementioned probability is taken according to
	the measure chosen for $\Omega$.
	\newline

	{\bf \noindent Statistical distance.} An essential tool for many cryptographic applications is the {\em statistical distance}, sometimes called the {\em total variational distance}. It is a distance for probability distributions, which in the case where $X$ and $Y$ are two random variables taking their values in a same finite space $\mathcal{E}$ is defined as
	\begin{equation}\label{eq:def1SD} 
	\Delta(X,Y) \eqdef \frac{1}{2} \sum_{a\in \mathcal{E}} \left| \mathbb{P}\left(X=a\right) - \mathbb{P}\left( Y= a \right) \right|. 
	\end{equation} 
	An equivalent definition is given by
	\begin{equation}\label{eq:def2SD} 
	\Delta\left( X,Y \right) \eqdef \mathop{\max}\limits_{A \subseteq\mathcal{E}}\left| \mathbb{P}_{X}(A)-\mathbb{P}_{Y}(A)\right|.
	\end{equation} 
	Depending on the context, \eqref{eq:def1SD} or \eqref{eq:def2SD} is the most useful. A direct consequence of \eqref{eq:def2SD} is that given any event $A$, we have $\left| \mathbb{P}_{X}(A) - \mathbb{P}_{Y}(A) \right| \leq \Delta(X,Y)$. Therefore, computing probabilities over $X$ or $Y$ will differ by at most $\Delta(X,Y)$.  Furthermore, given a single observation, coming from $X$ or $Y$ with probability $1/2$, we will be able to guess which with probability at most $1/2 + \Delta(X,Y)/2$ and there is a strategy to reach this probability of guessing correctly.

	The statistical distance enjoys many interesting properties. Among others,  	
it cannot increase by applying a function $f$,
	$$
	\Delta(f(X),f(Y)) \leq \Delta(X,Y) \quad (\mbox{\em data processing inequality}).
	$$
	The function $f$ can be randomized, but its internal randomness has to be independent from $X$ and $Y$ for the data processing inequality to hold. In particular, it implies that the ``success'' probability of any algorithm $\mathcal{A}$ for inputs distributed according to $X$ or $Y$, can only differ by at most $\Delta(X,Y)$.

	In our applications we will focus on distributions $X,Y$ such that their statistical distance is negligible. It will show (as a consequence of the data processing inequality) that $X$ and $Y$ are computationally indistinguishable\footnote{See here for a definition: \href{ }{https://www.cs.princeton.edu/courses/archive/spr10/cos433/lec4.pdf}} without requiring any computational argument with a reduction.

	One can define various other distances for capturing in a cryptographic context the differences between two distributions. For instance, we can cite the family Renyi divergences, but this is out of the scope of these lecture notes.

	\section{Random Codes}

	{\noindent \bf The model of random codes.} In these lecture notes we will use two probabilistic models that will be referred to as {\em random $\lbrack n,k \rbrack_{q}$-codes}. The first one is by choosing a code $\CC$ by picking uniformly at random a generator matrix $\vec{G}\in\Fq^{k\times n}$ ({\em i.e.} $\CC \eqdef \left\{ \vec{m}\vec{G} \mbox{ : } \vec{m} \in \Fq^{k} \right\}$). 
However, all the probabilistic results of these lecture notes are easier to prove if, instead, we choose $\CC$ by picking uniformly at random a parity-check matrix $\vec{H}\in \Fq^{(n-k)\times n}$ ({\em i.e.} $\CC \eqdef \left\{ \vec{c}\in\Fq^{n} \mbox{ : } \vec{H}\transpose{\vec{c}} = \mathbf{0} \right\}$). 
We will denote $\mathbb{P}_{\vec{G}}$ and $\mathbb{P}_{\vec{H}}$ respectively the probabilities in these two models.

	It may be pointed out that in both models we don't pick uniformly at random an $\lbrack n,k \rbrack_{q}$-code. Indeed,  the first model always produces codes of dimension $\leq k$ whereas in the second model codes are always of dimension $\geq k$. One may wonder why don't we pick $\vec{G}$ ({\em resp.} $\vec{H}$) uniformly at random among the $k \times n$ ({\em resp.} $(n-k)\times n$) matrices of rank $k$ ({\em resp.} $n-k$)? First, computations are much more complicated in this ``exact'' model. Furthermore, it turns out that it is pointless. Roughly speaking, the $\vec{G}$-model produces codes of dimension $= k$ with probability $1-O(q^{-(n-k)})$ while in the $\vec{H}$-model we get a code of dimension $=k$ with probability $1 - O(q^{-k})$. As shown in the following lemma this result can even be expressed in a stronger way, our probabilistic models are exponentially close, {\em for the statistical distance}, to the ``exact'' model. Therefore all our computations in the $\vec{G}$ or $\vec{H}$ models can really be thought as by picking uniformly at random an $\lbrack n,k \rbrack_{q}$-code $\CC$.

	\begin{lemma} \label{lemma:Rankk}
		Let $\vec{G} \in\Fq^{k \times n}$ (resp. $\vec{H}\in\Fq^{(n-k)\times n}$) be a uniformly random matrix and $\vec{G}_{k}\in\Fq^{k \times n}$ (resp. $\vec{H}_{n-k}\in\Fq^{(n-k)\times n}$) be a uniformly random matrix of rank $k$ (resp. $n-k$). We have:
		$$
		\Delta\left( \vec{G},\vec{G}_{k} \right) = O\left( q^{-(n-k)} \right) \quad \left(\mbox{\textit{resp.} } \Delta\left( \vec{H},\vec{H}_{n-k} \right) = O\left( q^{-k} \right) \right).
		$$
	\end{lemma}

	\begin{proof} 
		Let us prove the lemma for $(\vec{G},\vec{G}_{k})$, the other case will be similar. First it is a classical fact that the density of rank $k$ matrices among $\Fq^{k\times n}$ is equal to $1- O\left(q^{-(n-k)}\right)$. Therefore, given some rank $k$ matrix $\vec{R}\in\Fq^{k\times n}$, we have:
		$$
		\mathbb{P}\left( \vec{G}_{k} = \vec{R} \right) =  \frac{1}{q^{k \times n}\left( 1- O(q^{-(n-k)})\right)}.
		$$
		It leads to the following computation:
		\begin{align*}
			2\Delta\left( \vec{G},\vec{G}_{k} \right) &= \sum_{\substack{\vec{R}\in\Fq^{k \times n} \\ \textup{rank}(\vec{R}) = k}}\left| \mathbb{P}(\vec{G} = \vec{R}) - \mathbb{P}(\vec{G}_{k} = \vec{R})  \right| + \sum_{\substack{\vec{R}\in\Fq^{k \times n} \\ \textup{rank}(\vec{R}) \neq k}}  \mathbb{P}(\vec{G} = \vec{R}) \\
			&= \sum_{\substack{\vec{R}\in\Fq^{k \times n} \\ \textup{rank}(\vec{R}) = k}}\left| \frac{1}{q^{k\times n}}\left( 1- \frac{1}{1-O\left( q^{-(n-k)}\right)} \right)  \right| + \sum_{\substack{\vec{R}\in\Fq^{k \times n} \\ \textup{rank}(\vec{R}) \neq k}} \frac{1}{q^{k\times n}} \\
&= O\left(q^{-(n-k)}\right)
		\end{align*}
	which concludes the proof. 
	\end{proof}

	Now one may wonder why do we consider two models for random $\lbrack n,k \rbrack_{q}$-codes? It turns out that depending of the context, computations might be easier and/or more natural in one model rather than in the other one. In addition, for the same reasons as those given in the previous lemma, $\vec{G}$ and $\vec{H}$ models are closely related, computations in both probabilistic models will outcome the same results up to an additive exponentially small factor.
	
	\begin{lemma}
		Let $\mathcal{E}$ be a set of linear codes of length $n$ in $\Fq$ which is defined as an event. We have,
		$$
		\left|\mathbb{P}_{\vec{G}}(\mathcal{E}) - \mathbb{P}_{\vec{H}}(\mathcal{E}) \right|  = O\left( q^{-\min(k,n-k)} \right). 
		$$
	\end{lemma}

	\begin{proof}
	Let $\vec{G}_{k}$ and $\vec{H}_{n-k}$ be defined as in Lemma \ref{lemma:Rankk}. Notice that $\mathbb{P}_{\vec{G}_{k}}(\mathcal{E}) = \mathbb{P}_{\vec{H}_{n-k}}(\mathcal{E})$, in both models, we exactly pick uniformly at random an $\lbrack n,k \rbrack_{q}$-code. It leads to the following computation:
	\begin{align*}
		\left|\mathbb{P}_{\vec{G}}(\mathcal{E}) - \mathbb{P}_{\vec{H}}(\mathcal{E}) \right|  &\leq \left|\mathbb{P}_{\vec{G}}(\mathcal{E}) - \mathbb{P}_{\vec{G}_{k}}(\mathcal{E}) \right|  + \left|\mathbb{P}_{\vec{H}_{n-k}}(\mathcal{E}) - \mathbb{P}_{\vec{H}}(\mathcal{E}) \right| \\
		&\leq \Delta(\vec{G},\vec{G}_{k}) + \Delta(\vec{H},\vec{H}_{n-k})
	\end{align*}
where in the last line we used Equation \eqref{eq:def2SD}. It concludes the proof by using Lemma \ref{lemma:Rankk}. 
	\end{proof}

	\begin{exercise}\label{exo:GenversusPar}
		Let us introduce the following variant of $\mathsf{DP}$ (see Problem \ref{prob:decoGenR}) with generator matrices instead of parity-check matrices

		$\mathsf{DP'}(n,q,R,\tau)$. Let $k \eqdef \lfloor Rn \rfloor$ and $t \eqdef \lfloor \tau n \rfloor$.

			\begin{itemize}		
				\item \textup{Input:}  $(\vec{G},\vec{y}\eqdef \vec{s}\vec{G} + \vec{x})$ where $\vec{G},\vec{s}$ and $\vec{x}$ are  uniformly distributed over $\Fq^{k\times n}$, $\Fq^{k}$ and words of Hamming weight $t$ in $\Fq^{n}$.

				\item \textup{Output:} an error $\vec{e}\in \Fq^{n}$ of Hamming weight $t$ such that $\vec{y} - \ev = \vec{m}\vec{G}$ for some $\vec{m} \in \Fq^{k}$. 
			\end{itemize}
		Show that for any algorithm $\mathcal{A}$ solving this problem with probability $\varepsilon$ and time $T$, there exists an algorithm $\mathcal{B}$ which solves $\DP(n,q,R,\tau)$ in time $O\left(n^{3}+ T\right)$ with probability $\geq \varepsilon - O\left( q^{-\min(k,n-k)} \right)$. Show that we can exchange $\DP'$ by $\DP$ in the previous question. 
	\end{exercise}

	\begin{remark}
		The above exercise shows that defining $\DP$ with generator or parity-check matrices is just a matter of personal taste, it does not change the average hardness.  
	\end{remark}

	{\bf \noindent A first computation with random codes.} Now that random codes are well defined, we are ready to make our first computation in this probabilistic model. The following elementary lemma gives the probability (over the codes) that a fixed non-zero word $\vec{y}$ reaches some syndrome $\vec{s}$ according to the code. In particular, by setting $\vec{s}$ to $\vec{0}$, we obtain the probability that $\vec{y}$ belongs to the code.  This lemma will be at the core of all our results about random codes.

	\begin{lemma}\label{lemma:belongCode} Given $\vec{s} \in \Fq^{n-k}$ and $\vec{y}\in\Fq^{n}$ such that $\vec{y} \neq \mathbf{0}$, we have for $\vec{H}$ being uniformly distributed at random in $\Fq^{(n-k)\times n}$,
		$$
		\mathbb{P}_{\vec{H}}(\vec{y}\transpose{\vec{H}} = \vec{s}) = \frac{1}{q^{n-k}}.
		$$
	\end{lemma}

	\begin{proof} Let $h_{i,j}$ be the coefficient of $\vec{H}$ at position $(i,j)$. Without loss of generality, we can suppose that $y_{1} = 1$ (by permuting $\vec{y}$ if $y_{1}=0$ and then multiplying by $y_{1}^{-1}$ which is possible as we work in $\Fq$). The probability that we are looking for is the probability of the following event:
		$$
		\forall i \in \llbracket 1,n-k \rrbracket, \quad h_{i,1} = s_i - \sum_{j=2}^{n} h_{i,j}y_{j}
		$$
		 Recall that $\vec{H}$ is uniformly distributed: the $h_{i,j}$'s are independent and equidistributed. Therefore the above $n-k$ equations will be independently true with probability $1/q$ which concludes the proof. 
\end{proof} 
	
	\begin{exercise}
		Show that for any non-zero $\vec{y}\in\Fq^{n}$,
		$$
		\mathbb{P}_{\vec{G}}(\vec{y}\in\dual{\CC}) = \frac{1}{q^{k}}.
		$$
	\end{exercise}
	
	\section{Weight Distribution of Cosets of Random Codes}
	
	The aim of this section is to answer the following question: given a  random code $\CC$ and a fixed vector $\vec{y}\in\Fq^{n}$, how many codewords $\vec{c}\in\CC$ do we expect to be at Hamming distance $t$ from $\vec{y}$? Or equivalently, given a parity-check matrix of our random code $\CC$ and a fixed syndrome $\vec{s}$, how many vectors $\vec{e}$ of Hamming weight $t$ do we expect to reach the syndrome $\vec{s}$ according to $\vec{H}$? Notice that deriving an answer to these questions in the particular cases $\vec{y} = \vec{0}$ and $\vec{s} = \vec{0}$ enables to compute the expected number of codewords of weight $t$ in $\CC$. It will be useful to compute the expected minimum distance of a code.

	These results will have an important consequence: a better understanding of the Decoding Problem ($\mathsf{DP}$) \iftoggle{amsbook}{that was defined in Problem \ref{prob:decoGenR}.}{that we recall now.

		\begin{problem}[\textup{Decoding Problem - $\mathsf{DP}(n,q,R,\tau)$}]
		Let $k \eqdef \lfloor Rn \rfloor$ and $t \eqdef \lfloor \tau n \rfloor$.

		\begin{itemize}		
			\item \textup{Input :}  $(\vec{H},\vec{s}\eqdef \vec{x}\transpose{\vec{H}})$ where $\vec{H}$ (resp. $\vec{x}$) is uniformly distributed over $\Fq^{(n-k)\times n}$ (resp. $\mathcal{S}_{t}$, words of Hamming weight $t$ in $\Fq^{n}$).

			\item \textup{Output :} an error $\vec{e}\in \Fq^{n}$ of Hamming weight $t$ such that $\vec{e}\transpose{\vec{H}} = \vec{s}$. 
		\end{itemize}
	\end{problem}

}
	According to our probabilistic model, this problem really corresponds to decode a random $\lbrack n,k \rbrack_{q}$-code of parity-check matrix $\vec{H}$. In that case it is natural to wonder how many vectors $\vec{e} \in \mathcal{S}_{t}$  are expected to reach the syndrome $\vec{s}$ according to $\vec{H}$, but why? 
To understand this let us take a toy example.  A trivial solution to solve $\mathsf{DP}$ is to pick a random error $\vec{e}\in \mathcal{S}_{t}$ with the hope that it gives a solution. By definition there is a solution to our problem (here $\vec{x}$). If there is exactly one solution, our success probability is given by $\frac{1}{\binom{n}{t}(q-1)^{t}}$. But now imagine that we expect $N$ solutions to our problem. In that case we would expect our success probability to be equal to $\approx \frac{N}{\binom{n}{t}(q-1)^{t}}$. It is therefore important to know the value of $N$ to be able to predict the running time of our algorithm. It is the aim of what follows. 
	\newline

	{\noindent \bf Notation.} Given $\vec{H}\in\Fq^{(n-k)\times n}$ and $\vec{s}\in\Fq^{n-k}$, let 
	$$
	N_{t}(\CC,\vec{s}) \eqdef \sharp  \left\{ \vec{e}\in\mathcal{S}_{t} \mbox{ : } \vec{e}\transpose{\vec{H}}=\vec{s} \right\} 
$$
	where implicitly $\CC$ is defined as $\{\vec{c} \in \Fq^{n} \mbox{ : } \vec{H}\transpose{\vec{c}} = \mathbf{0}\}$. Notice that $N_{t}(\CC,\vec{s})$ is a random variable that gives the number of solutions of $\mathsf{DP}$ with input $(\vec{H},\vec{s})$ (where $\vec{s}\in\Fq^{n-k}$ is fixed and not necessarily computed as some $\vec{x}\transpose{\vec{H}}$ for $\vec{x}\in\mathcal{S}_{t}$). On the other hand, $N_{t}(\CC,\vec{0})$ is a random variable that gives the number of codewords $\vec{c} \in \CC$ of Hamming weight $t$. 
	\newline

	{\bf \noindent Expected weight distribution of cosets.} 
From now on, our main objective is to compute the expected number of Hamming weight $t$ vectors in a given coset, namely to compute the expectation of $N_{t}(\CC,\vec{s})$ over $\CC$. To avoid any suspense, we will prove that for any syndrome $\vec{s}\in\Fq^{n-k}$, the expectation of $N_{t}(\CC,\vec{s})$ is given by $\binom{n}{t}(q-1)^{t}/q^{n-k}$. 
However, before showing this result, let us start to understand how this quantity behaves as function of it parameters, namely $\tau =t/n$ and $R=k/n$. By Lemma \ref{lemma:asymptSphere}, we have
	\begin{equation}\label{eq:expectedSol}
	\frac{1}{n} \log_{q}  \binom{n}{t}(q-1)^{t}/q^{n-k} =  h_{q}(\tau) - (1-R) + O\left(\frac{\log_{q}n}{n}\right).
	\end{equation}
	Recall now that $x\in[0,1] \mapsto h_{q}(x)$ is an increasing function over $\left[0,\frac{q-1}{q}\right]$ and a decreasing function over $\left[\frac{q-1}{q},1\right]$. Furthermore, $h_{q}(0) = 0$ and  $h_{q}(1) = \log_{q}(q-1)$. This shows that $N_{t}(\CC,\vec{s})$ is expected (according to $\tau$)  to be exponentially small or large (in $n$) at the exception of one value $\tau^{-}$ and potentially a second one in the case where $(1-R) \geq \log_{q}(q-1)$, that we will denote $\tau^{+}$. We summarize the picture by drawing in Figure \ref{figure:logSol} the logarithm in basis $3$ (for $n$ large enough) of $\binom{n}{t}(q-1)^{t}/q^{n-k}$ when $q=3$ and $k/n = 1/4$. 
	\begin{center}
		\begin{figure}
			\includegraphics[height=6cm]{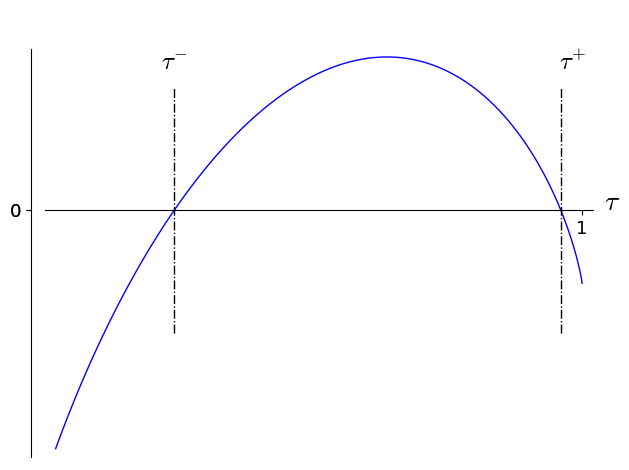}
			\caption{$\mathop{\lim}\limits_{n\to+\infty} \frac{1}{n}\;\log_{q} \binom{n}{t}(q-1)^{t}/q^{n-k}$ when $q = 3$ and $k/n = 1/4$ as function of $\tau = t/n$.
				\label{figure:logSol}}
		\end{figure}
	\end{center}

	It turns out that an analytic expression of $\tau^{-}$ and $\tau^{+}$ can be given,

	\begin{equation}\label{eq:tau-+}
		\tau^{-} \eqdef g_q^{-}(1-R) \quad \mbox{and} \quad \tau^{+}\eqdef g_q^{+}(1-R) \mbox{ when } R \leq 1 - \log_{q}(q-1)
\end{equation}
	where $g_{q}^{-}$ ({\em resp.} $g_{q}^{+}$) denotes the inverse of $h_{q}$ over $\left[0,\frac{q-1}{q}\right]$ \Big({\em resp.} $\left[\frac{q-1}{q},1\right]$\Big).

	\begin{remark}
		As we will see in Section \ref{sec:minDistance}, $\tau^{-}$ is commonly called the relative Gilbert-Varshamov distance or bound. 
	\end{remark}
	
	Quantities $\tau^{-}$ and $\tau^{+}$ give the boundaries between which we expect $N_{t}(\CC,\vec{s})$ to be exponentially large as we show now.

	\begin{proposition}\label{propo:expectation} Let $k \eqdef \lfloor Rn \rfloor$, $t \eqdef \lfloor \tau n \rfloor$ and $\vec{s}\in\Fq^{n-k}$. We have:
		\begin{equation} \label{eq:expectN} 
		\mathbb{E}_{\vec{H}}(N_{t}(\CC,\vec{s})) = \frac{\binom{n}{t}(q-1)^{t}}{q^{n-k}}.
		\end{equation} 
	When $\tau \in \left\{
	\begin{array}{ll}
		 (\tau^{-},\tau^{+}) & \mbox{if } R \leq 1 - \log_{q}(q-1)\\
		(\tau^{-},1) & \mbox{otherwise}
	\end{array}
	\right.$,
we expect $N_{t}(\CC,\vec{s})$ to be exponentially large:
	$$
	\mathbb{E}_{\vec{H}}(N_{t}(\CC,\vec{s})) = q^{\alpha n(1+o(1))} \quad \mbox{where} \quad \alpha \eqdef h_{q}(\tau) - 1 + R > 0.
	$$
	In the case where $\tau = 	\left\{ \begin{array}{ll}
		\tau^{-} \mbox{ or } \tau^{+} & \mbox{if } R \leq 1 - \log_{q}(q-1) \\
		\tau^{-} & \mbox{otherwise.}
	\end{array}
	\right.
	$ 
, the expectation of  $N_{t}(\CC,\vec{s})$ equals $P(n)(1+o(1))$ for some polynomial $P$. 
	\end{proposition}

	\begin{proof} Let $\mathds{1}_{\vec{e}}$ be the indicator function of the event ``$\vec{e}\transpose{\vec{H}} = \vec{s}$''. It is readily verified that by definition,
		\begin{equation*}
			N_{t}(\CC,\vec{s}) = \sum_{\vec{e}\in \mathcal{S}_{t}}\mathds{1}_{\vec{e}}.
		\end{equation*} 
	We have the following computation,
		\begin{align*}
			\mathbb{E}_{\vec{H}}(N_{t}(\CC,\vec{s})) &= \mathbb{E}_{\vec{H}}\left( \sum_{\vec{e}\in\mathcal{S}_{t}} \mathds{1}_{\vec{e}} \right) \\
			&=  \sum_{\vec{e}\in\mathcal{S}_{t}} \mathbb{E}_{\vec{H}}(\mathds{1}_{\vec{e}}) \quad \mbox{(by linearity of the expectation)} \\
			&= \sum_{\vec{e}\in\mathcal{S}_{t}} \mathbb{P}_{\vec{H}}\left( \vec{e}\transpose{\vec{H}} = \vec{s} \right) 
		\end{align*}
	which gives \eqref{eq:expectN} by using Lemma \ref{lemma:belongCode}. The second part of the proposition is a consequence of Lemma \ref{lemma:asymptSphere}. 
	\end{proof} 

	\begin{exercise}
		Show that the average number of solutions of $\DP(n,q,R,\tau)$ (over the input distribution) is given by 
		$$
		1+ \frac{\binom{n}{t}(q-1)^{t}-1}{q^{n-k}}
		$$
	\end{exercise}
	
	\begin{remark}
		Proposition \ref{propo:expectation} can actually be stated more generally. Given some set $\mathcal{E}\subseteq \mathbb{F}_{q}^{n}$, we can show that the expected number of vectors $\vec{e}\in\mathcal{E}$ that reach some syndrome for a random $\lbrack n,k \rbrack_{q}$-code is given by $\sharp \mathcal{E}/q^{n-k}$. 
	\end{remark}

	\begin{remark}
		The expected number of codewords of weight $t$ in a random $\lbrack n,k \rbrack_{q}$-code is given by $\binom{n}{t}(q-1)^{t}/q^{n-k}$ (by setting $\vec{s}$ to $\vec{0}$ in Proposition \ref{propo:expectation}). This statement, in the same manner as in the previous remark, can be generalized to give the expected number of codewords in any set $\mathcal{E} \subseteq \mathbb{F}_{q}^{n}$. In particular, it can be used to obtain the expected number of codewords  of ``weight'' $t$ that belong to a random code, for any notion of weight and therefore any metric. 
	\end{remark}

	\begin{exercise}
		Show that $(t>0)$,
		$$
		\mathbb{E}_{\vec{G}}\left( \sharp  \left\{\vec{m} \in \F_{q}^{k}  : |\vec{m}\vec{G}| = t \right\} \right) = \frac{q^{k}-1}{q^{n}} \binom{n}{t}(q-1)^{t} \quad \mbox{and} \quad \mathbb{E}_{\vec{H}}\left( \sharp \left\{ \vec{c}\in\CC : |\vec{c}| \mbox{ is odd} \right\} \right)  = \frac{1}{2}\;\frac{ q^{n} - (2-q)^{n}}{q^{n-k}}.
		$$
		\qquad\small{{\bf Hint:} {\em For the first part of the exercise first show that $\vec{m}\vec{G}$ is uniformly distributed over $\Fq^{n}$ when $\vec{m}\in\Fq^{k}\backslash\{\vec{0}\}$.}}
	\end{exercise}

	At this point our work has given the {\em expected} number of solutions of $\mathsf{DP}(n,q,R,\tau)$.
Situation is depicted in Figure \ref{fig:hardEasyDPInjSurj}.

	\begin{figure}[htb]
		\centering
		\begin{tikzpicture}[scale=0.83]
			\tikzstyle{valign}=[text height=1.5ex,text depth=.25ex]
			\draw[line width=2pt,gray] (-0.7,2) -- (0.3,2);
			\draw (3.25,2.5) node[red]{{\sf Hard}};
			\draw (12.25,2.5) node[right,red]{{\sf Hard}};
\draw[line width=2pt,red!50] (0.3,2) -- (5,2);
			\draw[line width=2pt,blue!50] (5,2) --
			node[above,midway,blue,valign]{{\sf Easy}} (11,2);
			\draw[line width=2pt,red!50] (11,2) -- (16,2);
			\draw[->,>=latex,line width=2pt,gray] (16,2) -- (17,2)
			node[right,black] {$\displaystyle \tau$};
			\tikzstyle{valign}=[text height=2ex]
			\draw[thick] (0.3,1.9) node[below,valign]{$0$} -- (0.3,2.1);
			\draw[thick] (5,1.9) node[below,valign]{$(1-R)\;\frac{q-1}{q}$} -- (5,2.1);
			\draw[thick] (11,1.9) node[below,valign]{$R+(1-R)\;\frac{q-1}{q}$~~} -- (11,2.1);
			\draw[thick] (13.75,0.5) node[below,valign]{$\tau^{+}$} -- (13.75,0.5);
			\draw[thick] (2.5,0.5) node[below,valign]{$\tau^{-}$} -- (2.5,0.5);
			\draw[thick] (16,1.9) node[below,valign]{$1$} -- (16,2.1);

			\draw[<->,>=latex,thin,purple!50] (2.5,0.9) -- node[below,purple,midway]{{\sf exponentially many solutions}} (13.75,0.9);
			\draw[thick] (2.5,0.5) node[below,valign]{$ $} -- (2.5,3.5);
			\draw[thick] (13.75,0.5) node[below,valign]{$ $} -- (13.75,3.5);
			
			\draw[<->,>=latex,thin,purple!50] (0.3,3.3) -- node[below,purple,midway]{{\sf one solution}} (2.5,3.3);
			
			\draw[<->,>=latex,thin,purple!50] (13.75,3.3) -- node[below,purple,midway]{{\sf one solution}} (16,3.3);

		\end{tikzpicture}
		\caption{Hardness and expected number of solutions of $\mathsf{DP}(n,q,R,\tau)$ as function of $\tau$.\label{fig:hardEasyDPInjSurj}}
	\end{figure}
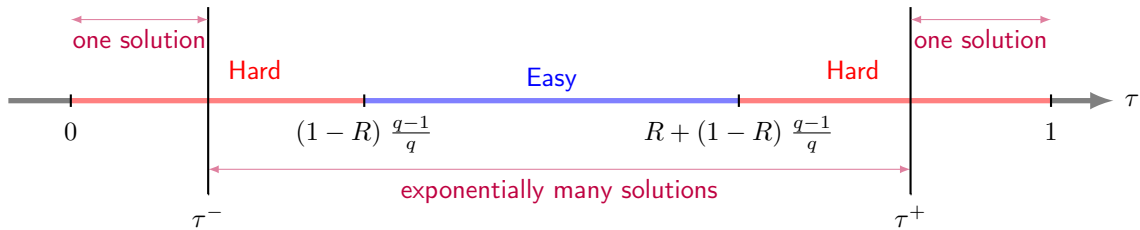

	However, can we be much more precise? For instance, can we give with an overwhelming probability, and therefore for almost all codes, the number of solutions of $\mathsf{DP}$? The answer is yes. Below is given two techniques for achieving this, the first one uses Markov's inequality ({\em first moment technique}) and the second one, which is more accurate, uses Bienaym\'{e}-Tchebychev's inequality ({\em second moment technique}).

	\begin{proposition}[First Moment Technique]\label{propo:FT}
		
Let $\vec{s}\in\Fq^{n-k}$. 
		For any $a > 0$, we have
		$$
		\mathbb{P}_{\vec{H}}(N_{t}(\CC,\vec{s}) > a) \leq \frac{1}{a}\; \frac{\binom{n}{t}(q-1)^{t}}{q^{n-k}}.
		$$
	\end{proposition}

	\begin{proof}
By Proposition \ref{propo:expectation}, $\mathbb{E}_{\vec{H}}(N_{t}(\CC,\vec{s})) = \frac{\binom{n}{t}(q-1)^{t}}{q^{n-k}}$. It remains to apply Markov's inequality to conclude the proof. 
	\end{proof}

	Notice that Proposition \ref{propo:FT} is not very accurate. One has to choose $a \gg \frac{\binom{n}{t}(q-1)^{t}}{q^{n-k}}$ to obtain a negligible probability that $N_{t}(\CC,\vec{s})$ is larger than $a$. But the expectation of $N_{t}(\CC,\vec{s})$ is exactly given by $\frac{\binom{n}{t}(q-1)^{t}}{q^{n-k}}$. A meaningful result would be
	$$
	\mathbb{P}_{\vec{H}}\left( \left|N_{t}(\CC,\vec{s}) - \frac{\binom{n}{t}(q-1)^{t}}{q^{n-k}} \right| > a  \right) < \varepsilon,
	$$
	for some $\varepsilon \in \textup{negl}(n)$ and $a$ be such that $a\in \frac{\binom{n}{t}(q-1)^{t}}{q^{n-k}}\textup{negl}(n)$.
	It is precisely the aim of the following proposition.

	\begin{proposition}[Second Moment Technique]\label{propo:ST} 
Let $\vec{s}\in\Fq^{n-k}$. For any $a > 0$, we have
		$$
		\mathbb{P}_{\vec{H}}\left(\left| N_{t}(\CC,\vec{s}) - \frac{\binom{n}{t}(q-1)^{t}}{q^{n-k}} \right| \geq  a \right)
		\leq \frac{(q-1)\binom{n}{t}(q-1)^{t}}{a^{2}q^{n-k}}.
		$$
		
\end{proposition}

	\begin{proof}
		Let $\mathds{1}_{\vec{e}}$ be the indicator function of the event ``$\vec{e}\transpose{\vec{H}} = \vec{s}$''.
By using Bienaym\'{e}-Tchebychev’s inequality with the random variable $N_{t}(\CC,\vec{s}) = \sum_{\vec{e}\in \mathcal{S}_{t}}\mathds{1}_{\vec{e}}$, we obtain
		\begin{align}
				\mathbb{P}_{\vec{H}}\left(\left| N_{t}(\CC,\vec{s}) - \frac{\binom{n}{t}(q-1)^{t}}{q^{n-k}} \right|  \geq a\right) &\leq \frac{\mathbf{Var}_{\vec{H}}(N_{t}(\CC,\vec{s}))}{a^{2}} \nonumber \\
			&= \frac{1}{a^{2}}\left(  {\sum_{\vec{e}\in \mathcal{S}_{t}}\mathbf{Var}_{\vec{H}}(\mathds{1}_{\vec{e}}) + \sum_{\substack{\vec{x},\vec{y}\in \mathcal{S}_{t}\\ \vec{x}\neq \vec{y} }} \mathbb{E}_{\vec{H}}(\mathds{1}_{\vec{x}}\mathds{1}_{\vec{y}}) - \mathbb{E}_{\vec{H}}(\mathds{1}_{\vec{x}})\mathbb{E}_{\vec{H}}(\mathds{1}_{\vec{y}}) } \right)\nonumber \\
			&\leq \frac{1}{a^2}\left( \sum_{\vec{e} \in \mathcal{S}_{t}} \mathbb{E}_{\vec{H}}(\und_\vec{e})+ \sum_{\substack{\vec{x},\vec{y}\in \mathcal{S}_{t}\\ \vec{x}\neq \vec{y} }} \mathbb{E}_{\vec{H}}(\mathds{1}_{\vec{x}}\mathds{1}_{\vec{y}}) - \mathbb{E}_{\vec{H}}(\mathds{1}_{\vec{x}})\mathbb{E}_{\vec{H}}(\mathds{1}_{\vec{y}}) \right)  \nonumber\\
			&=\frac{1}{a^{2}}\left( \frac{\binom{n}{t}(q-1)^{t}}{q^{n-k}}+ \sum_{\substack{\vec{x},\vec{y}\in \mathcal{S}_{t}\\ \vec{x}\neq \vec{y} }} \mathbb{E}_{\vec{H}}(\mathds{1}_{\vec{x}}\mathds{1}_{\vec{y}}) - \mathbb{E}_{\vec{H}}(\mathds{1}_{\vec{x}})\mathbb{E}_{\vec{H}}(\mathds{1}_{\vec{y}}) \right) \label{eq:varD}
		\end{align}
		where we used that $\mathbf{Var}_{\vec{H}}(\mathds{1}_{\vec{e}}) \leq \mathbb{E}_{\vec{H}}(\mathds{1}_{\vec{e}}^{2}) = \mathbb{E}_{\vec{H}}(\mathds{1}_{\vec{e}})$. Let us now upper-bound the second term of the inequality. To this aim let us prove the following lemma. 
		
		\begin{lemma}\label{lemma:upperBExpectation} We have
		\begin{equation*} 
			\mathbb{E}_{\vec{H}}(\mathds{1}_{\vec{x}}\mathds{1}_{\vec{y}}) \leq \left\{
			\begin{array}{ll}
				\nicefrac{1}{q^{n-k}}& \mbox{ if } \xv \mbox{ and } \yv \mbox{ are colinear} \\
				\nicefrac{1}{q^{2(n-k)}} & \mbox{ otherwise.}
			\end{array}
			\right.
		\end{equation*}
		\end{lemma} 
	
		\begin{proof} The result is clear when $\vec{x}$ and $\vec{y}$ are colinear by Lemma \ref{lemma:belongCode}. Let us suppose that $\vec{x}$ and $\vec{y}$ are not colinear and define
			$$
			\varphi : \vec{h} \in \Fq^{n} \longmapsto \left( \vec{h} \cdot \vec{x}, \vec{h} \cdot \vec{y} \right). 
			$$
			It is readily verified that this linear application has a kernel of $\Fq$-dimension $n-2$. Therefore, for any $a,b \in \Fq$ and $\vec{h} \in \Fq^{n}$ being uniformly distributed,
			\begin{equation}\label{eq:qsquare} 
			\mathbb{P}_{\vec{h}}\left( \varphi(\vec{h}) = (a,b) \right) = \frac{q^{n-2}}{q^{n}} = \frac{1}{q^{2}}
			\end{equation} 
			Let us remark now that
			\begin{align} 
			\mathbb{E}_{\vec{H}}\left(\mathds{1}_{\vec{x}}\mathds{1}_{\vec{y}}\right) &= \mathbb{P}_{\vec{H}}\left( \vec{x}\transpose{\vec{H}} = \vec{s} \;\; \mbox{\em and} \; \; \vec{y} \transpose{\vec{H}} = \vec{s} \right) \nonumber  \\
			&= \mathbb{P}_{\vec{h}}\left( \varphi(\vec{h}) = (a,b) \right)^{n-k} \label{eq:toPlug} 
			\end{align} 
		where in the last line we used that each rows of $\vec{H}$ are independent and uniformly distributed. To conclude the proof it remains to plug Equation \eqref{eq:qsquare} in Equation \eqref{eq:toPlug}. 
		\end{proof}  	
		Lemma \ref{lemma:upperBExpectation} enables us to deduce that
		\begin{align}
			\sum_{\substack{\vec{x},\vec{y}\in \mathcal{S}_{t}\\ \vec{x}\neq \vec{y} }} \mathbb{E}_{\vec{H}}(\mathds{1}_{\vec{x}}\mathds{1}_{\vec{y}}) - \mathbb{E}_{\vec{H}}(\mathds{1}_{\vec{x}})\mathbb{E}_{\vec{H}}(\mathds{1}_{\vec{y}}) &\leq \sum_{\vec{x}\in \mathcal{S}_{t}}\sum_{\substack{\vec{y}\in \mathcal{S}_{t}\setminus\vec{x}: \\ \text{ colinear to }\vec{x} }} \frac{1}{q^{n-k}} - \frac{1}{q^{2(n-k)}} \nonumber\\
			&\leq \sum_{\vec{x}\in \mathcal{S}_{t}}\sum_{\substack{\vec{y}\in \mathcal{S}_{t}\setminus\vec{x}: \\ \text{ colinear to }\vec{x} }} \frac{1}{q^{n-k}} \nonumber \\
			&\leq  \frac{(q-2) \binom{n}{t}(q-1)^{t}}{q^{n-k}} \label{eq:finVar}
		\end{align}
		It gives by plugging \eqref{eq:finVar} in \eqref{eq:varD} 
		\begin{align*}
			\mathbb{P}_{\vec{H}}\left(\left|N_{t}(\CC,\vec{s}) -\frac {\binom{n}{t}(q-1)^{t}}{q^{n-k}}\right| \geq a\right) &\leq \frac{1}{a^{2}}\left(   \frac{\binom{n}{t}(q-1)^{t}}{q^{n-k}}+ \frac{(q-2)\binom{n}{t}(q-1)^{t}}{q^{n-k}}\right)\\
			&= \frac{(q-1)\binom{n}{t}(q-1)^{t}}{a^{2}q^{n-k}}
		\end{align*} 
		which concludes the proof.
\end{proof}

	The point of this proposition is that, for relative weights $t/n \in (\tau^{-},\tau^{+})$
	the term $\binom{n}{t}(q-1)^{t}/q^{n-k}$, is exponentially large.
Therefore, by carefully setting $a$ in this case we deduce $N_{t}(\CC,\vec{s})$ with a very good precision.  For instance, if $t/n \in (\tau^{-},\tau^{+})$, let $a=\left( \frac{\binom{n}{t}(q-1)^{t}}{q^{n-k}} \right)^{3/4}$ to obtain: 
	$$
	\mathbb{P}_{\vec{H}}\left(\left| N_{t}(\CC,\vec{s}) - \frac{\binom{n}{t}(q-1)^{t}}{q^{n-k}} \right| \geq  \left( \frac{\binom{n}{t}(q-1)^{t}}{q^{n-k}}\right)^{3/4} \right) \leq (q-1)\; \sqrt{\frac{q^{n-k}}{\binom{n}{t}(q-1)^{t}}} \in \textup{negl}(n).
	$$
	The number of errors of weight $t$ that reach some syndrome is  with an overwhelming probability equal to its expectation $\binom{n}{t}(q-1)^{t}/q^{n-k}$ up to an additive factor $\left(\binom{n}{t}(q-1)^{t}/q^{n-k}\right)^{3/4}$ which is exponentially small with respect to $\binom{n}{t}(q-1)^{t}/q^{n-k}$.

	\section{Expected Minimum Distance of Codes}\label{sec:minDistance}

	We are now interested in computing the expected minimum distance of a random code. As we will see it is given by the so-called {\em Gilbert-Varshamov} distance. This result is for cryptographic purposes very important. Suppose that one finds in a code a word of Hamming weight much smaller than it is expected. This would mean that the code is peculiar
	and maybe even worse, this codeword of small weight may reveal some secret information.

	In view of the foregoing, it may be tempting to say that the expected minimum distance of a random $\lbrack n,k \rbrack_{q}$-code is given by the largest $t$ such that $\sum_{\ell \leq t}\frac{\binom{n}{\ell}(q-1)^{\ell}}{q^{n-k}} \leq 1$ and that is indeed what happens. It turns out that this value of $t$ plays an important role in coding theory and is known as the Gilbert-Varshamov distance.

	\begin{definition}[Gilbert-Varshamov Distance] Let $k\leq n$ and $q$ be integers. The Gilbert-Varshamov distance $t_{\textup{GV}}(q,n,k)$ is defined as the largest integer such that
		$$
		\sum_{\ell=0}^{t_{\textup{GV}}(q,n,k)} \binom{n}{\ell}(q-1)^{\ell} \leq q^{n-k}.
		$$
	\end{definition}

	\begin{remark}
		The Gilbert-Varshamov distance gives the maximum $r$ such that the volume of the ball of radius $r$ is smaller than the inverse density of any $\lbrack n,k \rbrack_{q}$-code. Its analogue for lattices (in the context of lattice-based cryptography) is called the Gaussian heuristic.  
	\end{remark} 

	It can be verified  (by using Lemma \ref{lemma:asymptSphere})
	$$
	\frac{t_{\textup{GV}}(q,n,Rn)}{n} = \tau^{-} + o(1)
	$$
	where $\tau^{-}$ is defined in Equation \eqref{eq:tau-+}. It explains why $\tau^{-}$ is commonly called the {\em relative Gilbert-Varshamov distance.} In the following proposition we show that the minimum distance of almost all $\lbrack n,k \rbrack_{q}$-codes is given by $\tau^{-}$. Interestingly, the proof that $d_{\textup{min}}(\CC)/n>\tau^{-}$ happens with a negligible probability relies on Proposition \ref{propo:ST} that used the second moment technique.

	\begin{proposition}
		Let $\varepsilon > 0$. We have
		$$
		\mathbb{P}_{\vec{H}}\left( (1-\varepsilon)\tau^{-} < \frac{d_{\textup{min}}(\CC)}{n} < (1+\varepsilon)\tau^{-}\right) \geq 1 - q^{- \alpha n(1 + o(1))}
		$$
	where $\alpha \eqdef \min \left( (1-R) - h_{q}\left( (1+\varepsilon)\tau^{-}\right) ,h_{q}\left((1-\varepsilon)\tau^{-}\right) - (1-R)\right) > 0$.
	\end{proposition}

	\begin{proof}
		First notice that,
		\begin{multline}\label{eq:TUP}
			\mathbb{P}_{\vec{H}}\left(\frac{d_{\textup{min}}(\CC)}{n} \notin \left((1-\varepsilon)\tau^{-},(1+\varepsilon)\tau^{-}\right) \right) \leq \mathbb{P}_{\vec{H}}\left( \frac{d_{\textup{min}}(\CC)}{n} \leq (1-\varepsilon)\tau^{-} \right) \\+ \mathbb{P}_{\vec{H}}\left( \frac{d_{\textup{min}}(\CC)}{n} \geq (1+\varepsilon)\tau^{-} \right).
		\end{multline}
	Let us upper-bound independently both terms of the above inequation. First,
	\begin{align}
		\mathbb{P}_{\vec{H}}\left( \frac{d_{\textup{min}}(\CC)}{n} \leq (1-\varepsilon)\tau^{-} \right) &= \mathbb{P}_{\vec{H}}\left( \exists \vec{x} \mbox{ } : \mbox{ } \frac{|\vec{x}|}{n}\leq (1-\varepsilon)\tau^{-} \mbox{ and } \vec{x}\in\CC \right)  \nonumber  \\
		&\leq \sum_{\substack{\vec{x}: \\ \frac{|\vec{x}|}{n} \leq (1-\varepsilon)\tau^{-}}} \mathbb{P}_{\vec{H}}\left( \vec{x}\in\CC\right)  \nonumber \\
		&= \sum_{\ell = 0}^{(1-\varepsilon)n\tau^{-}} \frac{\binom{n}{\ell}(q-1)^{\ell} }{q^{n-k}} \quad (\mbox{By Lemma \ref{lemma:belongCode}.}) \nonumber  \\
		&= \frac{q^{ n\left(h_{q}\left((1-\varepsilon)\tau^{-}\right)+o(1)\right)}}{q^{n-k}}\label{eq:uppBoundTau-} \\
		&= q^{n\left( h_{q}\left( (1-\varepsilon)\tau^{-}\right) - (1-R) + o(1)\right)} \label{eq:FUPB} 
	\end{align} 
	where we used in \eqref{eq:uppBoundTau-} Lemma \ref{lemma:asymptSphere} and the fact that $x \mapsto h_{q}(x)$ is an increasing function for $x \in \left[0,\tau^{-}\right]\subseteq \left[0,\frac{q-1}{q} \right]$.

	Let us now upper-bound the second term of \eqref{eq:TUP}. Let $\delta \in (0,\varepsilon)$ and $u \eqdef (1+\delta) n\tau^{-}$. Notice now that,
	\begin{align} 
	 \mathbb{P}_{\vec{H}}\left( \frac{d_{\textup{min}}(\CC)}{n} \geq (1+\varepsilon)\tau^{-} \right) &\leq \mathbb{P}_{\vec{H}}\left( N_{u}(\CC,\vec{0}) = 0 \right) \nonumber \\
	 &\leq \mathbb{P}_{\vec{H}}\left( \left| N_{u}(\CC,\vec{0}) - \frac{\binom{n}{u}(q-1)^{u}}{q^{n-k}} \right| \geq  \frac{\binom{n}{u}(q-1)^{u}}{q^{n-k}}  \right) \nonumber \\
	 &\leq (q-1)\frac{q^{n-k}}{\binom{n}{u}(q-1)^{u}} \label{eq:UPB1} 
	\end{align} 
	where in the last line we used Proposition \ref{propo:ST} by setting $a=\frac{\binom{n}{u}(q-1)^{u}}{q^{n-k}}$. But now by using Lemma \ref{lemma:asymptSphere} and that $u= (1+\delta)\tau^{-}$ we obtain:
	$$
	\frac{q^{n-k}}{\binom{n}{u}(q-1)^{u}} = q^{n\left( 1 -R - h_{q}\left((1+\delta)\tau^{-}\right) + o(1)\right)}.
	$$
	By plugging this in Equation \eqref{eq:UPB1} we obtain for any $\delta \in (0,\varepsilon)$: 
	$$
	\mathbb{P}_{\vec{H}}\left( \frac{d_{\textup{min}}(\CC)}{n} \geq (1+\varepsilon)\tau^{-} \right) \leq q^{n\left(1-R - h_{q}\left((1+\delta)\tau^{-}\right) + o(1)\right)}.
	$$
	Therefore, by letting $\delta \rightarrow \varepsilon$ we get:
	\begin{equation}\label{eq:SUPB}
		\mathbb{P}_{\vec{H}}\left( \frac{d_{\textup{min}}(\CC)}{n} \geq (1+\varepsilon)\tau^{-} \right) \leq q^{n\left(1-R - h_{q}\left((1+\varepsilon)\tau^{-}\right) + o(1)\right)}.
	\end{equation}  
	To conclude the proof it remains to put together Equations \eqref{eq:FUPB} and \eqref{eq:SUPB} in Equation \eqref{eq:TUP}. 
	\end{proof}

	\begin{remark} 
		In coding theory, the relative Gilbert-Varshamov distance $\tau^{-}$ is known as a ``lower-bound'', there exists a family of $\lbrack n,Rn \rbrack_{q}$-codes with a relative minimum distance $\geq \tau^{-}$. What we have actually proven is that almost all families of $\lbrack n,Rn \rbrack_{q}$-codes have asymptotically a relative minimum distance $=\tau^{-}$. Let us stress that it does not show that the relative Gilbert-Varshamov distance $\tau^{-}$ is an ``upper-bound'', all families   of $\lbrack n,Rn \rbrack_{q}$-codes are such that their asymptotically relative minimum distance is $\leq \tau^{-}$. This is a widely open conjecture in the case of $\mathbb{F}_{2}$ but which is not true for $\mathbb{F}_{q}$ as long as $q$ is a square $\geq49$. 
	\end{remark} 
	
	{\noindent \bf About the optimality of random codes.} We have seen in \iftoggle{amsbook}{Chapter \ref{chapt:1}}{lecture notes $1$} that balls centred at codewords $\vec{c}\in\CC$ and whose radius is $< \frac{d_{\textup{min}}(\CC)}{2}$ never overlap. This condition over the radius ensures that when an error $\vec{e}$ of Hamming weight smaller than $< \frac{d_{\textup{min}}(\CC)}{2}$ occurs, we are sure that computing the closest codeword from $\vec{c}+\vec{e}$ will outcome $\vec{c}$ (we say that the {\em maximum likelihood decoding} succeeds).  However, for random codes this property can be made even stronger. In that case we can show that $\vec{c}$ is indeed the closest codeword from $\vec{c}+\vec{e}$ (with overwhelming probability) if $\vec{e} \approx d_{\textup{min}}(\CC)$ where $d_{\textup{min}}(\CC) \approx \tau^{-}n$ as we show now in the following proposition (by using Markov's inequality).

	\begin{proposition} Let $t \eqdef  n(1-\varepsilon)\tau^{-}$ for some $\varepsilon > 0$ and let $\eta > 0$. We have
$$
		\mathbb{P}_{\vec{H}}\left( \frac{\sharp  \left\{ \vec{c},\vec{c}'\in \CC, \mbox{ } \vec{e},\vec{e}' \in \mathcal{S}_{t} \mbox{ : } \vec{c} + \vec{e} = \vec{c}'+\vec{e}' \right\}}{q^{k} \binom{n}{t}(q-1)^{t}} \geq 1+\eta \right) \leq \frac{1}{\eta} \; \frac{\binom{n}{t}(q-1)^{t}}{q^{n-k}} = \frac{1}{\eta} \; q^{-\alpha n(1+o(1))}
		$$
		where $\alpha \eqdef h_{q}\left( (1-\varepsilon)\tau^{-} \right) - 1+R > 0$. 		
\end{proposition}

	Notice that,
	$$
	\frac{\sharp  \left\{ \vec{c},\vec{c}'\in \CC, \mbox{ } \vec{e},\vec{e}' \in \mathcal{S}_{t} \mbox{ : } \vec{c} + \vec{e} = \vec{c}'+\vec{e}' \right\}}{q^{k} \binom{n}{t}(q-1)^{t}} = 1
	$$
	means that there are no collisions between noisy codewords with errors of Hamming weight $t$ (and that balls of radius $t$ and centered at codewords do not overlap). Therefore, the above proposition shows that for random codes, the maximum likelihood decoding will succeed for almost all noisy codewords, up to the Gilbert-Varshamov distance if $\eta$ is chosen sufficiently small. However, once again, we could be more accurate by using the second moment technique with Bienaym\'{e}-Tchebychev's inequality.

	\begin{proof}
		Let, 
		$$
		Z \eqdef  \sharp  \left\{ \vec{c},\vec{c}'\in \CC, \mbox{ } \vec{e},\vec{e}' \in \mathcal{S}_{t} \mbox{ : } \vec{c} + \vec{e} = \vec{c}'+\vec{e}' \right\}. 
		$$ 
		We have the following computation
		\begin{align*}
			Z &= \sum_{\vec{c}\in \CC,\vec{e}\in\mathcal{S}_{t}} 1 + \sum_{\substack{\vec{c},\vec{c}'\in\CC,\vec{e},\vec{e}'\in\mathcal{S}_{t}\\ (\vec{c},\vec{e}) \neq (\vec{c}',\vec{e}') \\ \vec{c}+\vec{e} = \vec{c}'+\vec{e}'}} 1\\
			&= q^{k}\binom{n}{t}(q-1)^{t} + q^{k} \sum_{\substack{\vec{e},\vec{e}'\in\mathcal{S}_{t} \\ \vec{e} \neq \vec{e}' : (\vec{e} - \vec{e}')\transpose{\vec{H}}=\mathbf{0}}} 1 \\
			&= q^{k}\binom{n}{t}(q-1)^{t}\left( 1 + \frac{1}{\binom{n}{t}(q-1)^{t}} \sum_{\substack{\vec{e},\vec{e}'\in\mathcal{S}_{t} \\ \vec{e} \neq \vec{e}' : (\vec{e} - \vec{e}')\transpose{\vec{H}}=\mathbf{0}}} 1\right).
		\end{align*}
	Let,
	$$
	X \eqdef \frac{1}{\binom{n}{t}(q-1)^{t}}\sum_{\substack{\vec{e},\vec{e}'\in\mathcal{S}_{t} \\ \vec{e} \neq \vec{e}' : (\vec{e} - \vec{e}')\transpose{\vec{H}}=\mathbf{0}}} 1
	$$
	By using Lemma \ref{lemma:belongCode},
	\begin{align*}
		\mathbb{E}_{\vec{H}}\left( X \right) &=  \frac{1}{\binom{n}{t}(q-1)^{t}}\sum_{\substack{\vec{e},\vec{e}'\in\mathcal{S}_{t} \\ \vec{e} \neq \vec{e}'}} \frac{1}{q^{n-k}} \\
		&\leq  \frac{1}{\binom{n}{t}(q-1)^{t}} \frac{\left( \binom{n}{t}(q-1)^{t} \right)^{2}}{q^{n-k}} \\
		&= \frac{\binom{n}{t}(q-1)^{t}}{q^{n-k}}.
	\end{align*}
	To conclude the proof it is enough to apply Markov's inequality to upper-bound $\mathbb{P}_{\vec{H}}(X > \eta)$.
	\end{proof}

	\section{Uniform Distribution of Syndromes}

	In \iftoggle{amsbook}{Chapter \ref{chapt:1}}{lecture notes $1$} we have seen that the decoding problem is defined (for cryptographic purposes) with some distribution in input, namely $(\vec{H},\vec{x}\transpose{\vec{H}})$ where $\vec{H}\in \Fq^{(n-k)\times n}$ and $\vec{x}\in\mathcal{S}_{t}$ are uniformly distributed. Our aim in what follows is to show that when $t/n \in (\tau^{-},\tau^{+})$ we could replace $\vec{x}\transpose{\vec{H}}$ by $\vec{s}\in\Fq^{n-k}$ being uniformly distributed, without changing our problem. More precisely, we are going to show that distributions $(\vec{H},\vec{x}\transpose{\vec{H}})$ and $(\vec{H},\vec{s})$ are {\em statistically close} under the condition $t/n \in (\tau^{-},\tau^{+})$, showing that for any algorithm solving $\mathsf{DP}(n,q,R,\tau)$ we could choose its inputs as $(\vec{H},\vec{s})$ without changing at much its success probability. This result may seem useless but in some applications where $\tau\in(\tau^{-},\tau^{+})$ it is much more comfortable to consider directly $(\vec{H},\vec{s})$ rather than $(\vec{H},\vec{x}\transpose{\vec{H}})$.

		\begin{proposition} 
		\label{proposition:lhl}
		Let $k \eqdef \lfloor Rn \rfloor$, $t \eqdef \lfloor \tau n \rfloor$ and $\vec{H}\in\Fq^{(n-k)\times n}$,  $\vec{s}\in\mathbb{F}_{q}^{n-k}$, $\vec{e} \in \mathcal{S}_{t}$ being uniformly distributed. We have
		\begin{equation}\label{eq:uppBoundSD}
			 \mathbb{E}_{\vec{H}}\left( \Delta\left( \vec{e}\transpose{\vec{H}}, \vec{s}\right) \right) \leq 
			\frac{1}{2} \; \sqrt{\frac{q^{n-k}-1}{\binom{n}{t}(q-1)^{t}}}. 
		\end{equation}
		In particular, when $\tau\in(\tau^{-},\tau^{+})$
		$$
		 \mathbb{E}_{\vec{H}}\left( \Delta\left( \vec{e}\transpose{\vec{H}}, \vec{s}\right) \right) \leq q^{-\alpha n(1+o(1))} \quad \mbox{where} \quad \alpha \eqdef \frac{1}{2}\left( h_{q}(\tau) - 1 + R \right) > 0.
		$$
		\end{proposition}

		\begin{exercise}
			Let $\vec{H}\in\Fq^{(n-k)\times n}$ being uniformly distributed,  $\vec{s}\in\mathbb{F}_{q}^{n-k}$, $\vec{e} \in \mathcal{S}_{t}$ be some random variables. Show that 
			\begin{align*} 
			\mathbb{E}_{\vec{H}}\left( \Delta\left( \vec{e}\transpose{\vec{H}},\vec{s} \right)\right) 
			&= \frac{1}{q^{(n-k)\times n}}\sum_{\vec{H}_{0}\in\Fq^{(n-k)\times n}} \Delta\left( \vec{e}\transpose{\vec{H}}_{0},\vec{s}\right)
			\end{align*} 
		\end{exercise}

		The proof of Proposition \ref{proposition:lhl} relies on the following lemma which is a rewriting in our context of the result known as the {\em left over hash lemma}. Roughly speaking, it shows that if the outputs of a function collide with a probability $\varepsilon$-close to the case where they would be randomly distributed, then this function is $\sqrt{\varepsilon}$-close to a random function.

		\begin{lemma} \label{lemme:leftOver} Let $\mathcal{H} = (h_i)_{i \in I}$ be a finite family of applications from $E$ in $F$. Let  $\varepsilon$ be the ``collision bias'' 
			\begin{displaymath}
				\mathbb{P}_{h,e,e'}(h(e)=h(e')) = \frac{1}{\sharp F} (1 + \varepsilon)
			\end{displaymath}
			where $h$ is uniformly drawn in $\mathcal{H}$, $e$ and $e'$ be distributed according to some random variable $X$ taking its values $E$. Let $\mathcal{U}$ be the uniform distribution over $F$ and $\mathcal{D}(h)$ be the distribution $h(e)$ when $e$ is distributed according to $X$. We have,
			\begin{displaymath}
			 \mathbb{E}_{h}\left( \Delta(\mathcal{D}(h),\mathcal{U}) \right) \leq \frac{1}{2} \; \sqrt{\varepsilon}.
			\end{displaymath}
		\end{lemma}
	
		\begin{proof}
By definition of the statistical distance we have
			\begin{align}				
				\mathbb{E}_{h}\left( \Delta(\mathcal{D}(h),\mathcal{U})\right) &= \sum_{h\in \mathcal{H}} \frac{1}{\sharp \mathcal{H}} \; \Delta\left( \mathcal{D}(h),\mathcal{U}\right)  \\
&= \frac{1}{2}  \sum_{h \in \mathcal{H}} \frac{1}{\sharp \mathcal{H}} \sum_{f \in F} \left| \mathbb{P}_e(h(e)=f) - \frac{1}{\sharp F} \right| \nonumber \\
				& =  \frac{1}{2} \sum_{(h,f) \in \mathcal{H} \times F}   \left| \mathbb{P}_{h_0,e}(h_0=h,h_0(e)=f) - \frac{1}{\sharp \mathcal{H} \; \sharp F} \right| \nonumber \\
				& =  \frac{1}{2} \sum_{(h,f) \in \mathcal{H} \times F}   \left| q_{h,f} - \frac{1}{\sharp \mathcal{H}  \; \sharp F} \right|.
			\end{align}
			where $q_{h,f} \eqdef \mathbb{P}_{h_{0},e}\left( (h_{0},h_{0}(e)) = (h,f) \right)$ with $h_{0}$ being uniformly chosen at random in $\mathcal{H}$ and $e$ be distributed according to $X$.
Using the Cauchy-Schwarz inequality, we obtain 
			\begin{equation} \label{eq:cauchy-schwarz}
				\sum_{(h,f) \in \mathcal{H} \times F}   \left| q_{h,f} - \frac{1}{\sharp \mathcal{H} \; \sharp F} \right| \leq  \sqrt{\sum_{(h,f) \in \mathcal{H} \times F}
					\left( q_{h,f} - \frac{1}{\sharp \mathcal{H} \; \sharp F}\right)^2} \; \sqrt{\sharp \mathcal{H} \; \sharp F}.
			\end{equation}
			Let us observe now that
			\begin{align}
				\sum_{(h,f) \in \mathcal{H} \times F}
				\left( q_{h,f} - \frac{1}{\sharp \mathcal{H} \; \sharp F}\right)^2 & =  \sum_{h,f}
				\left( q_{h,f}^2 - 2 \frac{q_{h,f}}{\sharp \mathcal{H} \; \sharp F}+ \frac{1}{\sharp \mathcal{H}^2 \; \sharp F^2}\right) \nonumber \\
				& = \sum_{h,f} q_{h,f}^2 -2 \frac {\sum_{h,f} q_{h,f}}{\sharp \mathcal{H} \; \sharp F} + \frac{1}{\sharp \mathcal{H} \; \sharp F} \nonumber \\
				& = \sum_{h,f} q_{h,f}^2  - \frac{1}{\sharp \mathcal{H} \; \sharp F}  \label{eq:q}.
			\end{align}
			Consider for $i \in \{0,1\}$ independent random variables $h_i$ and $e_i$ 
			that are  drawn uniformly at random in $\mathcal{H}$ and according to $X$ respectively.
			We continue this computation by noticing now that
			\begin{align}
				\sum_{h,f} q_{h,f}^2 & =  \sum_{h,f} \mathbb{P}_{h_0,e_0} (h_0 = h,h_0(e_0)=f) \mathbb{P}_{h_1,e_1} (h_1 = h,h_1(e_1)=f)   \nonumber \\
				& = \mathbb{P}_{h_0,h_1,e_0,e_1} \left(h_0=h_1,h_0(e_0)=h_1(e_1)\right) \nonumber  \\
				& = \frac{\mathbb{P}_{h_0,e_0,e_1}\left(h_0(e_0)=h_0(e_1)\right) }{\sharp \mathcal{H}}\nonumber \\
				& = \frac{1+\varepsilon }{\sharp \mathcal{H} \; \sharp F}.\label{eq:collision}
			\end{align}
			By substituting for $\sum_{h,f} q_{h,f}^2$ the expression obtained in \eqref{eq:collision} into
			\eqref{eq:q} and then back into \eqref{eq:cauchy-schwarz} we finally obtain
			\begin{align*} 
			\sum_{(h,f) \in \mathcal{H} \times F}   \left| q_{h,f} - \frac{1}{\sharp \mathcal{H} \; \sharp  F} \right|&\leq  \sqrt{\frac{1+\varepsilon}{\sharp \mathcal{H} \; \sharp F}
				- \frac{1}{\sharp \mathcal{H} \; \sharp F}}\; \sqrt{\sharp \mathcal{H} \; \sharp F}\\
			& = \sqrt{\frac{\varepsilon}{\sharp \mathcal{H} \; \sharp F}} \; \sqrt{\sharp \mathcal{H} \; \sharp F}\\
			& = \sqrt{\varepsilon}.
			\end{align*} 
		\end{proof}
	
	We are now ready to prove Proposition \ref{proposition:lhl}.
	
	\begin{proof}[Proof of Proposition \ref{proposition:lhl}]
	Let us compute the ``collision bias'' in our case. We have: 
	\begin{align*}
		\mathbb{P}_{\vec{H},\vec{e},\vec{e}'}\left( \vec{e}\transpose{\vec{H}} = \vec{e}'\transpose{\vec{H}} \right) &= \mathbb{P}_{\vec{H},\vec{e},\vec{e}'}\left( (\vec{e} - \vec{e}')\transpose{\vec{H}} = \mathbf{0} \right) \\ 
		&= \mathbb{P}_{\vec{H},\vec{e},\vec{e}'}\left( (\vec{e}-\vec{e}')\transpose{\vec{H}} = \mathbf{0} \mid \vec{e} \neq \vec{e}' \right) \mathbb{P}\left( \vec{e} \neq \vec{e}'\right) + \mathbb{P}_{\vec{e},\vec{e}'}(\vec{e} = \vec{e}') \\ 	
		&= \frac{1}{q^{n-k}}\left( 1 - \frac{1}{\binom{n}{t}(q-1)^{t}} \right) + \frac{1}{\binom{n}{t}(q-1)^{t}} \quad (\mbox{by Lemma \ref{lemma:belongCode}}) \\
		&= \frac{1}{q^{n-k}}\left( 1 + \varepsilon \right) \quad \mbox{where} \quad  \varepsilon \eqdef \frac{q^{n-k}-1}{\binom{n}{t}(q-1)^{t}}	
	\end{align*} 
	which shows \eqref{eq:uppBoundSD} by using Lemma \ref{lemme:leftOver}. The second part of the proposition easily follows from the definition of $\tau^{-}$ and $\tau^{+}$.  
	\end{proof}

	\begin{exercise}
		Show that if one replaces in the binary case ($q=2$) $\vec{e}$ in Proposition  \ref{proposition:lhl} by $\vec{e}^{\textup{Ber}}$, where the $e_{i}^{\textup{Ber}}$'s are distributed according to a Bernoulli distribution of parameter $\tau$, we would obtain
		$$
		\mathbb{E}_{\vec{H}}\left( \Delta\left(\vec{e}^{\textup{Ber}}\transpose{\vec{H}},\vec{s}  \right) \right) \leq \frac{1}{2}\sqrt{2^{-k}\left(1 + (1-2\tau)^{2} \right)^{n}}. 
		$$ 
		What can you deduce when comparing both results with $\vec{e}$ or $\vec{e}^{\textup{Ber}}$? What is (according to Proposition \ref{proposition:lhl}) the ``best'' choice of error $\vec{x}$ to ensure that $\vec{x}\transpose{\vec{H}}$ is uniformly distributed?
		
	\end{exercise}

	\begin{exercise}
		Let $\CC$ be a fixed $\lbrack n,k \rbrack_{q}$-code of parity-check matrix $\vec{H}$ and $\vec{y},\vec{s},\vec{e}\in\Fq^{n}\times\Fq^{n-k}\times \mathcal{S}_{t}$ be uniformly distributed. Our aim in this exercise is to show that $\Delta(\vec{c}+\vec{e},\vec{y}) = \Delta(\vec{e}\transpose{\vec{H}},\vec{s})$.
		\begin{itemize}
			\item[1.] Given $\vec{s} \in\Fq^{n-k}$, let $\vec{y}(\vec{s})\in\Fq^{n}$ be such that $\vec{y}(\vec{s})\transpose{\vec{H}} = \vec{s}$. Show that
			$$
			\sum_{\vec{y}\in\Fq^{n}} \left| \mathbb{P}_{\vec{e},\vec{c}}( \vec{c}+\vec{e} = \vec{y} ) - \frac{1}{q^{n}}  \right| = \sum_{\vec{s}\in\Fq^{n-k}} \sum_{\vec{c}'\in\CC} \left| \mathbb{P}_{\vec{e},\vec{c}}(\vec{c}+\vec{e} = \vec{y}(\vec{s}) + \vec{c}') - \frac{1}{q^{n}} \right|.
			$$

			\item[2.] Deduce that  $\Delta(\vec{c}+\vec{e},\vec{y}) = \Delta(\vec{e}\transpose{\vec{H}},\vec{s})$.
		\end{itemize}
	\end{exercise}

	Up to now, we have proven in Proposition \ref{proposition:lhl} that syndromes $\vec{e}\transpose{\vec{H}}$ are statistically close to the uniform distribution when $\vec{e}$ is picked uniformly at random in $\mathcal{S}_{t}$ with $t$ larger than the Gilbert-Varshamov distance (more precisely, when $t/n \in (\tau^{-},\tau^{+})$)   {\em but} in average over $\vec{H}$. One may ask if the result still holds for a fixed matrix $\vec{H}_{0}$? Actually we can prove, without too much effort, that the above result is true for almost all matrices but with a loss given by a square root as shown by the following proposition.

	\begin{proposition}
			Let $k \eqdef \lfloor Rn \rfloor$, $t \eqdef \lfloor \tau n \rfloor$, $\vec{H}\in\Fq^{(n-k)\times n}$ being uniformly distributed and  $\vec{s}\in\mathbb{F}_{q}^{n-k}$, $\vec{e} \in\Fq^{n}$ be some random variables. 
Suppose that
			$$
			\mathbb{E}_{\vec{H}}\left( \Delta\left( \vec{e} \transpose{\vec{H}},\vec{s} \right) \right) \leq \varepsilon
			$$ 
			Then, we have
			$$
			\frac{\sharp \left\{ \vec{H}_{0}\in\Fq^{(n-k)\times n} \; : \;\Delta\left( \vec{e}\transpose{\vec{H}}_{0},\vec{s}\right) \geq \sqrt{\varepsilon} \right\}}{q^{(n-k)\times n}} \leq \sqrt{\varepsilon}
			$$
	\end{proposition}

	\begin{proof}
		First, 
		\begin{equation}\label{eq:exp} 
		\mathbb{E}_{\vec{H}}\left( \Delta\left( \vec{e}\transpose{\vec{H}},\vec{s} \right)\right) 
		= \frac{1}{q^{(n-k)\times n}}\sum_{\vec{H}_{0}\in\Fq^{(n-k)\times n}} \Delta\left( \vec{e}\transpose{\vec{H}}_{0},\vec{s}\right).
		\end{equation}
		The idea of the proof is to split in two parts this sum according to the terms that are larger and smaller than $\sqrt{\varepsilon}$
		\begin{align}
		\sum_{\vec{H}_{0}\in\Fq^{(n-k)\times n}} \Delta\left( \vec{e}\transpose{\vec{H}}_{0},\vec{s}\right) &= \sum_{\substack{\vec{H}_{0} \; : \\ \Delta(\vec{e}\transpose{\vec{H}}_{0},\vec{s}) < \sqrt{\varepsilon}}}  \Delta\left( \vec{e}\transpose{\vec{H}}_{0},\vec{s}\right) + \sum_{\substack{\vec{H}_{0} \; : \\ \Delta(\vec{e}\transpose{\vec{H}}_{0},\vec{s}) \geq \sqrt{\varepsilon}}}  \Delta\left( \vec{e}\transpose{\vec{H}}_{0},\vec{s}\right)\nonumber\\
		&\geq \sum_{\substack{\vec{H}_{0} \; : \\ \Delta(\vec{e}\transpose{\vec{H}}_{0},\vec{s}) < \sqrt{\varepsilon}}}  \Delta\left( \vec{e}\transpose{\vec{H}}_{0},\vec{s}\right) + \sum_{\substack{\vec{H}_{0} \; : \\ \Delta(\vec{e}\transpose{\vec{H}}_{0},\vec{s}) \geq \sqrt{\varepsilon}}}  \sqrt{\varepsilon}\nonumber\\
		&\geq  \sum_{\substack{\vec{H}_{0} \; : \\ \Delta(\vec{e}\transpose{\vec{H}}_{0},\vec{s}) \geq \sqrt{\varepsilon}}}  \sqrt{\varepsilon}\nonumber\\
		&= \sqrt{\varepsilon} \; \sharp \left\{ \vec{H}_{0}\in\Fq^{(n-k)\times n} \; : \;\Delta\left( \vec{e}\transpose{\vec{H}}_{0},\vec{s}\right) \geq \sqrt{\varepsilon} \right\}\label{eq:lwB}
		\end{align}
	By plugging Equation \eqref{eq:lwB} in \eqref{eq:exp}, we obtain
	$$
		\mathbb{E}_{\vec{H}}\left( \Delta\left( \vec{e}\transpose{\vec{H}},\vec{s} \right)\right)  \geq  \sqrt{\varepsilon} \; \frac{\sharp \left\{ \vec{H}_{0}\in\Fq^{(n-k)\times n} \; : \;\Delta\left( \vec{e}\transpose{\vec{H}}_{0},\vec{s}\right) \geq \sqrt{\varepsilon} \right\}}{q^{(n-k)\times n}}
	$$
	But by assumption we have $\mathbb{E}_{\vec{H}}\left( \Delta\left( \vec{e}\transpose{\vec{H}},\vec{s} \right)\right)\leq \varepsilon$. It concludes the proof. 
	\end{proof} 

 	\chapter{Information Set Decoding Algorithms}\label{chapt:3}

	\section*{Introduction}

	The aim of any code-based cryptosystem is to rely its security on the hardness of the decoding problem when the input code is random. It is therefore crucial to study the best algorithms, usually called {\em generic} decoding algorithms, for solving this problem. Despite many efforts on this issue the best ones \cite{BJMM12,MO15,BM17}  are exponential in the number of errors that have to be corrected and all can be viewed as a refinement of the original Prange's algorithm \cite{P62}. They are actually referred to as {\em Information Set Decoding} (ISD). The aim of these lecture notes is to describe the ``first'' ISD algorithms. Our description will be mainly algorithmic with the use of parity-check matrices. However in this (too long) introduction we try to give another point of view, more related to the inherent mathematical structure of codes: linear subspaces of some $\Fq^{n}$.
	\newline

	{\bf \noindent Prange's approach: use linearity.} Given an $\lbrack n,k \rbrack_{q}-$code $\CC$ and one word $\vec{y}\in\Fq^{n}$, we are looking for some codeword $\vec{c}\in\CC$ at distance $t$ from $\vec{y}$, namely $|\vec{y}-\vec{c}| = t$. Here $\CC$ is defined as a {\em linear} subspace of $\Fq^{n}$ of dimension $k$. Therefore one can check that there exists some set of positions $\mathcal{I} \subseteq \llbracket 1,n \rrbracket$ of size $k$, {\em called an information set}, which uniquely determines every codewords, more precisely
	$$
	\forall \vec{x} \in \F_q^{k}: \quad  \exists! \vec{c} \in \CC \mbox{ (that we can easily compute by linear algebra) such that } \cv_{\mathcal{I}} = \xv
	$$
	where $\vec{c}_{\mathcal{I}}$ denotes the vector whose coordinates are those of $\vec{c} = (c_{i})_{1\leq i \leq n}$ which are indexed by $\mathcal{I}$, {\em i.e.} $\vec{c}_{\mathcal{I}} = (c_{i})_{i \in\mathcal{I}}$.

	\begin{exercise}
		Let $\CC$ be an $\lbrack n,k \rbrack_{q}$-code and $\mathcal{I} \subseteq \llbracket 1,n \rrbracket$ be of size $k$. Show that,
		\begin{align*} 
		\mathcal{I} \mbox{ is an information set for $\CC$ } &\iff \forall \vec{G} \mbox{ generator matrix of $\CC$}, \;  \vec{G}_{\mathcal{I}} \mbox{ is invertible } \\  
		&\iff \forall \vec{H} \mbox{ parity-check matrix of $\CC$}, \; \vec{H}_{\overline{\mathcal{I}}} \mbox{ is invertible }  
		\end{align*} 
	where given $\vec{M} \in \Fq^{r \times n}$, $\vec{M}_{\mathcal{I}}$ denotes the matrix whose {\em columns} are those of $\vec{M}$ which are indexed by $\mathcal{I}$.
	\end{exercise}

	Prange's idea to recover some solution $\vec{c}^{\textup{sol}}\in \CC$, where $\vec{y} = \vec{c}^{\textup{sol}} + \vec{e}^{\textup{sol}}$ and $|\vec{e}^{\textup{sol}}| =t$, is as follows. First we pick some random information set $\mathcal{I}$ and we hope that it contains no error positions, namely:
	\begin{equation} \label{eq:ISDidea}
\vec{c}_{\mathcal{I}}^{\textup{sol}} = \vec{y}_{\mathcal{I}} \quad \left(\iff \vec{e}^{\textup{sol}}_{\mathcal{I}} = \vec{0}\;\right).
	\end{equation} 
	If this is true, we are done. It remains to compute {\em the unique} codeword $\vec{c}$ such that $\vec{c}_{\mathcal{I}} = \vec{y}_{\mathcal{I}}$ as by uniqueness we get $\vec{c} = \vec{c}^{\textup{sol}}$. In other words, Prange's idea (when looking for a close codeword) simply consists in picking some information set $\mathcal{I}$, computing the unique codeword $\vec{c}$ equal to $\vec{y}$ on those coordinates, and then to check 
	if the constraint \eqref{eq:ISDidea} is verified, namely 
	if $|\vec{y} - \vec{c}| = t$. The average number of times we have to pick a set $\mathcal{I}$ in Prange's algorithm until finding a solution is therefore given by $1/\pr$ where $\pr$ is the probability that Equation \eqref{eq:ISDidea} is verified. As we will see, $1/\pr$ is 
		\begin{itemize}
		\item polynomial for $t/n = \left(\frac{q-1}{q}\right)\left(1-\frac{k}{n}\right)$,

		\item exponential in $t$ when $t/n \in \left\rbrack 0,\left(\frac{q-1}{q}\right)\left(1-\frac{k}{n}\right)\right\lbrack$
	\end{itemize} 
	
	as long as $n \rightarrow + \infty$ and $k/n$ is some constant. Interestingly, all improvements of Prange's algorithm, since sixty years, have the same behaviour with respect to $t/n$. Even though Prange's algorithm is quite naive, it really shows where decoding is easy, where it is not.

	Let us now describe how these improvements, in particular ISD algorithms, were obtained as they start from the same key idea. For this let us back up a bit.  
	\newline

	{\bf \noindent Dumer's approach: a collision search.} The simplest way to find a codeword $\vec{c}$ at distance $t$ from $\vec{y}$ is basically enumerating all the errors $\vec{e}$ of weight $t$ until finding one that reaches $\vec{y}-\vec{e}\in \CC$. This naive approach will obviously cost 
	$
	\binom{n}{t}(q-1)^{t}. 
	$  
	However, taking advantage of the birthday paradox, this exhaustive enumeration can be improved as Dumer showed \cite{D86}. Dumer's idea was to notice that if one splits in two parts\footnote{To simplify the presentation, the cut is explained by taking the first $\frac{n-k}{2}$ positions for the first part and the other $\frac{n-k}{2}$ else for the second part, but of course in general these two sets of positions are randomly chosen.} of the same size some parity-check matrix $\vec{H} = \begin{pmatrix}
		\vec{H}_{1} & \vec{H}_{2} 
	\end{pmatrix}$ of $\CC$, 
then solving the decoding problem boils down to finding $\vec{e}_{1}$ and $\vec{e}_{2}$ of Hamming weight $t/2$ such that $\vec{H}_{1}\transpose{\vec{e}}_{1} + \vec{H}_{2}\transpose{\vec{e}}_{2} = \vec{H}\transpose{\vec{y}}$. A natural strategy to compute a solution $\vec{e} = (\vec{e}_{1},\vec{e}_{2})$ reduces to compute the following lists of $\Fq^{n-k}$
$$
\mathcal{L}_{1} \eqdef \left\{  \vec{H}_{1}\transpose{\vec{e}}_1 : |\vec{e}_{1}| = \frac{t}{2} \right\} \quad \mbox{and} \quad \mathcal{L}_{2} \eqdef \left\{-\vec{H}_{2}\transpose{\vec{e}}_2 + \vec{H}\transpose{\vec{y}} : |\vec{e}_{2}| = \frac{t}{2} \right\}
$$
and then to compute their ``collision''  
$$
\mathcal{L}_{1} \bowtie \mathcal{L}_{2} \eqdef \left\{  (\vec{e}_{1},\vec{e}_{2}) \in \mathcal{L}_{1} \times \mathcal{L}_{2},\quad  \vec{H}_{1}\transpose{\vec{e}}_1 =  - \vec{H}_{2}\transpose{\vec{e}}_2 + \vec{H}\transpose{\vec{y}}  \right\}.
$$
This new list trivially leads to solutions of the decoding problem. However, what is the cost of this procedure? By using classical techniques such as hash tables or sorting lists, computing $\mathcal{L}_{1}\bowtie\mathcal{L}_{2}$ costs, up to a polynomial factor, $\max(\sharp \mathcal{L}_{1},\sharp \mathcal{L}_{2}) + \sharp  \left( \mathcal{L}_{1}\bowtie \mathcal{L}_{2}\right)$. Notice now that both lists $\mathcal{L}_{1}$ and $\mathcal{L}_{2}$ have the same size, namely
$$
\binom{n/2}{t/2}(q-1)^{t/2} = \widetilde{O}\left( \sqrt{\binom{n}{t}(q-1)^{t}} \right).
$$
To estimate the cost of this procedure it remains to estimate the size of $\mathcal{L}_{1}\bowtie\mathcal{L}_{2}$. One can check that Dumer's approach finds all the solutions of the decoding problem (given by $\approx\max\left(1,\binom{n}{t}(q-1)^{t}/q^{n-k}\right)$ as shown in \iftoggle{amsbook}{Chapter \ref{chapt:2}}{lecture notes $2$}) but up to some polynomial loss, given by the probability that a solution $\vec{e}$ is not split into two equal parts. Then, $\mathcal{L}_{1}\bowtie\mathcal{L}_{2}$ is the set of solution(s) of the considered decoding problem. To summarize,  Dumer's approach enables to roughly find (up to polynomial factors)
\begin{equation} \label{eq:Dum} 
	\max\left( 1,\frac{\binom{n}{t}(q-1)^{t}}{q^{n-k}} \right) \quad \mbox{solutions in time} \quad \sqrt{\binom{n}{t}(q-1)^{t}} + \frac{\binom{n}{t}(q-1)^{t}}{q^{n-k}}.
\end{equation} 
Notice in the case where $t$ is equal to the Gilbert-Varshamov distance, namely when $\binom{n}{t}(q-1)^{t} \approx q^{n-k}$, Dumer's algorithm has a quadratic gain compared to the exhaustive search. However, it is even better, as shown by the following proposition.

\begin{proposition} \label{propo:compDumPr} The running time of Prange's algorithm for solving $\mathsf{DP}(n,q,R,\tau)$ when $\tau = h_{q}^{-1}(1-R)$\footnote{The relative Gilbert-Varshamov distance.} and $R \rightarrow 1$ is given by:
	$$
	q^{n\;(1-R)(1+o(1))}
	$$
	while Dumer's algorithm will cost:
	$$
	q^{n\;\frac{1-R}{2}(1+ o(1)) }. 
	$$
\end{proposition}

Dumer's algorithm has therefore a quadratic gain over Prange when the code rate tends to one and decoding at the Gilbert-Varhsamov distance. Though, the primary interest of this approach is not here. First, Dumer's algorithm finds (almost) all solutions of the decoding problem even if there are many of them. Furthermore, the distance $t$ can be chosen such that it finds (almost) all of them in {\em amortized time one}.

\begin{definition}[Amortized time one]
	An algorithm that outputs $S$ solutions in time $T$ of some problem is said to be in amortized time one if $S = \frac{T}{P(n)}$ for some polynomial $P$. In the sequel we will always neglect this polynomial factor.
\end{definition}

Dumer's algorithm works in amortized time one when $t$ is beyond the Gilbert-Varshamov bound and verifies: 
\begin{equation}\label{eq:DumerAm1}
\sqrt{\binom{n}{t}(q-1)^{t}} = \frac{\binom{n}{t}(q-1)^{t}}{q^{n-k}} \iff \binom{n}{t}(q-1)^{t} = \left(q^{n-k}\right)^{2}.
\end{equation} 
As we are going to explain, most of the ideas to improve Prange's algorithm were based on these two remarks. The key idea is to reduce the initial decoding problem to a ``denser'' decoding problem where there are an exponential number of solutions but which can be found in amortized time one. 
\newline

{\bf \noindent A mixed approach: ISD.} The key point to improve Prange's algorithm starts from the following idea. Given some set of positions $\mathcal{J} \subseteq \llbracket 1,n\rrbracket$ of size $k+\ell$ where $\ell > 0$, compute first a set $\mathcal{S}$ of {\em decoding candidates} which are some vectors at distance $p$ from the target $\vec{y}$ when their coordinates are restricted to $\mathcal{J}$, namely:
\begin{equation}\label{eq:S}
 \mathcal{S} \subseteq \left\{ \vec{c}_{\mathcal{J}} \mbox{ : }  |\vec{c}_{\mathcal{J}} - \vec{y}_{\mathcal{J}}| = p  \mbox{ and } \vec{c} \in \CC \right\}.
\end{equation} 
Notice that $\mathcal{S}$ is a subset of the solutions of a decoding problem at distance $p$ when it is given as input the target $\vec{y}_{\mathcal{J}}$ and the code
\begin{equation}\label{eq:D}
\mathcal{D} \eqdef \left\{ \vec{c}_{\mathcal{J}} \in \F_q^{k+\ell} :  \vec{c} \in \CC \right\}.
\end{equation}
It turns out that $\mathcal{D}$ is a code known as the {\em punctured} code of $\CC$ at the positions $\overline{\mathcal{J}}$. Its length is $k +\ell$ and its dimension is $k$ if $\mathcal{J}$ is an {\em augmented information set}, namely it contains some information set of $\CC$, which will be assumed in what follows. Under this condition, $\vec{c}_{\mathcal{J}}$ uniquely determines its ``lift'' $\vec{c}\in\CC$ which can be easily computed by linear algebra.

\begin{exercise}
			Let $\CC$ be an $\lbrack n,k \rbrack_{q}$-code and $\mathcal{J} \subseteq \llbracket 1,n \rrbracket$ be of size $k+\ell$. Show that,
	\begin{align*} 
		\mathcal{J} \mbox{ is an augmented information set for $\CC$ } &\iff \mathcal{D} \mbox{ defined in Equation \eqref{eq:D} has dimension $k$}.
	\end{align*} 
\end{exercise}

Now, for the codeword $\vec{c}\in\CC$, such that $\vec{c}_{\mathcal{J}}\in\mathcal{S}$, to be a solution of the original decoding problem, 
it has necessarily to verify 
\begin{equation}\label{eq:ISDidee2}
	\left|\vec{c}_{\overline{\mathcal{J}}} - \vec{y}_{\overline{\mathcal{J}}}\right| = t-p.
\end{equation}
This condition is weaker than of Prange algorithm (see Equation \eqref{eq:ISDidea}): by picking our set $\mathcal{J}$ we do not hope to remove all the errors but only some fraction of it. Furthermore, contrary to Prange's approach we have many decoding candidates for each draw of the augmented information set $\mathcal{J}$. However, notice that smaller is $p$, harder it will be to compute even one decoding candidate. Therefore we cannot reasonably hope to choose $p$ too small if we seek to test many decoding candidates at each draw of $\mathcal{J}$. It also turns out that if $p$ is too small (below the Gilbert-Varshamov bound of the punctured code $\mathcal{D}$) no solutions are expected while on the other hand, if $p$ is just above the Gilbert-Varshamov distance, we expect an exponential number of solutions.

So all in all, we have reduced our problem to decode a code of length $n$ and dimension $k$  to the bet made in \eqref{eq:ISDidee2} {\em and} the computation of $\mathcal{S}$ ({\em i.e.} the decoding candidates) which is nothing else than decoding a ``sub''-code of length $k+\ell$ and dimension $k$. This whole approach is known as Information Set Decoding (ISD). Note that we are completely free to choose our favourite algorithm to compute $\mathcal{S}$.
Each ISD is then ``parametrized'' by the algorithm used as a subroutine for computing this set and,
the better the algorithm, the better the ISD. 
However one may ask our meaning of a ``better'' algorithm for computing $\mathcal{S}$. To understand this let us 
introduce  the probability $\alpha_{p,\ell}$
that a fixed $\vec{c}_{\mathcal{J}}\in\mathcal{S}$ leads to $\vec{c}\in\CC$ which verifies Equation \eqref{eq:ISDidee2}. We will show that the overall probability (after computing $\mathcal{S}$) to get a solution is given by $\approx\min\left(1,\sharp \mathcal{S} \; \alpha_{p,\ell}\right)$. 
It will lead to the following proposition that gives the running time of the whole algorithm  to solve $\DP$\footnote{The following proposition is the equivalent of Proposition \ref{propo:ISD} with the ``noisy codeword'' point of view.}.

\begin{proposition}\label{propo:ISDwkf} Assume that, given a random code $\mathcal{D}$ of length $k+\ell$, dimension $k$ and a target $\vec{z}\in\Fq^{k+\ell}$, we can compute in time $T$ a set of size $S$ of codewords $\vec{d} \in \mathcal{D}$ at distance $p$ from $\vec{z}$. Then, we can solve $\mathsf{DP}(n,q,R,\tau)$ in average time (up to a polynomial factor in $n$)
	\begin{equation}\label{eq:ISD} 
	T \; \max\left(1,\frac{1}{S \; \alpha_{p,\ell}}\right).
	\end{equation} 
\end{proposition}

The overall cost for solving $\mathsf{DP}$ is therefore crucially parametrized by the cost for decoding a code $\mathcal{D}$ of rate $k/(k+\ell)$ at distance $p$, but notice that we need to find $S$ solutions in time $T$ and a priori not only one. If we want to design algorithms achieving this task such that the ISD improves original Prange's algorithm we have first to understand how parameters $p$, $\ell$ and quantities $T$, $S$ interact.

Let us admit that
$
p \mapsto \alpha_{p,\ell}
$
is a decreasing function. Notice now that, the larger $p$, the larger the number of solutions and the easier the decoding of $\mathcal{D}$ at distance $p$. 
Therefore we can reasonably suppose that $p \mapsto T$ is also a decreasing function. These two facts lead to a contradictory situation to minimize the ISD cost, we need to choose $p$ as small as possible for minimizing $1/\alpha_{p,\ell}$ while at the same time we need to choose a large $p$ to decrease $T$. Notice now, as $T\geq S$, that we have
$$
T \; \max\left(1,\frac{1}{S \; \alpha_{p,\ell}}\right) \geq \frac{1}{\alpha_{p,\ell}} 
$$
Therefore we do not really have the choice, to minimize the cost of the ISD we have in the best case to design a sub-routine such that for parameters $p$ and $\ell$ we have above an equality instead of an inequality. 
In particular it shows that our decoding algorithm at distance $p$ (as small as possible) needs to find solutions in amortized time one, {\em i.e.} $S=T$. If this can be done we would get an improvement over Prange. Indeed, we have to remember that $\alpha_{p,\ell}$, the probability that our decoding candidate verifies Equation \eqref{eq:ISDidee2}, is exponentially larger than the probability to verify Equation \ref{eq:ISDidea} as in Prange's algorithm (our bet is weaker).

Our discussion has just shown that it is theoretically possible to improve Prange's algorithm if we succeed, given a code $\mathcal{D}$ of length $k+\ell$, dimension $k$ and any target $\vec{z}$, to compute in amortized time one many codewords $\vec{d}\in\mathcal{D}$ at distance $p$ (as small as possible) from $\vec{z}$.
The fundamental remark here is that $\mathcal{D}$ has a rate given by $k/(k+\ell)\approx 1$ when $\ell$ is not too large. It corresponds exactly to the range of parameters where Dumer's algorithm (that we have described earlier) can decode in amortized time one and can also have a quadratic gain over the original Prange algorithm. 
However parameters $p$ and $\ell$ have to be carefully chosen as in Equation \eqref{eq:DumerAm1} (where replace $t$ by $p$ and $n$ by $k+\ell$). In particular $p$ cannot be chosen too small. Even though this choice of parameters is extremely constrained, the ISD using Dumer's algorithm improves Prange algorithm. But the better was yet to come.
More sophisticated algorithms were designed, enabling to change the balance of parameters between $p$ and $\ell$ by still decoding in amortized time one (in particular decreasing $p/(k+\ell)$ but also increasing $\ell$ to move away the rate $k/(k+\ell)$ from one). In these lecture notes we will restrict our study to the improvement given by the generalized birthday algorithm \cite{W02}. But nowadays there exists far better techniques, such as ``representations technique'' (originally used for solving subset-sum problems) \cite{BJMM12} or nearest neighbours search \cite{MO15,BM17} but this is out of scope of these lecture notes.
\newline

{\bf Basic notation.} Given $\vec{H} \in \Fq^{r\times n}$ and $\mathcal{I}\subseteq \llbracket 1,n \rrbracket$ we will denote by $\vec{H}_{\mathcal{I}}$ the matrix whose {\em columns} are those of $\vec{H}$ which are indexed by $\mathcal{I}$.

During all these lecture notes both $R \in (0,1)$ and the field size $q$ will be supposed to be constants.
\newline

{\bf Described algorithms to solve $\DP$.} We will describe three ISD algorithms to solve $\DP(n,q,R,\tau)$ \iftoggle{amsbook}{(Problem \ref{prob:decoGenR} in Chapter \ref{chapt:1})}{(Problem $5$ in lecture notes $1$)}. In each case we will show that their running time (over the input distribution) is of the form $2^{n \; \alpha(n,q,R,\tau)(1+o(1))}$. For all of them (and all known algorithms), their exponent $\alpha(n,q,R,\tau)$ is $>0$ as long as $\tau$ does not belong to $\left[ \frac{q-1}{q}(1-R),R+\frac{q-1}{q}(1-R) \right]$ as roughly described in Figure \ref{fig:hardEasyDP}. Our aim during this lecture will be to compute the exponents of the three described algorithms. We draw them in Figures \ref{figure:Prangeq}, \ref{figure:PrangeR} and \ref{figure:PrangeDumerWagner} as function of $\tau$ for some rates $R$ and field sizes $q$.

	\begin{figure}[htb]
	\centering
	\begin{tikzpicture}[scale=0.83]
		\tikzstyle{valign}=[text height=1.5ex,text depth=.25ex]
		\draw[line width=2pt,gray] (-0.7,2) -- (0.3,2);
		\draw (2.5,4.25) node[red]{{\sf $\alpha(n,q,R,\tau) > 0$}};
		\draw (13.5,4.25) node[red]{{\sf $\alpha(n,q,R,\tau) > 0$}};
\draw[line width=2pt,red!50] (0.3,2) -- (5,2);
		\draw (8,2.75) node[blue]{{\sf $\alpha(n,q,R,\tau) = 0$}};
		\draw[line width=2pt,blue!50] (5,2) -- (11,2);
		\draw[line width=2pt,red!50] (11,2) -- (16,2);
		\draw[->,>=latex,line width=2pt,gray] (16,2) -- (17,2)
		node[right,black] {$\displaystyle \tau$};
		\tikzstyle{valign}=[text height=2ex]
		\draw[thick] (0.3,1.9) node[below,valign]{$0$} -- (0.3,2.1);
		\draw[thick] (5,1.9) node[below,valign]{$(1-R)\; \frac{q-1}{q}$} -- (5,2.1);
		\draw[thick] (11,1.9) node[below,valign]{$R+(1-R)\;\frac{q-1}{q}$~~} -- (11,2.1);
		\draw[thick] (13.75,0.5) node[below,valign]{$\tau^{+}$} -- (13.75,0.5);
		\draw[thick] (2.5,0.5) node[below,valign]{$\tau^{-}$} -- (2.5,0.5);
		\draw[thick] (16,1.9) node[below,valign]{$1$} -- (16,2.1);

		\draw[<->,>=latex,thin,purple!50] (2.5,0.9) -- node[below,purple,midway]{{\sf exponentially many solutions}} (13.75,0.9);
		\draw[thick] (2.5,0.5) node[below,valign]{$ $} -- (2.5,3.5);
		\draw[thick] (13.75,0.5) node[below,valign]{$ $} -- (13.75,3.5);
		
		\draw[<->,>=latex,thin,purple!50] (0.3,3.3) -- node[below,purple,midway]{{\sf one solution}} (2.5,3.3);
		
		\draw[<->,>=latex,thin,purple!50] (13.75,3.3) -- node[below,purple,midway]{{\sf one solution}} (16,3.3);

	\end{tikzpicture}
	\caption{Exponents of the best generic decoding algorithms and expected number of solutions of $\mathsf{DP}(n,q,R,\tau)$ as function of $\tau$.\label{fig:hardEasyDP}}
\end{figure}
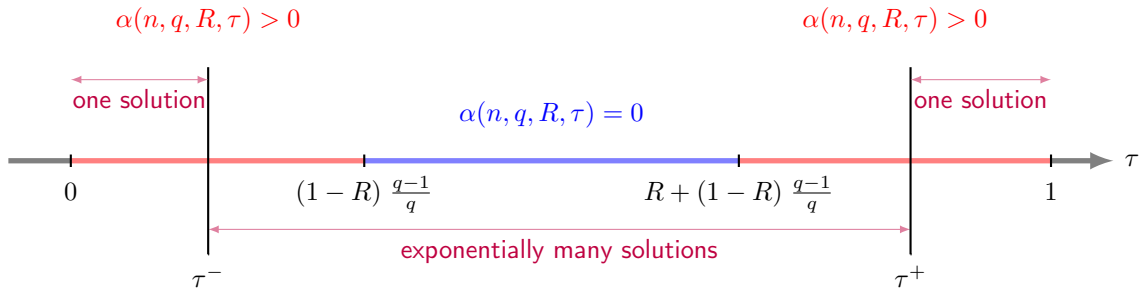

\begin{center}
	\begin{figure}[h]
		\includegraphics[height=7cm]{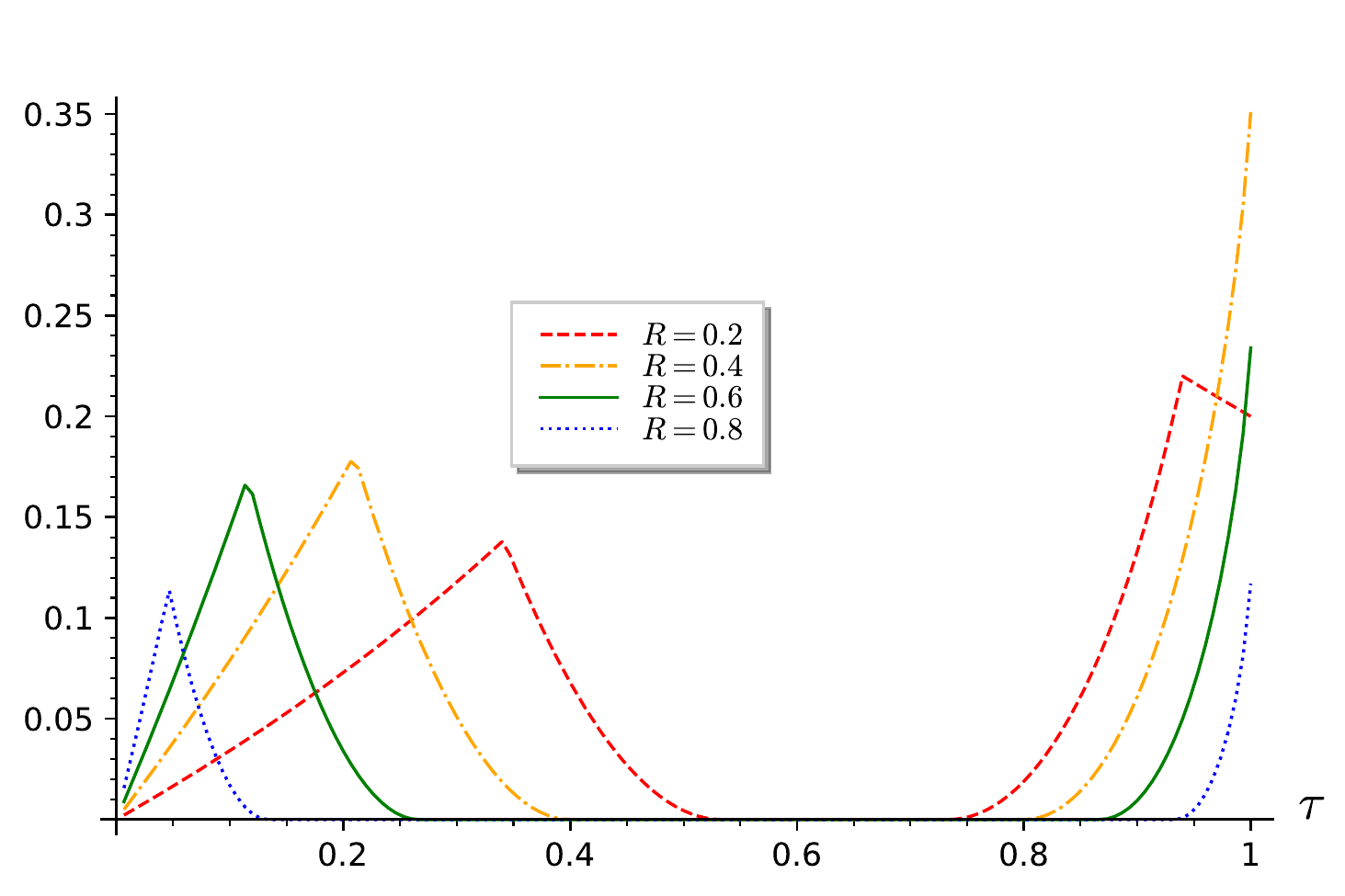}
		\caption{Exponent of Prange's algorithm (in base $2$) to solve $\DP(n,q,R,\tau)$ for $q = 3$ and different rates $R$ as function of $\tau$. 
			\label{figure:Prangeq}
		}
	\end{figure}
\end{center}

\begin{center}
	\begin{figure}[h]
		\includegraphics[height=7cm]{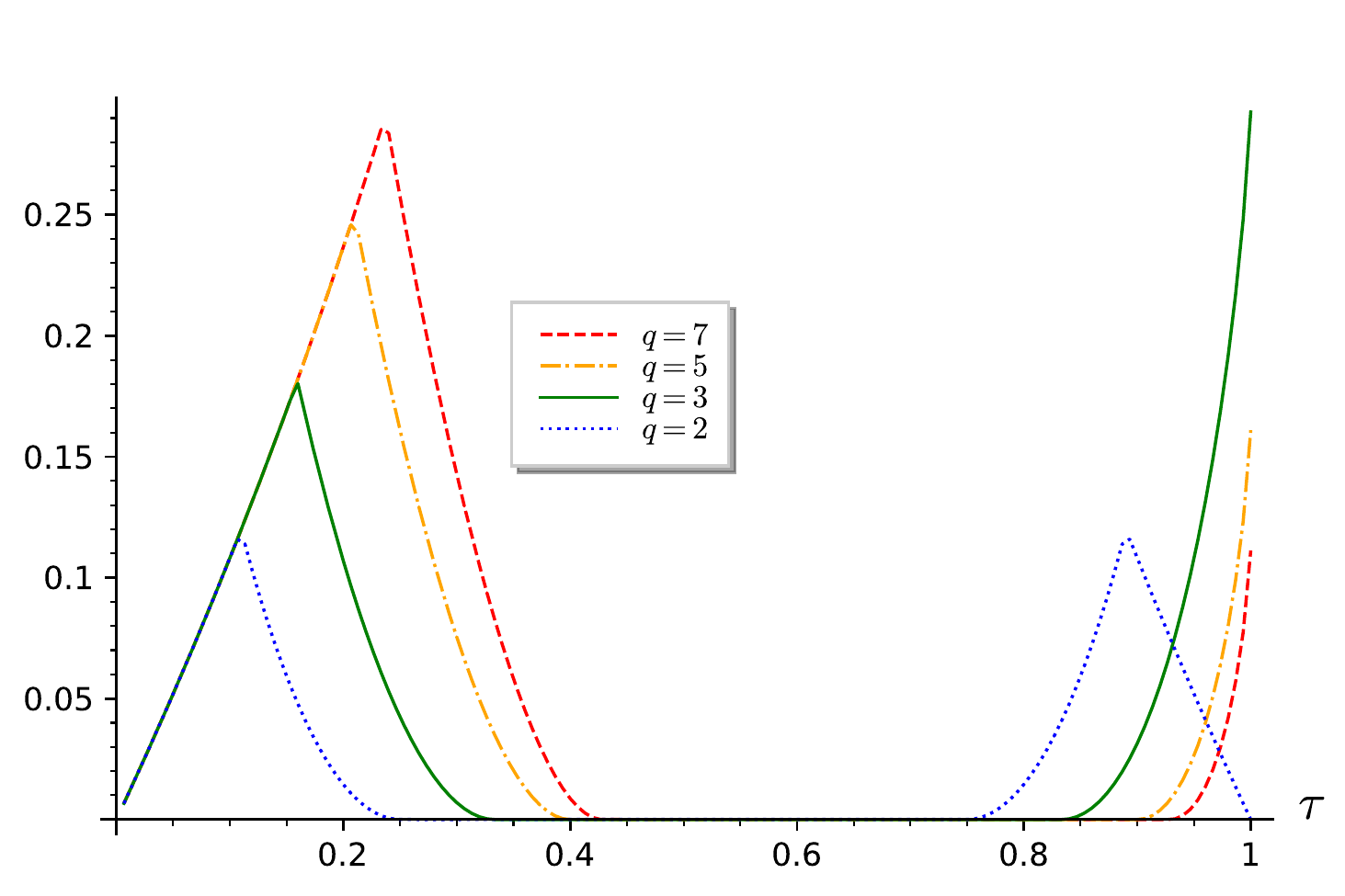}
		\caption{Exponent of Prange's algorithm (in base $2$) to solve $\DP(n,q,R,\tau)$ for $R = 0.5$ and different field sizes $q$ as function of $\tau$.
			\label{figure:PrangeR}
		}
	\end{figure}
\end{center}

\begin{center}
	\begin{figure}[h]
		\includegraphics[height=7cm]{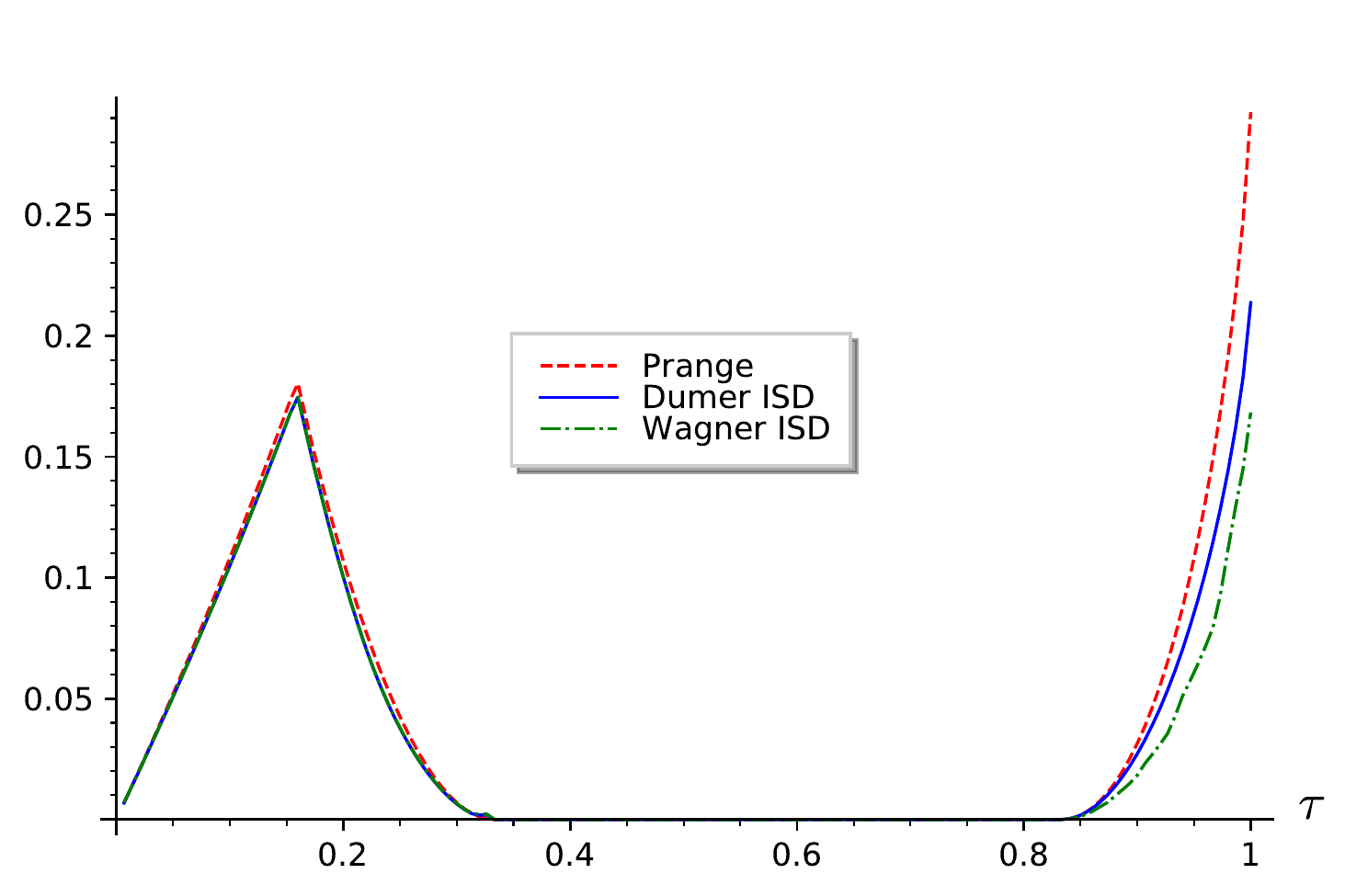}
		\caption{Exponents of Prange's algorithm and ISD with Dumer and Wagner' algorithms (in base $2$) to solve $\DP(n,q,R,\tau)$ for $R = 0.5$ and $q=3$ as function of $\tau$.
			\label{figure:PrangeDumerWagner}
		}
	\end{figure}
\end{center}

\section{Prange Algorithm}

	From now on, let us fix some instance $(\vec{H},\vec{s})\in\Fq^{(n-k)\times n} \times \Fq^{n-k}$ of the decoding problem $\DP(n,q,R,\tau)$ where recall that $k = \lfloor Rn \rfloor$. Our aim is to find $\vec{e} \in \Fq^{n}$ of Hamming weight $t = \lfloor\tau n\rfloor$ such that $\vec{e}\transpose{\vec{H}} = \vec{s}$ and we know that by definition there is at least one solution.

	It corresponds to solving a linear system with $n-k$ equations and $n$ unknowns with a constraint on the Hamming weight of the solution.  Prange's idea simply consists in ``fixing'' $k$ unknowns and then solving a square linear system of size $(n-k) \times (n-k)$ hoping that the solution will have the correct Hamming weight; if not, we just repeat by fixing $k$ other unknowns\footnote{Another interpretation of Prange's algorithm.}. More precisely, Prange's algorithm is as follows. Let us first introduce the following distribution $\mathcal{D}_{t}$ over vectors of $\F_q^{k}$, for reasons that will be clear in the sequel.
	\newline

	{\bf \noindent The distribution $\mathcal{D}_{t}$.}
	\begin{itemize}
		\item If  $t < \frac{q-1}{q}(n-k)$, $\mathcal{D}_{t}$ only outputs $\mathbf{0}\in\Fq^{k}$,

		\item if $t \in \llbracket \frac{q-1}{q}(n-k), k + \frac{q-1}{q}(n-k) \rrbracket$, $\mathcal{D}_{t}$ outputs uniform vectors of weight $t - \frac{q-1}{q}(n-k)$,

		\item if $t > k + \frac{q-1}{q}(n-k)$, $\mathcal{D}_{t}$ outputs uniform vectors of weight $k$.
			\newline
	\end{itemize}

	{\bf \noindent The algorithm.} 
\begin{itemize}
		\item[1.] {\em Picking the information set.} Let $\mathcal{I} \subseteq \llbracket 1,n \rrbracket$ be a random set of size $k$. If $\vec{H}_{\overline{\mathcal{I}}} \in \Fq^{(n-k)\times (n-k)}$ is not of full-rank, pick another set $\mathcal{I}$.

		\item[2.] {\em Linear algebra.} Perform a Gaussian elimination to compute a non-singular matrix $\vec{S} \in\Fq^{(n-k) \times (n-k)}$ such that $\vec{S} \vec{H}_{\overline{\mathcal{I}}} = \mathbf{1}_{n-k}$.

		\item[3.] {\em Test Step.} Pick $\vec{x}\in\Fq^{k}$ according to the distribution $\mathcal{D}_{t}$
and let $\vec{e} \in \Fq^{n}$ be such that
		\begin{equation}\label{eq:ePrange} 
		\vec{e}_{\overline{\mathcal{I}}} = \left( \vec{s} - \vec{x}\transpose{\vec{H}}_{\mathcal{I}}\right) \transpose{\vec{S}}  \quad \mbox{;} \quad \vec{e}_{\mathcal{I}} = \vec{x}.
	    \end{equation}
		If $|\vec{e}| \neq t$ go back to Step $1$, otherwise it is a solution. 
	\end{itemize}

	{\bf \noindent Correction of the algorithm.} It easily follows from the following computation,
	\begin{align*}
		\vec{S}\vec{H} \transpose{\vec{e}} &= \vec{S}\vec{H}_{\overline{\mathcal{I}}}\; \transpose{\vec{e}}_{\overline{\mathcal{I}}}  +  \vec{S}\vec{H}_{\mathcal{I}}\; \transpose{\vec{e}}_{\mathcal{I}}\\
		&= \transpose{\vec{e}}_{\overline{\mathcal{I}}} + \vec{S}\vec{H}_{\mathcal{I}}\; \transpose{\vec{e}}_{\mathcal{I}} \qquad \mbox{(by definition $\vec{S} \vec{H}_{\overline{\mathcal{I}}} = \mathbf{1}_{n-k}$)}\\
		&= \vec{S} \left( \transpose{\vec{s}} - \vec{H}_{\mathcal{I}} \; \transpose{\vec{x}} \right) + \vec{S} \vec{H}_{\mathcal{I}} \; \transpose{\vec{x}} \qquad \mbox{(by Equation \eqref{eq:ePrange})} \\
		&= \vec{S} \transpose{\vec{s}}
	\end{align*}
	which corresponds to $\vec{H}\transpose{\vec{e}} = \transpose{\vec{s}}$ as $\vec{S}$ is non-singular. Furthermore the end of Step $3$ is here to ensure that $\vec{e}$ will have the correct Hamming weight once the algorithm terminates.

	\begin{remark}
		Let $\vec{y} \in \Fq^{n}$ be such that $\vec{y}\transpose{\vec{H}} = \vec{s}$. Notice that $\vec{y} -\vec{e} =\vec{c} \in \CC$ and by definition of $\vec{e}$ output by Prange's algorithm we have $\vec{c}_{\mathcal{I}} = \vec{y}_{\mathcal{I}}$ when $\vec{x} = \mathbf{0}$. In other words, when $\vec{x} = \mathbf{0}$, we recover the interpretation given in introduction to find a close codeword. 
	\end{remark}

	\begin{exercise}
		Describe Prange's algorithm with the generator matrix formalism in the same fashion as above (with also three steps and the distribution $\mathcal{D}_{t}$).
	\end{exercise}

	{\bf \noindent Far or close codeword?} One may ask why did we pick some vector $\vec{x}$ in Step $3$ of the algorithm? Notice that it corresponds to fixing $k$ unknowns to the value that we want. Suppose now that we would like to find a solution $\vec{e}$ of small Hamming weight. Obviously fixing $\vec{x}$ to be a non-zero vector is both needless and counterproductive as it would increase the weight of the decoding candidate. It is therefore better to choose $\vec{x}$ as $\mathbf{0}$ if we seek a solution of small Hamming weight.
	But now what happens if someone is looking for an error $\vec{e}$ of large Hamming weight, let us say close to $n$? The exact opposite: we need to choose $\vec{x}$ as a non-zero Hamming weight vector to increase the weight of the potential solution and therefore improving our success probability.

	In summary, the vector $\vec{x}$ that we pick relies on what we want to do, finding a ``short'' or a ``large'' solution. The distribution $\mathcal{D}_{t}$, upon which $\vec{x}$ is picked, is precisely chosen according to the aforementioned aim. Equivalently, if one takes the generator matrix point of view, the distribution $\mathcal{D}_{t}$ enables to find a close or a far away codeword from a given target.  
	\newline

	{\bf \noindent Rough analysis of the algorithm.} Before giving a precise analysis of the running time of Prange's algorithm let us start by a rough analysis about what we ``expect''. First, let us assume that {\em $\vec{s}$ is uniformly distributed over $\Fq^{n-k}$.} Notice that it is, according to \iftoggle{amsbook}{Proposition \ref{proposition:lhl} in Chapter \ref{chapt:2}}{Proposition $6$ in lecture notes $2$}, equivalent to assuming that $t/n$ belongs to $[\tau^{-},\tau^{+}]$. Therefore, according to Equation  \eqref{eq:ePrange} the expected Hamming weight of the decoding candidate $\vec{e}$ is given by ($\vec{S}$ is non-singular hence it keeps invariant the uniform distribution of $\vec{s}$)
	$$
	\mathbb{E}\left( |\vec{e}| \right) =  |\vec{x}| + \frac{q-1}{q}\;(n-k). 
	$$
	By choosing $\vec{x}=\mathbf{0}$ we expect $\vec{e}$ to have a Hamming weight equals to $\frac{q-1}{q}\;(n-k)$. In other words, if one seeks a solution of $\DP$ with the aforementioned weight, its probability of success (in Step $3$) is roughly $\pr\approx 1$ and the number of repetitions of the whole algorithm will be given by $1/\pr \approx 1$. On the other hand, if one wants a weight smaller than $(1-\varepsilon) \; \frac{q-1}{q} \; (n-k)$ or larger than $(1+\varepsilon) \; \frac{q-1}{q}\; (n-k)$, its probability of success will be exponentially small in $\varepsilon\left( n-k \right)$ since $\vec{s}$ is uniformly distributed. In that case we will need to repeat the three steps an exponential number of times before succeeding.
	However we can turn the above strategy into a stronger one: by carefully choosing $|\vec{x}| \in \llbracket 0,k \rrbracket$ (recall that $\vec{x}$ is a vector of length $k$, the co-dimension of our ``constrained'' linear system to solve), we can easily reach any weight in
	$$
	\left\llbracket \frac{q-1}{q} \; (n-k), \;  k + \frac{q-1}{q}\; (n-k) \right\rrbracket.
	$$	
	It explains why there is a whole interval in which $\DP$ is claimed to be easy to solve (as drawn in Figure \ref{fig:hardEasyDP}). Let us stress once again that no algorithm is known to solve $\DP$ in polynomial time outside this range of parameters (up to an additive logarithmic factors in the above interval).
	\newline

	{\bf \noindent Precise analysis of the algorithm.} All the challenge in the analysis of Prange's algorithm running time relies on the computation of the success probability in Step $3$. From now on we will we make the following assumption concerning Step $1$ of the algorithm
	\begin{assumption}\label{ass:I}
	  The success probability of Prange's algorithm is equal (up to a polynomial factor) to the probability of success when $\mathcal{I}$ is supposed to be uniformly distributed in Step $1$.
	\end{assumption}

	It is an usual assumption (or heuristic) to make when studying the complexity of Prange's algorithm. Notice that we did not suppose that $\mathcal{I}$ is uniformly distributed, but that our probabilities will be well approximated by making this assumption. It would be obviously false to suppose directly that $\mathcal{I}$ is uniformly distributed as $\vec{H}_{\overline{\mathcal{I}}}$ will be non-singular for some sets $\mathcal{I}$ (at the exception of very particular cases). However, when $\vec{H}$ is random, $\vec{H}_{\overline{\mathcal{I}}}$ is typically non-singular.

	In the following lemma we give the success probability of Prange's algorithm. Our proof is written with a lot of details. 	
We will not repeat this and we will process in a simpler way. The idea is that we study algorithms to solve $\DP$ {\em in average} and from a cryptographic point of view we are on the cryptanalysis side. Our aim is to show that solving the decoding problem requires at least some number of operations. 
It is why we can suppose to live in the best world for a cryptanalyst. For instance, an event that is expected or that occurs with a probability given by the inverse of a polynomial, {\em always} happens and we are not concerned with approximation factors (although some heuristics may be hidden). The rationale behind the following proof is to show that what follows during these lecture notes could be stated and proved very precisely but at the price of significantly increasing the complexity of statements and their proofs while at the same time without changing conclusions.

	\begin{proposition}\label{propo:probaPrange} 
	Let $\pr$ be the success probability in Step $3$ of the above algorithm. Under Assumption \ref{ass:I}, we have
	$$
	\pr = \Theta\left( \frac{\binom{n-k}{t-j}(q-1)^{t-j}}{\min\left( q^{n-k},\binom{n}{t}(q-1)^{t} \right)} \right) 
	$$
	for a density $1 - 2^{-\Omega(n)}$ of matrices $\vec{H}\in\Fq^{(n-k)\times n}$, where
	$$
	j \eqdef \left\{
	\begin{array}{ll}
	0 & \mbox{ if } t < \frac{q-1}{q}(n-k) \\
	t - \frac{q-1}{q}(n-k)	 & \mbox{ if } t \in \llbracket \frac{q-1}{q}(n-k), k + \frac{q-1}{q}(n-k) \rrbracket \\
	k & \mbox{ otherwise.}
	\end{array}
	\right.
	$$
	\end{proposition}

	\begin{proof}
			Notice that input $(\vec{H},\vec{s})$ is fixed, the randomness of the algorithm comes from $\vec{x}$ picked according to $\mathcal{D}_{t}$ and the drawing of the information set $\mathcal{I}$. Under Assumption \ref{ass:I}, our probability computations over $\mathcal{I}$ are up to a polynomial factor given by the case where $\mathcal{I}$ is uniformly distributed (we will not write the polynomial during the computations).

			Let us fix $\vec{e}^{(1)}$ to be a solution of our decoding problem (we know that there is at least one). To compute the success probability of Prange's algorithm let us first notice that an iteration will succeed if $\vec{x} = \vec{e}_{\mathcal{I}}^{(1)}$, namely
			\begin{equation}\label{eq:Pre1}
				\mathbb{P}\left( \mbox{an iteration of Prange finds $\vec{e}^{(1)}$} \right) = \mathbb{P}_{\mathcal{I},\vec{x}}\left(   \vec{x} = \vec{e}^{(1)}_{\mathcal{I}} \right) 
			\end{equation} 
			It comes from the fact that $\vec{e}^{(1)}$ is uniquely determined by $\vec{e}_{\mathcal{I}}^{(1)}$ as necessarily $\vec{e}^{(1)}_{\overline{\mathcal{I}}} = \left( \vec{s} - \vec{e}^{(1)}_{\mathcal{I}}\transpose{\vec{H}}_{\mathcal{I}}\right)\transpose{\vec{S}}$. Furthermore, $\mathcal{D}_{t}$ only outputs vectors $\vec{x}$ of Hamming weight $j$. To find $\vec{e}^{(1)}$ it is necessary that during an iteration we have $\left|\vec{e}^{(1)}_{\overline{\mathcal{I}}}\right| = t-j$ as $|\vec{e}^{(1)}| = t$ and $|\vec{x}| = j$.
Therefore, using the law of total probability, we obtain the following computation
			\begin{align}
			\mathbb{P}_{\mathcal{I},\vec{x}}\left(   \vec{x} = \vec{e}^{(1)}_{\mathcal{I}} \right)  
&= \mathbb{P}_{\mathcal{I},\vec{x}}\left( \vec{x} = \vec{e}_{\mathcal{J}}^{(1)} \mid \left|\vec{e}^{(1)}_{\mathcal{J}}\right| =t-j \right)\mathbb{P}_{\mathcal{I},\vec{x}}\left( \left|\vec{e}^{(1)}_{\mathcal{J}}\right| =t-j \right) \label{eq:probPrangeCdt}
			\end{align}
			 The probability to find $\vec{e}^{(1)}$ in one iteration is given by the probability that
			\begin{enumerate}
			  \item[$(i)$] \label{proba:(i)} $\mathcal{I}$ is such that $\left|\vec{e}_{\overline{\mathcal{I}}}^{(1)}\right| = t-j$
			  \item[$(ii)$] \label{proba:(ii)} $\vec{x} = \vec{e}^{(1)}_{\mathcal{J}}$ supposing $(i)$. 
			\end{enumerate}

			Under Assumption \ref{ass:I}, the probability of $(i)$ is $\frac{\binom{t}{j}\binom{n-t}{k-j}}{\binom{n}{k}}$ while the probability of $(ii)$ is given by $\frac{1}{\binom{k}{j}(q-1)^{j}} $ as $\vec{x}\in\Fq^{k}$ picked according to $\mathcal{D}_{t}$ is uniformly distributed among words of Hamming weight $j$. Plugging this in Equation \eqref{eq:Pre1} and using Equation \eqref{eq:probPrangeCdt} we obtain the following computation

			\begin{align*} 
			\mathbb{P}\left( \mbox{an iteration of Prange finds $\vec{e}^{(1)}$} \right) &= \mathbb{P}_{\mathcal{I},\vec{x}} \left( \vec{x} = \vec{e}_{\mathcal{I}}^{(1)} \mid \left|\vec{e}_{\mathcal{I}}\right| =t-j \right)\mathbb{P}_{\mathcal{I}}\left( \left|\vec{e}_{\mathcal{I}}\right| =t-j \right) \\
			&=  \frac{1}{\binom{k}{j}(q-1)^{j}} \; \frac{\binom{t}{j}\binom{n-t}{k-j}}{\binom{n}{k}}\\
			& = \frac{\binom{n-k}{t-j}(q-1)^{t-j}}{\binom{n}{t}(q-1)^{t}}
			\end{align*} 	
			where the last equality follows from a simple computation.

			Recall now that we are sure that there is at least one solution of the decoding problem. However, depending on $t$, it may happen that there are more. Let us denote by $N$ the number of solutions. According to the above equation, the probability to find none of these in one iteration of the algorithm is given by
			\begin{equation*}
\left(1- \frac{\binom{n-k}{t-j}(q-1)^{t-j}}{\binom{n}{t}(q-1)^{t}} \right)^{N} = 1 - \Theta\left(  N \;  \frac{\binom{n-k}{t-j}(q-1)^{t-j}}{\binom{n}{t}(q-1)^{t}} \right) 
			\end{equation*} 
			Here we used that the randomness $\mathcal{I}$, $\vec{x}$ of the algorithm is independent of the solutions. Therefore, the probability $\pr$ that Prange's algorithm succeeds is 
			\begin{equation}\label{eq:succN}
				\pr = \Theta\left( N \;  \frac{\binom{n-k}{t-j}(q-1)^{t-j}}{\binom{n}{t}(q-1)^{t}} (1+o(1)) \right)
			\end{equation} 	
			But now recall from \iftoggle{amsbook}{Chapter \ref{chapt:2} \footnote{Depending on which term achieves the maximum, we use Propositions $\ref{propo:FT}$ or $\ref{propo:ST} $.}}{lecture notes $2$\footnote{Depending on which term achieves the maximum, we use from lecture notes $2$, Proposition $2$ or $3$.}} that for any constant $C$, 
			$$
			\mathbb{P}_{\vec{H}}\left( \left| N - \max\left( 1, \frac{\binom{n}{t}(q-1)^{t}}{q^{n-k}}\right) \right| > C  \max\left( 1, \frac{\binom{n}{t}(q-1)^{t}}{q^{n-k}}\right) \right) =   2^{-\Omega(n)}
			$$ 
			Therefore, since $\vec{H}$ is uniformly distributed in the above probability, we have for a density $1-2^{-\Omega(n)}$ of matrices $\vec{H}$, 
\begin{align*}
			\pr =  \Theta\left(  \max\left( 1, \frac{\binom{n}{t}(q-1)^{t}}{q^{n-k}}\right) \; \frac{\binom{n-k}{t-j}(q-1)^{t-j}}{\binom{n}{t}(q-1)^{t}} \right) = \Theta\left( \frac{\binom{n-k}{t-j}(q-1)^{t-j}}{\min\left( q^{n-k},\binom{n}{t}(q-1)^{t} \right)} \right)
			\end{align*} 
			where we used Equation \eqref{eq:succN}. It concludes the proof. 
	\end{proof} 
	
	We are now ready to give the running-time of Prange's algorithm to solve $\DP$. It will be a simple consequence of the above proposition. 
	
	\begin{corollary}
		Under Assumption \ref{ass:I}, the complexity $C_{\textup{Prange}}(n,q,R,\tau)$  of Prange's algorithm to solve $\DP(n,q,R,\tau)$ is up to a polynomial factor (in $n$) given by
		$$
\frac{\min\left( q^{n-k},\binom{n}{t}(q-1)^{t} \right)}{\binom{n-k}{t-j}(q-1)^{t-j}}
		$$
where $k \eqdef \lfloor Rn \rfloor$, $t \eqdef \lfloor \tau n \rfloor$ and
			$$
		j \eqdef 	\left\{
		\begin{array}{ll}
			0 & \mbox{ if } t < \frac{q-1}{q}(1-R) \\
			t - \frac{q-1}{q}(n-k)	 & \mbox{ if } t \in \llbracket \frac{q-1}{q}(n-k), k + \frac{q-1}{q}(n-k) \rrbracket \\
			k & \mbox{ otherwise.}
		\end{array}
		\right.
		$$
	\end{corollary}

	\begin{proof}
		The cost of an iteration of Prange's algorithm is dominated by the time to perform a Gaussian elimination. Let $\pr$ be the probability of success of an iteration which is given in Proposition \ref{propo:probaPrange}. The number of iterations is (up to a polynomial in $n$) $1/\pr$ with a probability exponentially close to one (in $n$). The latter affirmation comes from the fact that the number of iterations is a geometric distribution.
	\end{proof}

	To conclude this section let us briefly study the asymptotic complexity (in $n$) of Prange's algorithm.
	\newline

	{\bf \noindent Asymptotic complexity: use the entropy function.} When studying the asymptotic complexity of ISD algorithms it will be important to be familiar with the $q$-ary entropy function and its properties. Recall from \iftoggle{amsbook}{Chapter \ref{chapt:2}}{lecture notes $2$} that it is defined as (and extended by continuity)
	$$
	h_{q} : x \in [0,1] \longmapsto -x \log_{q}\left(\frac{x}{q-1}\right) - (1-x)\log_{q}(1-x). 
	$$
	The $q$-ary entropy is an increasing function over $\left[0,\frac{q-1}{q}\right]$ and a decreasing function over $\left[ \frac{q-1}{q},1\right]$. It reaches its maximum $1$ in $\frac{q-1}{q}$.

	This function has the nice property to describe the asymptotic behaviour of binomials, namely \iftoggle{amsbook}{(see Lemma \ref{lemma:asymptSphere})}{(Lemma $1$ in lecture notes $2$)}
	\begin{equation}\label{eq:asymptEnt} 
	\frac{1}{n} \log_{q} \binom{n}{t}(q-1)^{t} \mathop{=}\limits_{n \to +\infty} h_{q}(\tau) + O\left(\frac{\log_{q}n}{n}\right)
	\end{equation}
	where $\tau =t/n$.
	From this we easily deduce the exponent of Prange's algorithm 
	\begin{equation}\label{eq:PrangeExponent} 
	\frac{1}{n} \log_{q} C_{\textup{Prange}}(n,q,R,\tau) \mathop{=}\limits_{n \to +\infty}   \min\left(1-R, h_{q}(\tau)\right)  -  (1-R) \; h_{q}\left( \frac{\tau - \gamma}{1-R} \right) + O\left( \frac{\log_{q}n}{n} \right)
	\end{equation} 
	where,
	$$
	\gamma \eqdef \left\{
	\begin{array}{ll}
		0 & \mbox{ if } \tau < \frac{q-1}{q}(n-k) \\
		\tau - \frac{q-1}{q}(1-R)	 & \mbox{ if } \tau \in \left[ \frac{q-1}{q}(1-R), R + \frac{q-1}{q}(1-R) \right] \\
		R & \mbox{ otherwise.}
	\end{array}
	\right.
	$$

	There are particular ranges of parameters for which Equation \eqref{eq:asymptEnt} simplifies. First, in the case where $\tau \in \left[  \frac{q-1}{q}(1-R), R + \frac{q-1}{q}(1-R) \right]$, this exponent is a $O\left( \frac{\log_q n}{n}\right)$ (it is easily verified using that $h_{q}((q-1)/q) = 1$). It corresponds to what we have expected as we claimed that Prange's algorithm is polynomial in this range of parameters.

	We used Equation \eqref{eq:PrangeExponent} multiplied by a factor $\log_{2}(q)$ (exponents are in base $2$) to draw Figures \ref{figure:Prangeq}, \ref{figure:PrangeR} and \ref{figure:PrangeDumerWagner}. Notice in the case $q=2$ that the complexity of Prange's algorithm is symmetric around $1/2$ which is expected in that case. Considering short or large weight in the binary case is equivalent as you have to show in the following exercise. In particular, in the sequel we will not compute the exponents of algorithms for solving $\DP(n,2,R,\tau)$ with $\tau > 1/2$.

	\begin{exercise}\label{ex:sym} 
		Let $\tau \in [0,1/2]$. Show how from an algorithm solving $\DP(n,2,R,\tau)$ we can deduce an algorithm solving $\DP(n,2,R,1-\tau)$ in the same running-time (and reciprocally). 
	\end{exercise}

	Let us consider now the case $\tau = o(1)$ (and therefore $\gamma = 0$). It corresponds to parameters of all code-based public-key encryption schemes  (for instance \cite{M78,A03,MTSB13,HQC17}). Using
	\begin{equation}\label{eq:asymptqh}
	h_{q}(x) \mathop{=}\limits_{x\rightarrow 0} -x \log_{q}\left(\frac{x}{q-1}\right) + x + o(x)
	\end{equation}
	and Equation \eqref{eq:PrangeExponent}  (with $\gamma = 0$) we obtain the following computation,
	\begin{align*}
		\frac{1}{n} \log_{q} C_{\textup{Prange}}(n,q,R,\tau) &=  \min\left(1-R, h_{q}(\tau)\right) - (1-R) \; h_{q}\left( \frac{\tau}{1-R} \right) + O\left( \frac{\log_q n}{n} \right) \\
		&= h_{q}(\tau) - (1-R) \; h_{q}\left( \frac{\tau}{1-R} \right) + O\left( \frac{\log_q n}{n} \right) \\
		&= -\tau\log_{q}\left(\frac{\tau}{q-1}\right)  + \tau\log_{q}\left( \frac{\tau}{1-R} \; \frac{1}{q-1} \right) + o(\tau) + O\left( \frac{\log_q n}{n} \right) \; \mbox{(see Eq. \eqref{eq:asymptqh})} \\
		&= - \tau\log_{q} (1-R) + o(\tau) + O\left( \frac{\log_q n}{n} \right).
	\end{align*} 
	Therefore, when $\tau = o(1)$ ($t = \tau n$), the complexity of Prange algorithm is given by (for some constant $C$)
	$$
	C_{\textup{Prange}}(n,q,R,\tau) = n^{C} \; q^{-t \log_{q}(1-R)}.
	$$
	It is even more remarkable that no algorithm is known to have a complexity $q^{c t(1+o(1))}$ with $c < -\log_{q}(1-R)$ as soon as $t = o(n)$. Furthermore, all known ISD (even the most sophisticated) have the same asymptotic complexity than Prange's algorithm for these parameters \cite{CS16}. Despite its extreme simplicity, Prange's algorithm is the best known algorithm to solve asymptotically $\DP(n,q,R,\tau)$ when the decoding distance is sub-linear, namely $\tau = o(1)$.

	\section{Birthday Paradox Techniques}

	We present in this section two algorithms for solving $\DP(n,q,R,\tau)$. Both rely on the following {\em crucial} lemma which is essentially an average version of the birthday paradox

	\begin{lemma}\label{lemm:BP}
		Let $\mathcal{L}_{1},\mathcal{L}_{2}$ be two lists of $L$ random and independent elements in $\Fq^{r}$. We have,
		$$
		\mathbb{E}\left(\sharp \;\mathcal{L}_{1} \cap \mathcal{L}_{2}\right) = \frac{L^{2}}{q^{r}}.
		$$ 
	\end{lemma}
	Notice that we expect one element in the intersection of the two lists included in $\Fq^{r}$ when their size verifies $L = \sqrt{q^{r}}$. Recall that the birthday paradox, asserts that when there are $\sqrt{N}$ elements picked uniformly at random among a set of size $N$, we will get with a good probability two equal elements. 
	It explains why we refer to the above lemma as the birthday paradox.

	\begin{proof}
		By definition, $\mathcal{L}_{1} = \left\{X_1,\dots,X_{L}\right\}$ and $\mathcal{L}_{2}=\left\{Y_1,\dots,Y_L \right\}$  where the $X_{i}$'s and $Y_{j}$'s are independent and uniformly distributed random variables taking their values in $\Fq^{r}$. We have
		$$
		\sharp \; \mathcal{L}_{1} \cap \mathcal{L}_{2} = \sum_{i,j=1}^{L} \mathds{1}_{\left\{ X_{i} = Y_{j} \right\}}.
		$$ 
		By linearity of the expectation we have the following computation,
		\begin{align*}
			\mathbb{E}\left( \sharp \; \mathcal{L}_{1} \cap \mathcal{L}_{2} \right) = \sum_{i,j=1}^{L} \mathbb{E}\left( \mathds{1}_{\left\{ X_{i} = Y_{j}\right\}} \right) = \sum_{i,j=1}^{L} \frac{1}{q^{r}}  = \frac{L^{2}}{q^{r}}
		\end{align*} 
		which concludes the proof. 
	\end{proof} 
	
	\subsection{Dumer's Algorithm}\label{subsec:Dumer}

	Let us now quickly present Dumer's algorithm \cite{D86} to solve $\DP(n,q,R,\tau)$. This short subsection may be skipped as the description of this algorithm has already been given in the introduction (in the same fashion). 
	\newline

	{\bf \noindent The algorithm.}
	\begin{itemize}
		\item[1.] {\em Splitting in two parts.} First we randomly select a set $\mathcal{S} \subseteq \llbracket 1,n \rrbracket$ of $n/2$ positions.

		\item[2.] {\em Building lists step.} We build, 
		$$
		\mathcal{L}_{1} \eqdef \left\{  \vec{H}_{\mathcal{S}}\transpose{\vec{e}}_1 : |\vec{e}_{1}| = \frac{t}{2} \right\} \quad \mbox{;} \quad \mathcal{L}_{2} \eqdef \left\{-\vec{H}_{\overline{\mathcal{S}}}\transpose{\vec{e}}_2 + \transpose{\vec{s}} : |\vec{e}_{2}| = \frac{t}{2} \right\}.
		$$

		\item[3.] {\em Collisions step.} We merge the above lists (with an efficient technique like hashing or sorting) 
		$$
		\mathcal{L}_{1} \bowtie \mathcal{L}_{2} \eqdef \left\{  (\vec{e}_{1},\vec{e}_{2}) \in \mathcal{L}_{1} \times \mathcal{L}_{2},\quad  \vec{H}_{\mathcal{S}}\transpose{\vec{e}}_1 =  - \vec{H}_{\overline{\mathcal{S}}}\transpose{\vec{e}}_2 + \transpose{\vec{s}}  \right\}.
		$$
		and output this new list. If it is empty we go back to Step $1$ and pick another set of $n/2$ positions. 
	\end{itemize}

	\begin{proposition}\label{propo:cpxDum} The complexity $C_{\textup{Dumer}}(n,q,R,\tau)$ of Dumer's algorithm to solve $\DP(n,q,R,\tau)$ is up to a polynomial factor (in $n$) given by 
		$$
		 \sqrt{\binom{n}{t}(q-1)^{t}} + \frac{\binom{n}{t}(q-1)^{t}}{q^{n-k}}
		$$
		Furthermore, Dumer's algorithm finds $\max\left( 1, \frac{\binom{n}{t}(q-1)^{t}}{q^{n-k}} \right)$ solutions (up to a polynomial factor in $n$) where $k \eqdef \lfloor Rn \rfloor$ and $t \eqdef \lfloor \tau n \rfloor$.
	\end{proposition}

	\begin{proof} 
Dumer's algorithm will find a fixed solution of the considered decoding problem with probability
		$$
		\frac{\binom{t}{t/2}\; \binom{n-t}{n/2-t/2}}{\binom{n}{n/2}} = 2^{ th_{2}(1/2) + (n-t)h_{2}(1/2)  - nh_{2}(1/2) + O(\log_2 n) } = 2^{O(\log_2 n)}
		$$
		which is polynomial. Therefore the number of iterations of Dumer's algorithm to find a solution will be polynomial.

		The cost of one iteration is given by the time to build lists $\mathcal{L}_{1},\mathcal{L}_{2}$, namely $\binom{n/2}{t/2}(q-1)^{t/2}$ plus the time to merge them. With efficient techniques such as sorting or hashing this can be done in time $\sharp\; \mathcal{L}_{1}\bowtie \mathcal{L}_{2}$. But, according to Lemma \ref{lemm:BP}, the expected size of $\mathcal{L}_{1}\bowtie \mathcal{L}_{2}$ is in average over $\vec{H}$ given by ${\left(\binom{n/2}{t/2}(q-1)^{t/2}\right)^{2}}/{q^{n-k}}$ which is equal to (up to a polynomial factor) ${\binom{n}{t}(q-1)^{t}}/{q^{n-k}}$ as collisions are made on vectors which belong to $\Fq^{n-k}$. It concludes the proof. 
	\end{proof}

	\begin{remark}
		We have presented Dumer's algorithm to find all solutions of $\DP$. But one can also tweak this algorithm to build less solutions in one iteration. 
	\end{remark}

	\begin{exercise}\label{ex:Dumer} We have made the choice when presenting Dumer's algorithm to build lists of maximum size, namely $\binom{n/2}{t/2}(q-1)^{t/2}$ (why is it the largest possible list size?). Let $(\vec{H},\vec{s})\in\Fq^{(n-k)\times n} \times \Fq^{n-k}$ be an instance of a decoding problem that we would like to solve at distance $t$. 
Show that a slight variation of Dumer's algorithm enables to compute $\max\left( 1,\frac{L^{2}}{q^{n-k}}\right)$ solutions in time $L + \max\left( 1,\frac{L^{2}}{q^{n-k}}\right)$ (up to polynomial factors). How $L$ needs to be chosen for this algorithm to output solutions in amortized time one? Deduce a necessary condition over $t$ for this to be possible. 
\end{exercise}

	\subsection{Wagner's Algorithm} \label{subsec:Wagner}

	We have just seen that Dumer's algorithm finds all solutions of $\DP$ in roughly one iteration. It is an extremely nice property but that may be an impediment in some contexts. 
	Suppose that one needs $M$ solutions of $\DP$ to achieve some task. The best situation would be to find them in amortized time one. Suppose now that Dumer's algorithm is able to compute $N$ solutions of $\DP$ in amortized time one, namely it builds lists of size $N$ which verify
	$$
	N  = \frac{N^{2}}{q^{n-k}} \iff N = q^{n-k}
	$$
	but unfortunately $N \gg M$. In other words, Dumer's algorithm finds too many solutions. To avoid this useless situation we may be tempted to decrease the size of the built lists, namely $N$, to decrease the number of output solutions. However by doing this we would not produce decoding solutions in amortized time one, which would be less efficient for our purpose. To improve this situation, the fundamental remark is that Dumer's algorithm produces all its solutions with a shape $(\vec{e}_{1},\vec{e}_{2})$, where $|\vec{e}_{1}| = |\vec{e}_{2}| = t/2$, and by looking at collisions directly on $n-k$ symbols.  The idea to produce less solutions, still in amortized time one, is to look for solutions with more constraints on their shapes and the way that collisions are built. It is precisely the idea of Wagner's algorithm \cite{W02} (producing less solutions in amortized time one by decimating the search space) that we are going to present precisely in the sequel. However, as a picture is better than a long discourse, let us first describe in Figure \ref{fig:Wagner} a simplified version of this algorithm when we try to find $\vec{e}$ of Hamming weight $t$ such that $\vec{H}\transpose{\vec{e}} = \mathbf{0}$.

	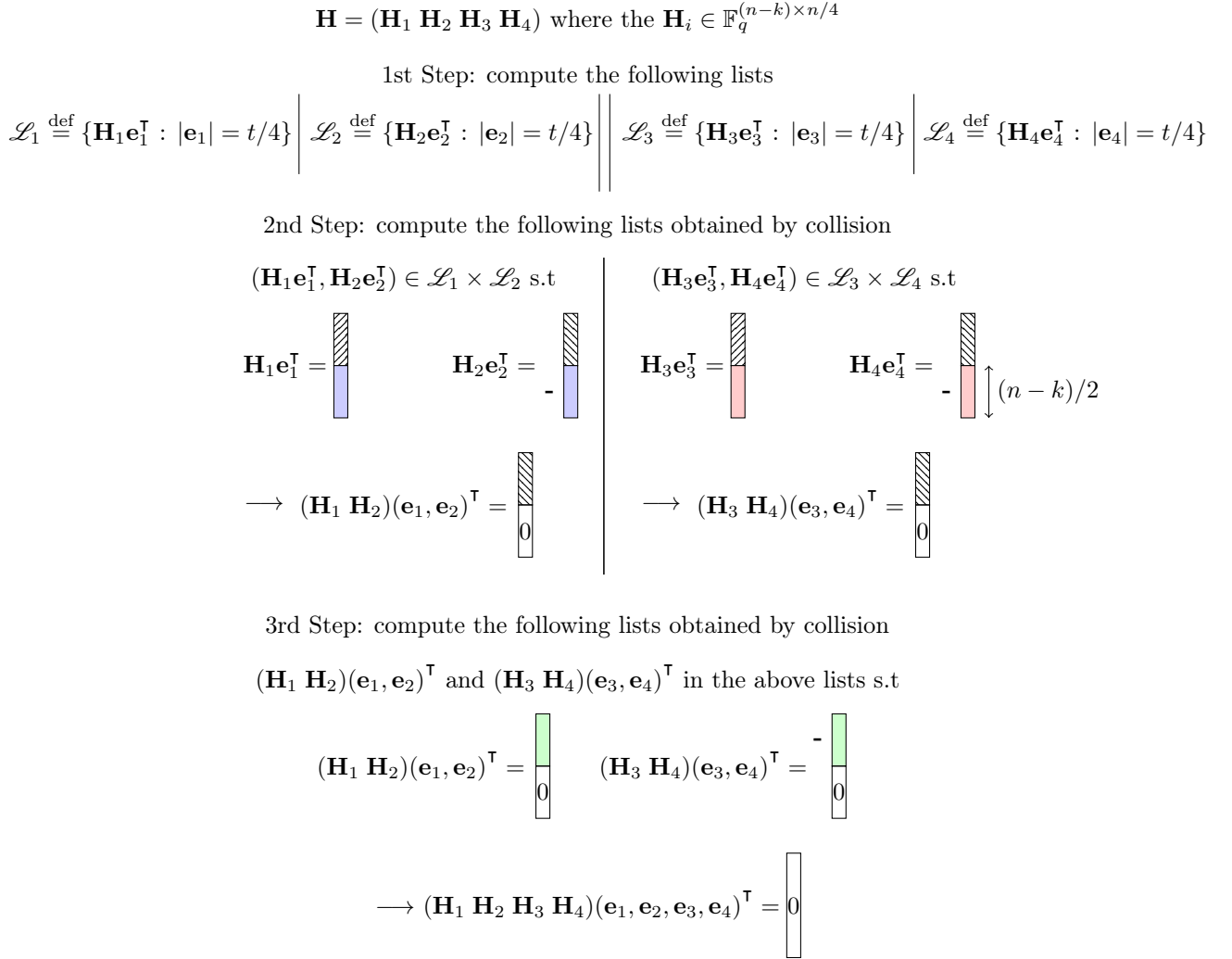
\begin{figure}
	\begin{center} 
		\begin{tikzpicture}
		\node at (5,0) {$\vec{H} = \left(\vec{H}_{1} \; \vec{H}_{2} \; \vec{H}_{3} \; \vec{H}_{4}\right)$ where the $\vec{H}_{i}\in \Fq^{(n-k)\times n/4}$};
		\node at (5,-0.85) {$1$st Step: compute the following lists};
		\node at (3.2,-1.6) {$\mathcal{L}_{2} \eqdef \left\{ \vec{H}_{2}\transpose{\vec{e}}_{2} \mbox{ : } |\vec{e}_{2}| = t/4 \right\}$};
		\node at (-1.125,-1.6) {$\mathcal{L}_{1} \eqdef \left\{ \vec{H}_{1}\transpose{\vec{e}}_{1} \mbox{ : } |\vec{e}_{1}| = t/4 \right\}$};
		\draw (1,-1.1) -- (1,-2.25);
		\draw (5.46,-1.1) -- (5.46,-2.5);
		\draw (5.31,-1.1) -- (5.31,-2.5);
		\node at (7.675,-1.6) {$\mathcal{L}_{3} \eqdef \left\{ \vec{H}_{3}\transpose{\vec{e}}_{3} \mbox{ : } |\vec{e}_{3}| = t/4 \right\}$};
		\node at (12,-1.6) {$\mathcal{L}_{4} \eqdef \left\{ \vec{H}_{4}\transpose{\vec{e}}_{4} \mbox{ : } |\vec{e}_{4}| = t/4 \right\}$};
		\draw (9.825,-1.1) -- (9.825,-2.25);

		\node at (5,-3) {$2$nd Step: compute the following lists obtained by collision};
		\draw (5.375,-3.45) -- (5.375,-8);
		\node at (2.5,-3.75)  {$(\vec{H}_{1}\transpose{\vec{e}}_{1},\vec{H}_{2}\transpose{\vec{e}}_{2})\in \mathcal{L}_{1}\times \mathcal{L}_{2}$ s.t};	
		\node at (8.25,-3.75)  {$(\vec{H}_{3}\transpose{\vec{e}}_{3},\vec{H}_{4}\transpose{\vec{e}}_{4})\in \mathcal{L}_{3}\times \mathcal{L}_{4}$ s.t};
			
		\node at (0.8,-5) {$\vec{H}_{1}\transpose{\vec{e}}_{1}=$};
		\draw[pattern= north east lines] (1.5,-5) rectangle (1.7,-4.25);
		\draw[fill=blue!20] (1.5,-5.75) rectangle (1.7,-5);
		\node at (3.8,-5) {$\vec{H}_{2}\transpose{\vec{e}}_{2}=$};
		\draw[pattern= north west lines] (4.8,-5) rectangle (5,-4.25);
		\draw[fill=blue!20] (4.8,-5.75) rectangle (5,-5);
		\node at (4.6,-5.375) {{\bf -}};
		
		\node at (6.5,-5) {$\vec{H}_{3}\transpose{\vec{e}}_{3}=$};
		\draw[pattern= north east lines] (7.2,-5) rectangle (7.4,-4.25);
		\draw[fill=red!20] (7.2,-5.75) rectangle (7.4,-5);
		\node at (9.5,-5) {$\vec{H}_{4}\transpose{\vec{e}}_{4}=$};
		\draw[pattern= north west lines] (10.5,-5) rectangle (10.7,-4.25);
		\draw[fill=red!20] (10.5,-5.75) rectangle (10.7,-5);
		\node at (10.3,-5.375) {{\bf -}};
		\node at (11.75,-5.375) {$(n-k)/2$};
		\draw[<->] (10.9,-5.75) -- (10.9,-5);
		
		\node at (0.5,-7) {$\longrightarrow$};
		\node at (2.5,-7) {$(\vec{H}_{1}\; \vec{H}_{2})\transpose{(\vec{e}_{1},\vec{e}_{2})}=$};
		\draw[pattern= north west lines] (4.15,-7) rectangle (4.35,-6.25);
		\draw (4.15,-7.75) rectangle (4.35,-7);
		\node at (4.25,-7.375) {$0$};

		\node at (6.2,-7) {$\longrightarrow$};
		\node at (8.2,-7) {$(\vec{H}_{3}\; \vec{H}_{4})\transpose{(\vec{e}_{3},\vec{e}_{4})}=$};
		\draw[pattern= north west lines] (9.85,-7) rectangle (10.05,-6.25);
		\draw (9.85,-7.75) rectangle (10.05,-7);
		\node at (9.95,-7.375) {$0$};

		\node at (5,-8.75) {$3$rd Step: compute the following lists obtained by collision};
		\draw (5.375,-3.45) -- (5.375,-8);
		\node at (5,-9.5) {$(\vec{H}_{1} \; \vec{H}_{2})\transpose{(\vec{e}_{1},\vec{e}_{2})}$ and $(\vec{H}_{3} \; \vec{H}_{4})\transpose{(\vec{e}_{3},\vec{e}_{4})}$  in the above lists s.t };
		\node at (2.75,-10.75) {$(\vec{H}_{1}\; \vec{H}_{2})\transpose{(\vec{e}_{1},\vec{e}_{2})}=$};
		\draw[fill = green!20] (4.4,-10.75) rectangle (4.6,-10);
		\draw (4.4,-11.5) rectangle (4.6,-10.75);
		\node at (4.5,-11.125) {$0$};
		
		\node at (6.8,-10.75) {$(\vec{H}_{3}\; \vec{H}_{4})\transpose{(\vec{e}_{3},\vec{e}_{4})}=$};
		\draw[fill = green!20] (8.65,-10.75) rectangle (8.85,-10);
		\draw (8.65,-11.5) rectangle (8.85,-10.75);
		\node at (8.75,-11.125) {$0$};
		\node at (8.45,-10.375) {{\bf -}};

		\node at (5,-12.75) {$\longrightarrow$ $(\vec{H}_{1} \; \vec{H}_{2} \; \vec{H}_{3} \; \vec{H}_{4})\transpose{(\vec{e}_{1},\vec{e}_{2},\vec{e}_{3},\vec{e}_{4})}=$ };
		\draw (8,-13.5) rectangle (8.2,-12);
		\node at (8.1,-12.75) {$0$};

		\end{tikzpicture}
	\end{center} 
	\caption{Simplified version of Wagner's algorithm with four layers to find $\vec{e}$ of weight $t$ such that $\vec{H}\vec{e} = \mathbf{0}$. The same colours on vectors means that they are equal (be careful on the minus signs). \label{fig:Wagner}}
	\end{figure}

	The output of Wagner's algorithm described in Figure \ref{fig:Wagner} is $(\vec{e}_{1},\vec{e}_{2},\vec{e}_{3},\vec{e}_{4})$. It is a solution as by construction it reaches the syndrome $\mathbf{0}$ with respect to $\vec{H}$ and it has the right Hamming weight since each $\vec{e}_{i}$ has weight $t/4$. Notice that this solution has a particular ``shape'' when compared to the output of Dumer's algorithm.
	During Steps $2$ and $3$ we do not perform collisions on all the $n-k$ symbols of the syndromes but on $(n-k)/2$ symbols. Therefore solutions have the following property: $\left(\vec{H}_{1} \; \vec{H}_{2}\right)\transpose{(\vec{e}_{1},\vec{e}_{2})}$ and $\left(\vec{H}_{3} \; \vec{H}_{4}\right)\transpose{(\vec{e}_{3},\vec{e}_{4})}$ are  equal to $0$ on the last $(n-k)/2$ positions. If one splits an output of Dumer's algorithm and the parity-check matrix in four parts, there is no reason that it verifies the above property. It will be true only for an exponentially small fraction of the solutions. It explains why Wagner's algorithm ``decimates'' the solutions. At the same time this algorithm has the advantage to be able to produce solutions in amortized time one. The idea in that case is to build the lists $\mathcal{L}_{i}$'s with size $L$ such $L^{2}/q^{(n-k)/2} = L$, {\em i.e.} $L = \sqrt{q^{n-k}}$. Therefore, each collision step has the same cost given by the size $L$ of the lists that are output. However it may happen that the number of solutions is still too large. If so, the next idea of Wagner's algorithm is to consider more lists at the beginning, like $8$ (and not $4$) and then to make collisions on $(n-k)/3$ symbols. It enables smaller lists, namely  $L^{2}/q^{(n-k)/3} = L$, {\em i.e.} $L = \sqrt[3]{q^{n-k}}$. However, in that case, there will be three steps of collisions. Then we can extend this by considering a number of lists given by some $2^{a}$ and making $a$ steps of collisions. Nonetheless, it is not possible to do this for any $a$. If one uses Wagner's algorithm with initially $2^{a}$ lists, all of them need to be built from vectors $\vec{e}_{i}$ of Hamming weight $t/2^{a}$. If $a$ is too large it will be impossible to build lists large enough to produce collisions in amortized time one.

	Let us emphasize that the above discussion is not only a thought exercise. It turns out that the above situation happens with ISD algorithms (Dumer's algorithm produces too many solutions in one iteration). It explains why the ISD with Wagner's algorithm outperforms the ISD with Dumer's algorithm for some parameters as we can see in Figure \ref{figure:PrangeDumerWagner} (in particular when $\DP$ is such that there are a lot ofh solutions).

	\begin{proposition}\label{propo:Wagner}

		Wagner's algorithm solves $\DP(n,q,R,\tau)$ by (where $k \eqdef \lfloor Rn \rfloor$) 
		\begin{enumerate}
			\item finding one solution in time and space
			$
			q^{\frac{n-k}{a+1}}
			$ (up to a polynomial factor in $n$)
			for any integer $a$ such that $q^{\frac{n-k}{a+1}} \leq \binom{n/2^{a}}{t/2^{a}}(q-1)^{t/2^{a}}$ which asymptotically can be written as,
			\begin{equation*}\label{cons:1} 
			\frac{1-R}{h_q(\tau)} \leq \frac{a+1}{2^{a}}.
			\end{equation*} 
			\item finding $ q^{\frac{n-k}{a}}$ solutions in amortized time $1$ (up to a polynomial factor in $n$) for any integer $a$ such that $q^{\frac{n-k}{a}} \leq \binom{n/2^{a}}{t/2^{a}}(q-1)^{t/2^{a}}$ which asymptotically can be written as,
			\begin{equation*}\label{cons:2}
				\frac{1-R}{h_q(\tau)} \leq \frac{a}{2^{a}}.
			\end{equation*}
		\end{enumerate}
		
	\end{proposition}

	During the description of  Wagner's algorithm that follows  (which will give the proof of the above proposition) we will use Lemma \ref{lemm:BP} to estimate the size of the lists after merging.
	\newline

	{\bf \noindent Wagner's algorithm.} The first step is to split $\vec{H}$ in $2^{a}$ parts of the same size, for a parameter $a$ that is called {\em depth of the algorithm}. For the sake of simplicity let us split $\vec{H}$ as (we can choose the partition)
	$$
	\vec{H} = \begin{pmatrix}
		\vec{H}_{1} & \dots & \vec{H}_{2^{a}} 
	\end{pmatrix}
	$$
	where for all $i$ we have $\vec{H}_{i}\in\Fq^{(n-k)\times \frac{n}{2^{a}}}$. Then we build the following $2^{a}$-lists for some parameter $L$ that will be fixed later ($t = \lfloor \tau n \rfloor$)
	$$
	\forall i \in \llbracket 1,2^{a} \rrbracket, \quad \mathcal{L}_{i}\subseteq \left\{ \vec{e}\transpose{\vec{H}}_{i} \mbox{ : } \vec{e} \in \Fq^{n/2^{a}}, \; |\vec{e}| = \frac{t}{2^{a}} \right\} \quad \mbox{{\em and}} \quad \sharp \; \mathcal{L}_{i} = L. 
	$$
	Notice that by construction we have the following constraint,
	\begin{equation}\label{eq:consWagL}  
	L \leq \binom{n/2^{a}}{t/2^{a}} (q-1)^{t/2^{a}}.
	\end{equation} 
Let $\ell \in \llbracket 1,n-k \rrbracket$ be a parameter that will be chosen later. Then, Wagner's algorithm performs the collision of these lists two by two on their last $\ell$ symbols \footnote{Given a vector $\vec{x}\in\Fq^{m}$, it denotes $x_{m-\ell+1},\dots,x_{m}$.} to build the new lists $\mathcal{L}_{i,i+1}$'s, namely
	$$
	\mathcal{L}_{i,i+1} \eqdef \left\{ \vec{s}_{i} + \vec{s}_{i+1} \mbox{ : } \vec{s}_{i} \in\mathcal{L}_{i} \mbox{ and the last $\ell$ symbols of $\vec{s}_{i} + \vec{s}_{i+1}$ are $\mathbf{0}$} \right\},
	$$
	where by construction we have access to the errors $\vec{e}_{i}$ and $\vec{e}_{i+1}$ of Hamming weight $t/2^{a}$ that reach $\vec{s}_{i}$ and $\vec{s}_{i+1}$ through  $\vec{H}_{i}$ and $\vec{H}_{i+1}$. The last list $\mathcal{L}_{2^{a}-1,2^{a}}$ is built by merging $\mathcal{L}_{2^{a}}$ and $\mathcal{L}_{2^{a}-1}$ but this time according to the last $\ell$ symbols of $\vec{s}$, namely
	$$
		\mathcal{L}_{2^{a}-1,2^{a}} \eqdef \{ \vec{s}_{2^{a}-1} + \vec{s}_{2^{a}} \mbox{ : } \vec{s}_{i} \in
		\mathcal{L}_{i} \mbox{ and the last $\ell$ symbols of $\vec{s}_{2^{a}-1} + \vec{s}_{2^{a}}$ are equal to those of $\vec{s}$}  \}.
	$$
	Using Lemma \ref{lemm:BP}, these new lists built after merging on $\ell$ symbols will be of the same size, 
	$$
	\frac{L^{2}}{q^{\ell}}.
	$$
	Furthermore, we produce them at cost $L + \frac{L^{2}}{q^{\ell}}$ (up to a polynomial factor). Once this is done, we start again this process $a-2$ times by merging each time on the $\ell$ next new symbols. Wagner proposed to choose $L$ such that at each step, {\em the time for merging is the same than the one to build lists}, namely
	\begin{equation}\label{eq:tailleLWag}
		L = q^{\ell}.
	\end{equation}  
	This implies under Constraint \eqref{eq:consWagL} that the parameter $\ell$ is such that
	\begin{equation}\label{eq:consWagL2}
		q^{\ell} \leq \binom{n/2^{a}}{t/2^{a}}(q-1)^{t/2^{a}}.
	\end{equation} 

	\begin{remark} 
	One may ask if this strategy of an amortized time one at each merge is optimal. It turns out that the answer is yes as proved in \cite{MS09}. 
	\end{remark}

	Up to now we have made $a-1$ merges and we still have two lists, that we denote by $\mathcal{S}_{1}$ and $\mathcal{S}_{2}$. They are such that
	$$
		\mathcal{S}_{1} = \left\{ \vec{s}_{1} + \cdots + \vec{s}_{2^{a-1}} \mbox{ : } \vec{s}_{i} \in \mathcal{L}_i \mbox{ and the last $(a-1)\ell$ symbols of $\vec{s}_1+ \cdots + \vec{s}_{2^{a-1}}$  are equal to  $\mathbf{0}$}  \right\}
	$$
	$$
		\mathcal{S}_2 = \left\{ \vec{s}_{2^{a-1}+1}+ \cdots + \vec{s}_{2^{a}} \mbox{ : } \vec{s}_i \in \mathcal{L}_i \mbox{ and the last $(a-1)\ell$ symbols of $\vec{s}_1+ \cdots + \vec{s}_{2^{a-1}}$ are equal to those of  $\vec{s}$}  \right\}
	$$
$$
	\mbox{where,} \qquad \sharp \; \mathcal{S}_1  = \sharp \; \mathcal{S}_2 = \frac{L^{2}}{q^{\ell}} = L =  q^{\ell}. 
	$$ 
	Therefore, it remains to merge these two lists on the last $(n-k) - (a-1)\ell$ first symbols. It yields in time $\frac{q^{\ell(a+1)}}{q^{ (n-k)} }$ a list of solutions of size
	\begin{equation} \label{eq:tailleMoySolWag} 
		\frac{q^{2\ell}}{q^{ (n-k) - (a-1)\ell } } = \frac{q^{\ell(a+1)}}{q^{ (n-k)} }.
	\end{equation}
	Now the parameter $\ell$ has to be set whether one wants only one solution or many solutions in amortized time one.
	\newline

	{\bf \noindent Wagner to reach one solution.} According to Equation \eqref{eq:tailleMoySolWag}, it remains to choose parameters such that
	$$
	\ell = \frac{n-k}{a+1}.
	$$
	All the lists in the $a-1$ first steps of the algorithm have the same size, namely $L = q^{\ell}$ (Equation \eqref{eq:tailleLWag}), therefore the algorithm has a cost given by
	$$
	q^{\frac{n-k}{a+1}}.
	$$
	However, we have to be careful with the depth $a$ of the algorithm, unfortunately it cannot be chosen too large. According to Equation \eqref{eq:consWagL} 
	$$
	q^{\ell} = q^{\frac{n-k}{a+1}} \leq \binom{n/2^{a}}{t/2^{a}}(q-1)^{t/2^{a}}
	$$
	which leads to the following asymptotic constraint,
	$$
	\frac{1-R}{a+1} \leq \frac{1}{2^{a}} \; h_{q}(\tau) \iff \frac{1-R}{h_{q}(\tau)} \leq \frac{a+1}{2^{a}}.
	$$
	It concludes the proof of $(1)$.
	\newline

	{\bf \noindent Wagner to compute many solutions in amortized time one.} In this case, according to Equation \eqref{eq:tailleMoySolWag}, we just need to choose $\ell$ such that
	$$
	 q^{\ell} = \frac{q^{\ell(a+1)}}{q^{n-k}} \iff \ell = \frac{n-k}{a}
	$$
	As above we obtain the claimed constraint on $a$. It concludes the proof. \qed
	\newline

	We draw in Figure \ref{figure:Wagner} the exponent of Wagner's algorithm to solve $\DP(n,q,R,\tau)$ as function of $\tau \geq \tau^{-}$ (the relative Gilbert-Varshamov distance defined \iftoggle{amsbook}{in Equation \eqref{eq:tau-+}}{in lecture notes $2$}). We choose parameters of the algorithm to output one solution and $a$ being the largest integer that satisfies the constraint \eqref{cons:1} of Proposition \ref{propo:Wagner} (to have an optimal complexity). As it can be seen the exponent is a decreasing function of $\tau$. Indeed, when $\tau$ increases, $a$ can be chosen larger.
	Furthermore, there is a discontinuity in the exponent. It comes from the fact that $a$ is an {\em integer}. It is possible to adapt the algorithm to ``smooth'' its complexity but this is out of scope of these lecture notes. 
	
	\begin{center}
		\begin{figure}[h]
			\includegraphics[height=4.5cm]{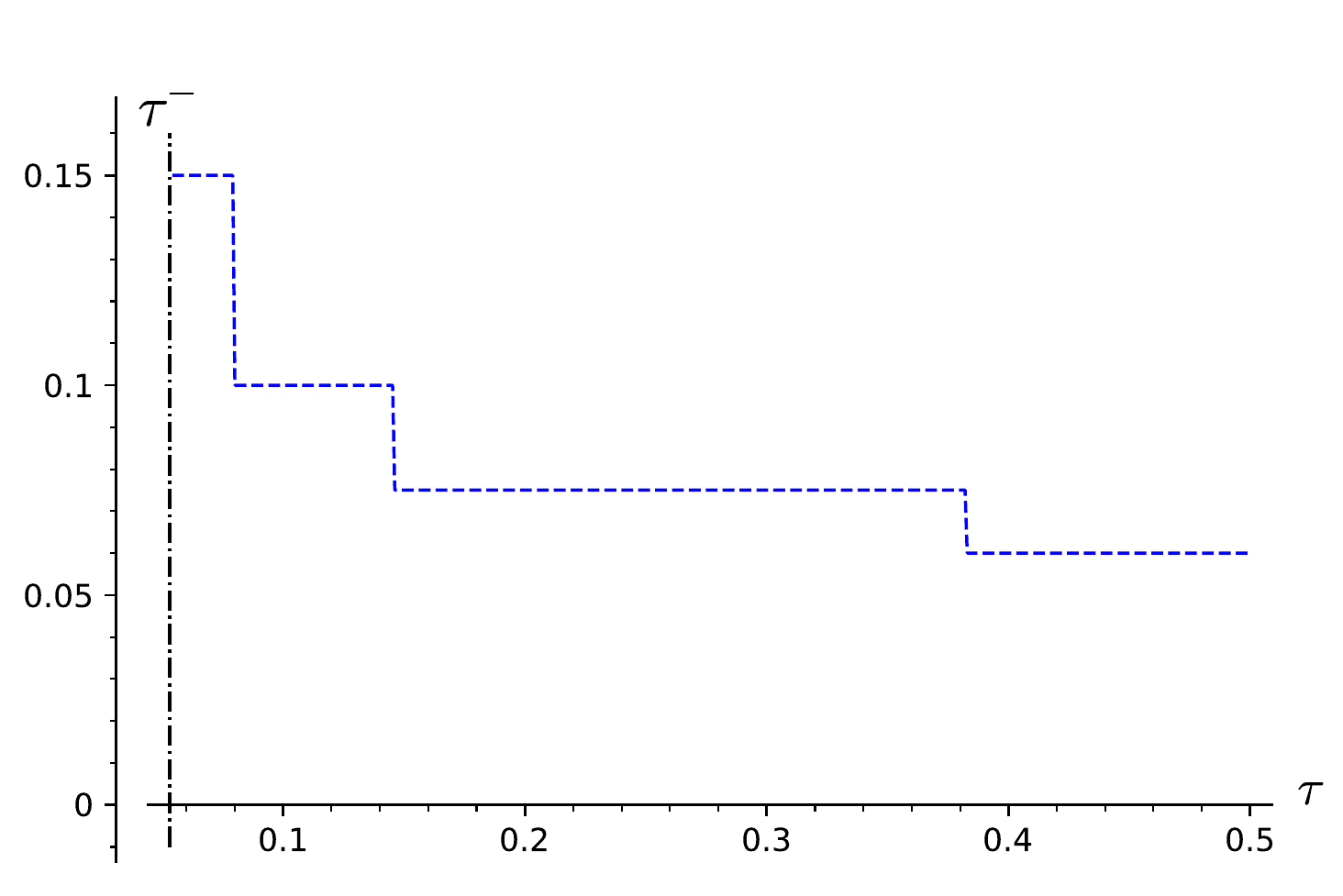}
			\caption{Exponent of Wagner's algorithm to solve $\DP(n,q,R,\tau)$ for $R = 0.7$ as function of $\tau$.
				\label{figure:Wagner}
			}
		\end{figure}
	\end{center}

	\section{Combining Linear Algebra and Birthday Paradox Techniques}

	We are now ready to present the general framework (introduced in \cite{FS09}) of {\em Information Set Decoding} (ISD) algorithms. 
	\newline

	{\bf \noindent The algorithm.} Let us introduce the following parameters,
	$$
	\ell \in \llbracket 0,n-k \rrbracket \quad \mbox{and} \quad p \in \llbracket 0,\min(t,k+\ell) \rrbracket 
	$$
	\begin{itemize}
		\item[1.] {\em Picking the augmented information set.} Let $\mathcal{J} \subseteq \llbracket 1,n \rrbracket$ be a random set of size $k+\ell$. If $\vec{H}_{\overline{\mathcal{J}}}~\in~\Fq^{(n-k)\times (n-k)}$ is not of full-rank, pick another set $\mathcal{J}$.

		\item[2.] {\em Linear algebra.} Perform a Gaussian elimination to compute a non-singular matrix $\vec{S} \in\Fq^{(n-k) \times (n-k)}$ such that $\vec{S} \vec{H}_{\overline{\mathcal{I}}} = \begin{pmatrix} 
			\mathbf{1}_{n-k-\ell} \\
			\mathbf{0}_{\ell \times (n-k-\ell)}
		\end{pmatrix}$.	
		Let $\vec{H}' \in \Fq^{(n-k-\ell) \times (k+\ell)}$, $\vec{H}'' \in \Fq^{\ell \times (k+\ell)}$, $\vec{s}'\in \Fq^{n-k-\ell}$ and $\vec{s}'' \in \Fq^{\ell}$ be such that 
		\begin{equation}\label{eq:H'H''s's''}
		\vec{S}\vec{H}_{\mathcal{J}} = \begin{pmatrix}
			\vec{H}' \\
			\vec{H}''
		\end{pmatrix} \quad \mbox{and} \quad \vec{S}\transpose{\vec{s}} = \transpose{(\vec{s}',\vec{s}'')}
		\end{equation}

		\item[3.] {\em Sub-decoding problem.} Compute a set,
		\begin{equation}\label{eq:sS} 
		\mathcal{S} \subseteq \left\{ \vec{e}''\in\Fq^{k+\ell} \; : \; \vec{e}''\transpose{\vec{H}''} = \vec{s}'' \; \mbox{and} \; \left| \vec{e}''\right| = p\right\}. 
		\end{equation}

		\item[4.] {\em Test.} Find $\vec{e}''\in\mathcal{S}$ such that $\left| \vec{s}' - \vec{e}'\transpose{\vec{H}''} \right| = t-p$. If not, return to Step $1$; otherwise output $\vec{e}\in\Fq^{n}$ such that
		\begin{equation}\label{eq:defSol}  
		\vec{e}_{\overline{\mathcal{J}}} = \vec{s}' - \vec{e}''\transpose{\vec{H}'} \quad ; \quad \vec{e}_{\mathcal{J}} = \vec{e}''
		\end{equation} 
	\end{itemize}

	{\bf \noindent Correction of the algorithm.} It easily follows from the following computation,
	\begin{align*}
		\vec{S}\vec{H}\transpose{\vec{e}} &= \vec{S}\vec{H}_{\overline{\mathcal{J}}}\transpose{\vec{e}}_{\overline{\mathcal{J}}} + \vec{S}\vec{H}_{\mathcal{J}}\transpose{\vec{e}}_{\mathcal{J}} \\
		&= \begin{pmatrix} 
			\mathbf{1}_{n-k-\ell} \\
			\mathbf{0}_{\ell\times (n-k-\ell)}
		\end{pmatrix}  \left( \transpose{\vec{s}'} - \vec{H}'' \transpose{\vec{e}''} \right)+ \begin{pmatrix}
		\vec{H}' \\
		\vec{H}''
	\end{pmatrix} \transpose{\vec{e}''} \quad \mbox{(By definition of $\vec{S}\vec{H}_{\overline{\mathcal{J}}}$, $\vec{S}\vec{H}_{\mathcal{J}}$, $\vec{e}_{\overline{\mathcal{J}}}$ and $\vec{e}_{\mathcal{J}}$)} 
\\
		&= \begin{pmatrix}  \transpose{\vec{s}'} - \vec{H}''\transpose{\vec{e}''} \\
			\mathbf{0}_{\ell\times (n-k-\ell)} 
		\end{pmatrix} + \begin{pmatrix}
		\vec{H}'\transpose{\vec{e}''} \\
		\vec{H}''\transpose{\vec{e}''}
	\end{pmatrix} \\
	&= \begin{pmatrix}
		\transpose{\vec{s}'} \\
		\transpose{\vec{s}''}
	\end{pmatrix} \qquad \mbox{(By definition of $\vec{e}''\in\mathcal{S}$, see Equation \eqref{eq:sS})} \\
	&= \vec{S} \transpose{\vec{s}} \qquad \mbox{(By Equation \eqref{eq:H'H''s's''})}
	\end{align*}
	which corresponds to $\vec{H}\transpose{\vec{e}} = \transpose{\vec{s}}$ since $\vec{S}$ is non-singular. Furthermore, by definition $\vec{e}''$ has Hamming weight $p$ and the test ensures that $\vec{s}' - \vec{e}''\transpose{\vec{H}''}$ has weight $t-p$. Therefore, once the algorithm terminates, $\vec{e}$ reaches $\vec{s}$ with respect to $\vec{H}$ and has Hamming weight $t$.

	\begin{exercise}
		Let
		$$
		\mathcal{D} \eqdef \left\{ \vec{c}''\in\Fq^{k+\ell} \; : \; \vec{c}''\transpose{\vec{H}''} = \mathbf{0} \right\}.
		$$
		Show that $\mathcal{D}$ is a code of length $k+\ell$ and dimension $k$. 
	\end{exercise}

	\begin{remark}
		The code of parity-check matrix $\vec{H}''$ is known as the punctured code (defined by the parity-check matrix $\vec{H}$) at the positions $\overline{\mathcal{J}}$. Computing $\mathcal{S}$ in Equation \eqref{eq:sS} amounts to solve a decoding problem at distance $p$ with this input code and the syndrome $\vec{s}''$.  Therefore, for each drawn of the augmented information set $\mathcal{J}$ we test many decoding candidates (given by elements of the list $\mathcal{S}$ and with associated lift defined in Equation \eqref{eq:defSol}). We recover the interpretation of \textup{ISD} algorithms given in the introduction. 
	\end{remark}

	{\bf \noindent Far or close codeword?} One may wonder why don't we use the distribution $\mathcal{D}_{t}$ in ISD algorithms to be able to produce ``short'' or ``large'' solutions? To answer this question let us take a look at the typical weight of a vector $\vec{e}$ that will pass the test at the end of an iteration (see Equation \eqref{eq:defSol}). By supposing that $\vec{s}$ is uniformly distributed, we have
	\begin{equation}\label{eq:typicalWeight}  
	\mathbb{E}\left( \left| \vec{e} \right| \right) = p + \frac{q-1}{q}\; \left( n-k - \ell \right)
	\end{equation} 
	The $\frac{q-1}{q}\; \left( n-k - \ell \right)$ term comes from the fact that $\vec{s}'$ is uniformly distributed over $\Fq^{n-k-\ell}$ while $p$ is here as by definition $\vec{e}''$ has weight $p$. If one wants to get a solution of small weight, the best approach is to decode the punctured code at a small as possible distance $p$. On the other hand, if one seeks a solution of large weight, one has to decode this punctured code at the largest as possible distance, namely $p = k+\ell$. Therefore the strategy to reach short or large error relies on how we choose the parameter $p$.

	The above discussion hints us why we can not reasonably hope, with ISD algorithms, to solve $\DP$ in polynomial time outside the interval $\llbracket \frac{q-1}{q}\; (n-k), k + \frac{q-1}{q} \; (n-k) \rrbracket$.  For instance, if one is looking for an ISD algorithm solving $\DP$ in polynomial time for some $t <\frac{q-1}{q}(n-k)$, one has according to Equation \eqref{eq:typicalWeight}  to find a subroutine decoding in polynomial time a random code of length $k+\ell$ and dimension $k$ at distance $p$ such that
	$$
	p - \frac{q-1}{q} \ell < 0.
	$$
	But at the same time, the smaller $p$ for which we known how to decode in polynomial a random $\lbrack k+\ell, k\rbrack_{q}$-code is precisely $\frac{q-1}{q} \; (k+\ell -k) = \frac{q-1}{q} \; \ell$ which is the above limit to get an improvement. Therefore, if one seeks an ISD enlarging the interval of weights in which Prange's algorithm is polynomial, one has to first enlarge this interval.  
	\newline

	{\bf \noindent Analyse of the algorithm.} As in Prange's algorithm, all the challenge in the analysis of ISD algorithms running time relies on the computation of the success probability in Step $4$. However, contrary to Prange's algorithm it will not be necessary to make an assumption on how the augmented information sets are picked. We can suppose directly, when $\ell = \Theta(n)$ (which will be the case in our applications), that they are uniformly distributed as shown by the following proposition.

	\begin{proposition} 
		Let $\vec{H} \in \Fq^{(n-k)\times n}$ and $\mathcal{J} \subseteq \llbracket 1, n \rrbracket$ being uniformly distributed over the sets of size $k+\ell$ (where $\ell > 0$) such that $\vec{H}_{\overline{\mathcal{J}}}$ is non-singular. Let $\mathcal{J}_{\textup{unif}} \subseteq \llbracket 1, n \rrbracket$ being uniformly distributed over the sets of size $k+\ell$. We have,
		$$
		\mathbb{E}_{\vec{H}}\left( \Delta\left( \mathcal{J},\mathcal{J}_{\textup{unif}} \right) \right) = O\left( \frac{1}{q^{\ell}} \right)
		$$
		where $\Delta$ denotes the statistical distance. 
	\end{proposition}

	\begin{proof}
		Let us index from $1$ to $\binom{n}{k+\ell}$ the subset of size $k+\ell$ of $\llbracket 1,n \rrbracket$ and let $X_{i}$ be the indicator of the event ``the subset $\mathcal{J}_{i}$ of index $i$ is such that $\vec{H}_{\overline{\mathcal{J}_{i}}}$ has not a full rank''. Let,
		$$
		N \eqdef \sum_{i} X_{i}
		$$ 
		It can be verified that we have
			\begin{equation}\label{eq:statDist}
			\mathbb{E}_{\vec{H}}\left( \Delta\left( \mathcal{J},\mathcal{J}_{\textup{unif}} \right) \right) = \mathbb{E}_{\vec{H}} \left(\frac{N}{\binom{n}{k+\ell}} \right) = \frac{1}{\binom{n}{k+\ell}}\; \sum_{i=1}^{\binom{n}{k+\ell}} \mathbb{E}_{\vec{H}}(N_{i}) 
		\end{equation} 
		where the last equality follows from the linearity of the expectation.

		Notice now that $\vec{H}_{\overline{\mathcal{J}}}\in\Fq^{(n-k) \times (n-k - \ell)}$ has not a full rank with probability (over $\vec{H}$) given by a $O\left( \frac{1}{q^{\ell}} \right)$. Therefore,
		$$
		\mathbb{P}\left( X_{i} = 1 \right) = O\left( \frac{1}{q^{\ell}} \right)
		$$
		Plugging this in Equation \eqref{eq:statDist} concludes the proof.
\end{proof}

	However, although the above proposition enables to avoid an assumption, there will be as for Prange, an assumption to make when studying ISD algorithms. 
	\newline

	{\bf \noindent An important quantity.} Let us use notation of the above algorithm. Let $\alpha_{p,\ell}$ be the probability that given a fixed $\vec{x}\in \Fq^{k+\ell}$ be such that  $\left\{
	\begin{array}{ll}
		\vec{x}\transpose{\vec{H}''} = \vec{s}'' & \\
		|\vec{x}| = p	 &
	\end{array}
	\right.$, the vector $\vec{e}' \eqdef \vec{s}' - \vec{x}\transpose{\vec{H}''}$ has Hamming weight $t-p$, namely
	$$
	\alpha_{p,\ell} \eqdef \mathbb{P}\left( \left| \vec{e}' \right| = t-p \right). 
	$$ 
	In other words, $\alpha_{p,\ell}$ denotes the probability that given a solution $\vec{x}$ of the decoding problem at distance $p$ with input $(\vec{H}'',\vec{s}'')$, then its lift gives a solution of weight $t$ of the initial decoding problem with input $(\vec{H},\vec{s})$. Notice that we did not suppose that $\vec{x} \in \mathcal{S}$.

	\begin{proposition}\label{propo:probaISD}
		The probability $\alpha_{p,\ell}$ is up to a polynomial factor (in $n$) given by, 
		$$
		 \frac{\binom{n-k-\ell}{t-p}(q-1)^{t-p}} {\min
			\left( q^{n-k-\ell},\binom{n}{t}(q-1)^{t}q^{-\ell} \right)}
		$$
	\end{proposition}

	The proof of this proposition is similar to the one of Proposition \ref{propo:probaPrange} and here we only provide a sketch of it.

	\begin{proof}[Sketch of proof.]
		 One can remark that the only difference between formulas of Propositions \ref{propo:probaPrange} and \ref{propo:ISD} is the factor $q^{-\ell}$ in the denominator. The difference comes from the fact that the probability (over $\vec{H}$) that an error $(\vec{e}',\vec{e}'')$ of weight $t$ verifies  $\vec{e}''\transpose{\vec{H}''} = \vec{s}''$, where $\vec{s}'' \in \Fq^{\ell}$, is $q^{-\ell}$. Therefore we have to consider a fraction $q^{-\ell}$ of possible solutions in our probability, which roughly explains the factor $q^{-\ell}$ in the denominator.
	\end{proof}

	We are now ready to give the running time of ISD algorithms to solve $\DP$. It will use the following assumption
	\begin{assumption}\label{ass:ISD} Let us use notation of \textup{ISD} algorithm that is described above. 
		Given $\vec{x}^{(1)},\dots,\vec{x}^{(S)}$ be solution of $\left\{ \begin{array}{l}
			\vec{x}\transpose{\vec{H}''} = \vec{s}''\\
			|\vec{x}| = p
		\end{array}\right.$ Then, the vectors $\vec{e}^{(i)}$'s  in $\Fq^{n}$ be defined as 
	$$
		\vec{e}^{(i)}_{\overline{\mathcal{J}}} = \vec{s}' - \vec{x}^{(i)}\transpose{\vec{H}'} \quad ; \quad \vec{e}^{(i)}_{\mathcal{J}} = \vec{x}^{(i)}
	$$
	are independent random variables (where $\mathcal{J}$ is a random augmented information set). 
	\end{assumption}

	\begin{proposition}\label{propo:ISD} Let $\ell \in \llbracket 0,n-k \rrbracket$. Let $\mathcal{A}$ be an algorithm that can compute $S$ solutions in time $T$ of the following problem	$\left\{ \begin{array}{l}
			\vec{x}\transpose{\vec{H}''} = \vec{s}''\\
			|\vec{x}| = p
		\end{array}\right.$ where $\vec{H}\in\Fq^{\ell \times (k+\ell)}$ and $\vec{s}''\in\Fq^{\ell}$. Furthermore, we suppose that outputs of $\mathcal{A}$ verify Assumption \ref{ass:ISD}. Then, the \textup{ISD} algorithm using $\mathcal{A}$ in Step $3$ solves $\DP(n,q,R,\tau)$ up to a polynomial factor (in $n$) in time 
	$$
	T \; \max \left( 1,\frac{1}{S \; \alpha_{p,\ell}} \right)
	$$
	where $\alpha_{p,\ell}$ is given in Proposition \ref{propo:probaISD}. 	
	\end{proposition}

	As in Proposition \ref{propo:probaISD} we only provide a sketch of proof of this proposition. 
	
	\begin{proof}[Sketch of proof.] The number of iterations of ISD algorithms is, as for Prange's algorithm, up to a polynomial factor (in $n$) given by $1/p_{\textup{ISD}}$ where $p_{\textup{ISD}}$ is the probability of success of an iteration. Furthermore, each iteration has a cost given by the time to computing $\mathcal{S}$ in Step $3$ (the cost of a Gaussian elimination is polynomial). Therefore the cost of the ISD using $\mathcal{A}$ is given by $T \; \frac{1}{p_{\textup{ISD}}}$.

		Let us compute $p_{\textup{ISD}}$. Let $\vec{e}''_{1},\dots,\vec{e}''_{S}$ be the outputs of $\mathcal{A}$. The probability that any $\vec{e}''_{i}$ does not lead to a solution is given by $1-\alpha_{p,\ell}$. Using the independence given by Assumption \ref{ass:ISD}, the probability that none of the $\vec{e}_{i}''$'s leads to a solution is given by
		$$
		1 - (1-\alpha_{p,\ell})^{S} = 1 - \Theta\left( \min(1,S\alpha_{p,\ell}) \right).
		$$
		Therefore, $p_{\textup{ISD}} = \Theta\left( \min(1,S\alpha_{p,\ell}) \right)$ and
		$$
		T \; \frac{1}{p_{\textup{ISD}}} = T \; \frac{1}{\Theta\left( \min(1,S\alpha_{p,\ell}) \right)} = \Theta\left(  T\; \max\left( 1,\frac{1}{S \; \alpha_{p,\ell}} \right) \right)
		$$
		which concludes the proof. 
\end{proof}
	
	We are now ready to ``instantiate'' ISD algorithms with Dumer and Wagner algorithms that we have described in Subsections \ref{subsec:Dumer} and \ref{subsec:Wagner}.
	\subsection{ISD with Dumer's algorithm}

	A slight variation of Proposition \ref{propo:cpxDum} shows that, given an instance $(\vec{H}'',\vec{s}'') \in \Fq^{\ell \times (k+\ell)} \times \Fq^{\ell}$ of a decoding problem at distance $p$, Dumer's algorithm find 
	$$
		\frac{\binom{k+\ell}{p}(q-1)^{p}}{q^{\ell}}
	$$
	solutions in average time 
	$$
	\sqrt{\binom{k+\ell}{p}(q-1)^{p}} + \frac{\binom{k+\ell}{p}(q-1)^{p}}{q^{\ell}}.
	$$
	Here there is no maximum in the fomula as we are not sure that there is always a solution to our decoding problem.

	Therefore we easily deduce the following proposition which gives the complexity of the ISD using Dumer's algorithm.

	\begin{proposition} The complexity $C_{\textup{Dumer}}(n,q,R,\tau)$ of the \textup{ISD} using Dumer's algorithm (described in Subsection \ref{subsec:Dumer}) to solve $\DP(n,q,R,\tau)$ is up to a polynomial factor (in $n$) given by
			\begin{equation}\label{eq:DumCpxISD}
			 \left( \sqrt{\binom{k+\ell}{p}(q-1)^{p}} + \frac{\binom{k+\ell}{p}(q-1)^{p}}{q^{\ell}} \right) \cdot \max \left(1,\frac{\min\left(q^{n-k},\binom{n}{t}(q-1)^{t}\right)}{\binom{n-k-\ell}{t-p}(q-1)^{t-p}\; \binom{k+\ell}{p}(q-1)^{p}} \right)   
		\end{equation}
	\end{proposition}

	This complexity is parametrized by $p$ and $\ell$. According to our wish, finding a short or large solution, the optimization will not be the same. Let us describe our strategy for both of them but before let us fix the relative quantities that we will consider
	$$
		R \eqdef \frac{k}{n}, \quad \tau \eqdef \frac{w}{n}, \quad \lambda \eqdef \frac{\ell}{n} \quad \mbox{and} \quad \pi \eqdef \frac{p}{n}.
	$$
	These quantities will be useful as we are interested in the {\em asymptotic complexity} of the ISD's. 
	\newline

	{\bf \noindent Strategy to reach short solutions.} Our first choice is to force Dumer's algorithm to produce decoding solutions in amortized time one. Let us stress that here we give a method to optimize the complexity of ISD's, but we do not claim that it will lead to optimal parameters. Anyway, Dumer's algorithm computes solutions in amortized time one if
	\begin{equation*}
		\sqrt{\binom{k+\ell}{p}(q-1)^{p}} = \frac{\binom{k+\ell}{p}}{q^{\ell}} \iff q^{\ell} = \sqrt{\binom{k+\ell}{p}(q-1)^{p}}
	\end{equation*}
	Using Equation \eqref{eq:asymptEnt}, it implies asymptotically the following equality
	\begin{equation} \label{eq:Dumconstraint}
		\lambda = \frac{R+\lambda}{2}h_{q}\left( \frac{\pi}{R+\lambda} \right) \iff \pi = (R+\lambda) h_{q}^{-1} \left( \frac{2\lambda}{R+\lambda} \right)
	\end{equation} 
	Let $\pi(\lambda)$ be the parameter $\pi$ that reaches the above equality. We can now verify that according to Equations \eqref{eq:asymptEnt}, \eqref{eq:DumCpxISD} and \eqref{eq:Dumconstraint} that
		$$
	\frac{1}{n} \; \log_q(C_{\textup{Dumer}}) = f(\lambda)(1+o(1))
	$$
	where
	\begin{equation*}
		f(\lambda) \eqdef \lambda + \max \left( 0, \min\left( 1 - R, h_{q}(\tau)\right)
		- (1-R-\lambda)h_{q}\left( \frac{\tau - \pi(\lambda)}{1 - R - \lambda} \right) - 2\lambda\right) . 
	\end{equation*}
	To optimize $\lambda \mapsto f(\lambda)$, a good approximation (which can be verified for many parameters) is to suppose that it is an unimodal function. Then its minimization is easy to obtain with for instance the golden section search (see \url{https://en.wikipedia.org/wiki/Golden-section_search}). We used this method to draw the exponent (for relative weights $\tau \leq (q-1)/q (1-R)$) of the ISD with Dumer's algorithm given in Figures \ref{figure:PrangeDumerWagner}, \ref{figure:PRvsDum} and \ref{figure:PRvsDumZoom} . Furthermore we multiplied the above formula by a term $\log_{2}(q)$ to get exponents in base $2$.  
	\newline

	{\bf \noindent Strategy to reach large solutions.} Let us suppose that $q > 2$. Otherwise we can symmetrize the complexity of the algorithm from the short case as shown in Exercise \ref{ex:sym}. Contrary to the strategy to get short solutions, if one wants to use an ISD to compute solutions with a large weight, one has to choose $p$ as $k+\ell$ (see the discussion in the beginning of this section entitled ``Far or close codeword''). Therefore, with Dumer's algorithm we will choose parameters such that
		\begin{equation*} 
		\lambda = \frac{R+\lambda}{2}h_{q}\left( \frac{\pi}{R+\lambda} \right) \; \; \mbox{and} \;\; \pi = R + \lambda
	\end{equation*} 
	which leads to (as $h_{q}(1) = \log_{q}(q-1)$),
		\begin{equation*}
			\lambda = \frac{R+\lambda}{2} \log_{q}(q-1) \iff \lambda = \frac{R}{2}  \; \frac{\log_{q}(q-1)}{1 - \frac{1}{2}\log_{q}(q-1)}
		\end{equation*}
	However if one uses this strategy directly with Dumer's algorithm it would lead to very high exponent as build lists of size $q^{\lambda n}$ would be too large.  The idea (before using Wagner's algorithm as we are going to do) is to change Dumer's algorithm and to use the variation given in Exercise \ref{ex:Dumer}. Suppose that one build lists of size $S$ in Dumer's algorithm. Then, according to Proposition \ref{propo:ISD}, the complexity of the ISD with this algorithm is given (up to polynomial factor by) (we fixed $p$ to $k+\ell$)
	$$
	 \left( S + \frac{S^{2}}{q^{\ell}} \right) \cdot \max \left(1,\frac{\min\left(q^{n-k},\binom{n}{t}(q-1)^{t}\right)}{\binom{n-k-\ell}{t-k-\ell}(q-1)^{t-k-\ell}\; S^{2}} \right)   
	$$
	Let $\sigma \eqdef \frac{1}{n} \log_{q} S$. Using this algorithm leads to the following asymptotic complexity
	\begin{equation}\label{eq:DumerLW}
		g(\lambda,\sigma) \eqdef \max \left( \sigma, 2\sigma - \lambda \right) + \max \left( 0, \min\left( 1 - R, h_{q}(\tau)\right)
		- (1-R-\lambda)h_{q}\left( \frac{\tau - R - \lambda}{1 - R - \lambda} \right) - 2\sigma\right)
	\end{equation}
	However we do not have to forget that we have a constraint on the size of built lists, namely $S \leq \binom{(k+\ell)/2}{p/2}(q-1)^{p/2}$, therefore $\sigma$ has necessarily to verify 
	\begin{equation}\label{cons:sigma} 
	\sigma \leq \frac{R+\lambda}{2} \; h_{q}\left( \frac{\pi}{R+\lambda} \right ). 
	\end{equation} 
	To optimize \eqref{eq:DumerLW} we used the golden section search to first finding $\sigma(\lambda)$ ``minimizing'' (according to the method) $\sigma \mapsto g(\lambda,\sigma)$ for a fixed $\lambda$ and $\sigma$ verifying Constraint \eqref{cons:sigma}. Then we also used the golden section search to ``minimize'' $\lambda \mapsto g(\lambda,\sigma(\lambda))$. We draw in Figures \ref{figure:PRvsDum} and \ref{figure:PRvsDumZoom} the exponent of Prange and the ISD with Dumer's algorithm for a fixed rate and as function of $\tau$. As we see Dumer's algorithm provides an improvement over Prange's algorithm. Even if the improvement seems slight, don't forget that it means an {\em exponential} improvement as we draw exponents.

	\begin{figure}[!htb]
		\minipage{0.5\textwidth}
		\includegraphics[width=\linewidth]{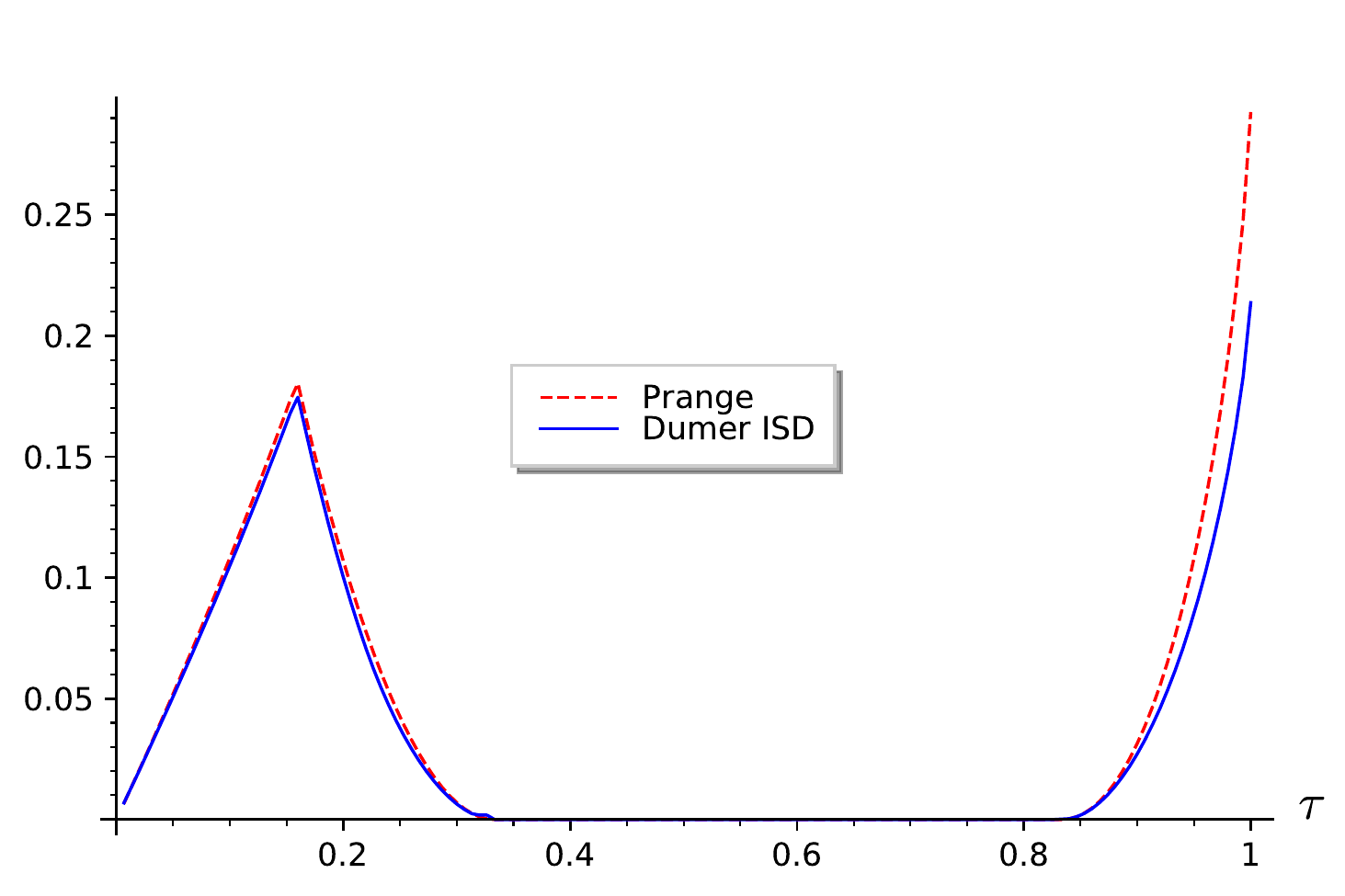}
		\caption{Exponents in base $2$ of Prange's algorithm and ISD with Dumer's algorithm (in base $2$) to solve $\DP(n,q,R,\tau)$ for $q = 3$ and $R = 1/2$ as function of $\tau\in [0,1]$. }\label{figure:PRvsDum}
		\endminipage\hfill
		\minipage{0.5\textwidth}
		\includegraphics[width=\linewidth]{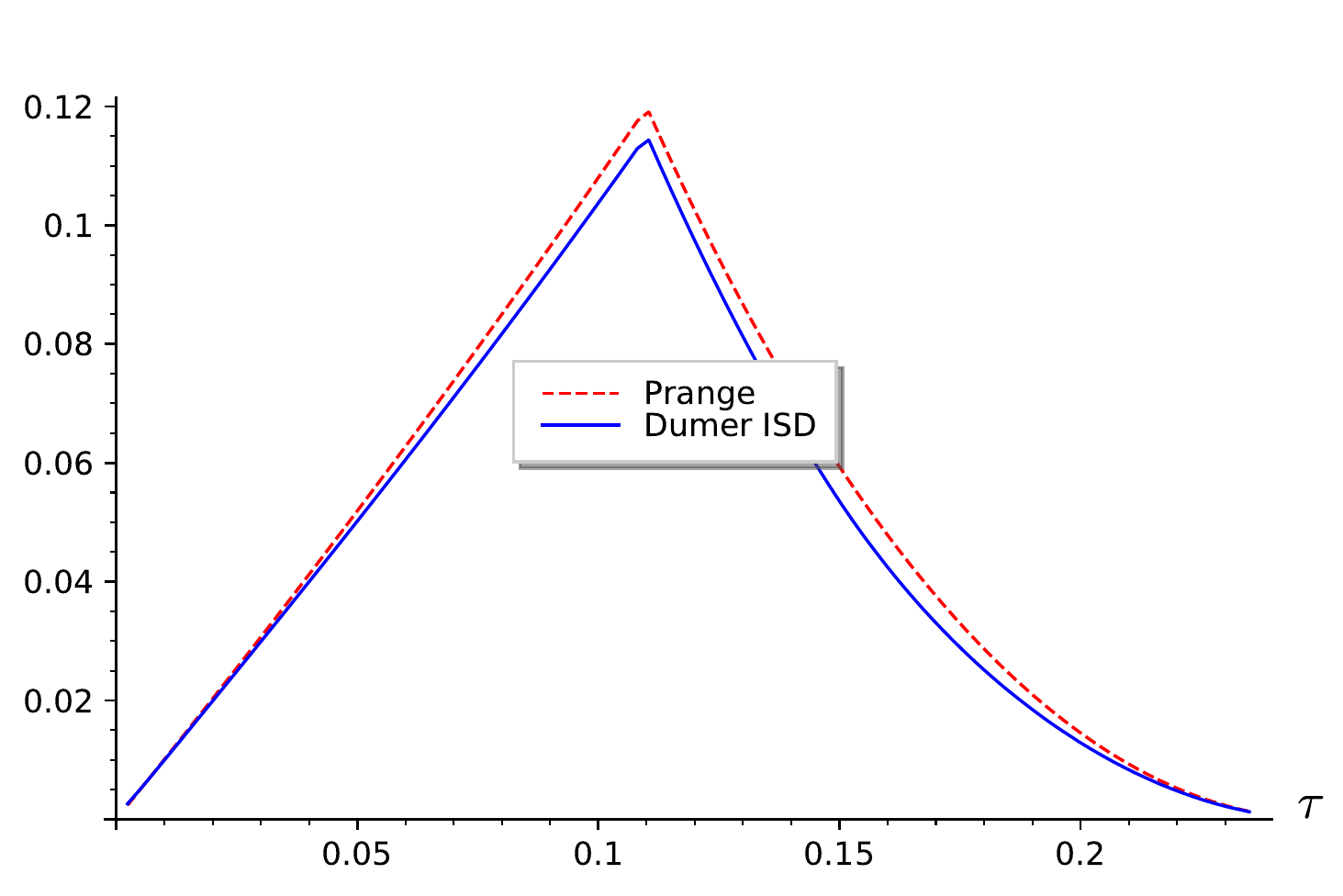}
		\caption{Exponents in base $2$ of Prange's algorithm and ISD with Dumer's algorithm (in base $2$) to solve $\DP(n,q,R,\tau)$ for $q = 2$ and $R = 1/2$ as function of $\tau\in \left[0,\frac{q-1}{q}(1-R)\right]$.}\label{figure:PRvsDumZoom}
		\endminipage
	\end{figure}

	\subsection{ISD with Wagner's algorithm} We are now ready to instantiate an ISD with Wagner's algorithm as a subroutine. In this case we choose to parametrize the algorithm to output solutions in amortized time one. 
	Combining Propositions \ref{propo:ISD} and \ref{propo:Wagner} (assertion $(2)$) leads to the following proposition

	\begin{proposition} The complexity $C_{\textup{Dumer}}(n,q,R,\tau)$ of the \textup{ISD} using Wagner's algorithm (described in Subsection \ref{subsec:Wagner} ) to solve $\DP(n,q,R,\tau)$ is up to a polynomial factor (in $n$) given by
		\begin{equation}\label{eq:WagnerCpxISD}
			q^{\frac{\ell}{a}} \; \max \left(1,\frac{\min\left(q^{n-k-\ell},\binom{n}{t}(q-1)^{t}q^{-\ell}\right)}{\binom{n-k-\ell}{t-p}(q-1)^{t-p}\; q^{\frac{\ell}{a}}} \right)   
		\end{equation}
	where $a$ is the largest integer such that $q^{\frac{\ell}{a}} \leq \binom{(k+\ell)/2^{a}}{p/2^{a}}(q-1)^{p/2^{a}}$.
	\end{proposition}

	We used this proposition (with the same kind of strategy that above) to draw the exponent of the ISD with Wagner's algorithm. As we can see in Figure \ref{figure:PrangeDumerWagner} the ISD with Wagner's algorithm has far better exponent compared to the ISD with Dumer's algorithm for large weight; otherwise exponents are the same.

 	\newpage
	\bibliographystyle{alpha}

	\newcommand{\etalchar}[1]{$^{#1}$}

\end{document}